\documentclass[11pt]{article}
\usepackage{amsmath,amsthm,amsfonts,amssymb,amscd}
\usepackage{epsfig,epstopdf}
\usepackage{graphicx}
\usepackage[svgnames]{xcolor}
\usepackage{tabularx,caption}
\usepackage{array}
\usepackage{booktabs}
\usepackage{xr}
\usepackage{caption}
\usepackage[normalem]{ulem}

\theoremstyle{plain}
\newtheorem{theorem}{Theorem}[section]
\newtheorem*{theorem*}{Theorem}
\newtheorem{lemma}{Lemma}[section]
\newtheorem{corollary}{Corollary}[section]
\newtheorem{proposition}{Proposition}[section]

\theoremstyle{definition}

\newtheorem*{definition*}{Definition}
\newtheorem{example}{Example}[section]

\newtheorem{non-example}{Non-example}[section]

\theoremstyle{remark}
\newtheorem{remark}{Remark}[section]
\newtheorem*{remark*}{Remark}

\newcommand{\MSD}{\textnormal{TAMSD}}
\newcommand{\bbR}{{\Bbb R}}
\newcommand{\bbN}{{\Bbb N}}
\newcommand{\bbC}{{\Bbb C}}

\newcommand{\bbP}{{\Bbb P}}

\newcommand{\bbE}{{\Bbb E}}

\newcommand{\cov}{{\mathrm{Cov}}}

\newcommand{\var}{{\mathrm{Var}}}
\newcommand{\Var}{{\mathrm{Var}}}

\newcommand{\diag}{{\mathrm{diag}}}

\newcommand{\ols}{{\mathrm{stand}}}

\newcommand{\AMSD}{\textnormal{AMSD}}

\newcommand{\abs}[1]{\left|#1\right|}
\newcommand{\norm}[1]{\left\|#1\right\|}

\newcommand*{\affaddr}[1]{#1} % No op here. Customize it for different styles.
\newcommand*{\affmark}[1][*]{\textsuperscript{#1}}

\begin{document}
\title{{Fluid heterogeneity detection based on the asymptotic distribution of the time-averaged mean squared displacement in single particle tracking experiments}
\author{%
Kui Zhang\affmark[1], Katelyn P. R. Crizer\affmark[3], Mark H. Schoenfisch\affmark[3], David B. Hill\affmark[2] and Gustavo Didier\affmark[1]\\
\affaddr{\affmark[1]Department of Mathematics, Tulane University}\\
\affaddr{\affmark[2]The Marsico Lung Institute and Department of Physics and Astronomy,\\University of North Carolina at Chapel Hill}\\
\affaddr{\affmark[3]Department of Chemistry, University of North Carolina at Chapel Hill}\\
}
\thanks{D.B.H.\ was partially supported by the awards DMS 1462992 (National Science Foundation), AI1 12029 and HL 108808 (National Institutes of Health), and Hill16XX0 (Cystic Fibrosis Foundation). G.D.\ was partially supported by the prime award no.\ W911NF-14-1-0475 from the Biomathematics subdivision of the Army Research Office, USA. The authors would like to thank John Fricks for his suggestions and comments on this paper.}
}
%\date{}
\maketitle

\begin{abstract}
A tracer particle is called anomalously diffusive if its mean squared displacement grows approximately as $\sigma^2 t^{\alpha}$ as a function of time $t$ for some constant $\sigma^2$, where the diffusion exponent satisfies $\alpha \neq 1$. In this article, we use recent results on the asymptotic distribution of the time-averaged mean squared displacement \cite{didier:zhang:2017} to construct statistical tests for detecting physical heterogeneity in viscoelastic fluid samples starting from one or multiple observed anomalously diffusive paths. The methods are asymptotically valid for the range $0 < \alpha < 3/2$ and involve a mathematical characterization of time-averaged mean squared displacement bias and the effect of correlated disturbance errors. The assumptions on particle motion cover a broad family of fractional Gaussian processes, including fractional Brownian motion and many fractional instances of the generalized Langevin equation framework. We apply the proposed methods in experimental data from treated \emph{P.\ aeruginosa} biofilms generated by the collaboration of the Hill and Schoenfisch Labs at UNC-Chapel Hill.
\end{abstract}

\section{Introduction}

In this paper, we start from the asymptotic distribution of the time-averaged mean squared displacement of nanometric tracer particles \cite{didier:zhang:2017} to construct statistical protocols for detecting physical fluid heterogeneity. The assumptions on particle motion cover a broad family of fractional Gaussian processes, including fractional Brownian motion and many instances of the generalized Langevin equation framework. The testing protocols allowed providing more accurate quantitative analysis of experimental data from the Hill and Schoenfisch Labs (UNC-Chapel Hill), and the results reported in \cite{reighard:hill:dixon:worley:schoenfisch:2015} were generally confirmed.

Improvements in light microscopy, fluorescence techniques, nanoparticle synthesis and high-speed video have ushered in a flurry of experimental activity \cite{sokolov:2008}. Single particle tracking has become a common tool in many scientific areas, such as colloid physics \cite{hess:girirajan:mason:2006}, the microrheology of complex fluids \cite{mason:weitz:1995,suh:dawson:hanes:2005,matsui:wagner:hill:etal:2006,lai:wang:cone:wirtz:hanes:2009,hill:vasquez:mellnik:mckinley:vose:mu:henderson:donaldson:alexis:boucher:forest:2014} and the study of nanobiophysical systems, both \textit{in vivo} and \textit{in vitro} \cite{barkai:garini:metzler:2012,burnecki:muszkieta:sikora:weron:2012}. This includes the diffusion of single molecules, e.g., proteins, on biopolymers such as DNA or microtubules, on surfaces or in lipid membranes, inside \textit{in vivo} cells and in actin solutions \cite{tafvizi:mirny:oijen:2011,gorman:greene:2008,halford:marko:2004,vale:soll:gibbons:1989,helenius:brouhard:kalaidzidis:diez:howard:2006,minoura:katayama:eisaku:sekimoto:muto:2010,sonesson:elofsson:callisen:brismar:2007,wieser:schutz:2008,lasne:etal:2006,nishimura:etal:2006,smith:karatekin:etal:2011,bertseva:grebenkov:schmidhauser:gribkova:jeney:forro:2012,grebenkov:vahabi:bertseva:forro:jeney:2013},
as well as two-dimensional biological membranes \cite{metzler:jeon:cherstvy:2016} and heterogeneuous tracer diffusion and first passage characteristics in two-dimensional crowded environments \cite{ghosh:cherstvy:grebenkov:metzler:2016}.

Of primary concern in the analysis of particle path data is the ensemble mean squared displacement (MSD), where $X$ is the tracer particle's position. A basic dynamic characterization of the latter is given by the relation
\begin{equation}\label{e:EX2(t)=Dt(alpha)}
\langle X^2(t)\rangle = \bbE X^2(t) \propto \sigma^2 t^\alpha, \quad \sigma^2,\alpha > 0, \quad t \geq 0, \quad {\boldsymbol \xi} := (\log\sigma^2,\alpha).
\end{equation}
In \eqref{e:EX2(t)=Dt(alpha)}, $\alpha$ is the diffusion exponent and $\sigma^2 = 2D$, where $D$ is the diffusivity constant. The parameter value $\alpha = 1$ corresponds to classical diffusion. If $\alpha \neq 1$, the stochastic process $X$ is said to be \textit{anomalously diffusive}, more specifically sub- or superdiffusive depending on whether $\alpha < 1$ or $> 1$, respectively. Anomalous diffusion may emerge, for example, as a consequence of binding-unbinding events, of geometrical constraints on the particle's movement, or of fluid viscoelasticity \cite{meroz:sokolov:2015,saxton1996anomalous,saxton1994anomalous,levine2000one}. %However, in a larger context this boils down to identifying the particular stochastic model for a diffusing bead, be that a generalized Langevin equation, a continuous time random walk (CTRW), or a fractional Brownian motion (\cite{jeon:barkai:metzler:2013,lysy:pillai:hill:forest:mellnik:vasquez:mckinley:2016}). While the physical interpretation of the last model might be difficult, it may arise as a limiting model of the first two with the experiments lacking the resolution to distinguish the previous models.

The dominant statistical technique in the biophysical literature for estimating the parameters $\sigma^2$ and $\alpha$ is based on the so-named time-averaged mean squared displacement ($\MSD$). Suppose that a single particle experiment generates a tracer bead sample path with observations $X(j)$, $j=1,\hdots, N$. The pathwise statistic
\begin{equation}\label{e:MSD^}
M_N(\tau) := \frac{1}{N- \tau}\sum_{j=1}^{N- \tau} \{X(j+\tau)-X(j)  \}^2
\end{equation}
is the $\MSD$ at lag value $\tau$, i.e., the statistical counterpart of the MSD $\langle X^2( \tau) \rangle $. One generates an estimator of ${\boldsymbol \xi} = (\log \sigma^2, \alpha)$ by means of the linear regression
\begin{equation}\label{e:regression}
\log M_N(\tau_k)  = \log \sigma^2 + \alpha \log \tau_k  +\varepsilon_k, \quad k=1,\hdots,m,
\end{equation}
possibly over several independent particle paths, where $m$ is the number of lag values used and $\{\varepsilon_k \}_{k = 1,\hdots,m}$ is a random vector with an unspecified distribution and correlated entries (see \eqref{e:xi_OLS} and \eqref{e:lm_estimator}). Plots of $\MSD$ curves as a function of the lag value $\tau$, often on a log-log scale, are widely reported as part of anomalous diffusion data analysis (e.g., \cite{valentine:kaplan:thota:crocker:gisler:prudhomme:beck:weitz:2001,lieleg:vladescu:ribbeck:2010}). The choice of lag values $\tau_1,\hdots,\tau_m$ reflects the analyst's visual perception of the range where the slope of the $\MSD$ curves stabilize and thus indicate the true diffusive regime and power law.

The potential \textit{heterogeneity} of fluid samples in fields such as microrheology implies that estimating ${\boldsymbol \xi}$ from single trajectories is of great interest~\cite{burnecki:2012,burnecki:kepten:janczura:bronshtein:garini:weron:2012,vestergaard:blainey:flyvbjerg:2014,lysy:pillai:hill:forest:mellnik:vasquez:mckinley:2016}. However, the experimental and statistical difficulties involved in estimating ${\boldsymbol \xi}$ based on the regression system \eqref{e:regression}
have been pointed out by many authors. A non-exhaustive list of issues includes limited fluorophore lifetimes, proteins diffusing out of the field of view, finite-resolution imaging and motion blurring due to camera integration times, measurement errors, the presence of drifts and intra-path correlation \cite{qian:sheetzL:elson:1991,berglund:2010,michalet:berglund:2012,kepten:bronshtein:garini:2013,vestergaard:blainey:flyvbjerg:2014,burnecki:kepten:garini:sikora:weron:2015,mellnik:lysy:vasquez:pillai:hill:cribb:mckinley:forest:2016,briane:kervrann:vimond:2017,hoze:hochman:2017}.
Such difficulties call for a deeper understanding of the stochastic behavior of the $\MSD$ and, accordingly, a wealth of literature on the subject has developed. Starting from an underlying fractional stochastic process, several properties of the $\MSD$ such as ergodicity were established \cite{deng:barkai:2009,metzler:tejedor:jeon:he:deng:burov:barkai:2009,jeon:metzler:2010,burov:jeon:metzler:barkai:2011,
sandev:metzler:tomovksi:2012,boyer:dean:mejia:oshanin:2012,jeon:barkai:metzler:2013}. In particular, finite sample exact characterizations and mathematically convenient approximations to the distribution of the $\MSD$ under Gaussianity are provided in \cite{qian:sheetzL:elson:1991,grebenkov:2011prob,grebenkov:2011functionals,boyer:dean:mejia:oshanin:2012,andreanov:grebenkov:2012,nandi:heinrich:lindner:2012,boyer:dean:mejia:oshanin:2013,grebenkov:2013,sikora:teuerle:wylomanska:grebenkov:2017,gajda:wylomanska:kantz:chechkin:sikora:2018}.
In \cite{sikora:teuerle:wylomanska:grebenkov:2017}, assuming an observed fractional Brownian motion (see Example \ref{ex:fBm}), it is shown that the standard $\MSD$-based estimator is consistent, with vanishing bias and variance.

We say that a cumulative distribution function (c.d.f.) $F$ gives the \textit{asymptotic distribution }of a sequence of random variables $\{W_N\}_{N \in \bbN}$ if the c.d.f.\ $F_N(x)$ of $W_N$ converges to $F(x)$ at every $x \in \bbR$ where $F$ is continuous. Results on convergence in distribution such as the classical central limit theorem (CLT; see Example \ref{ex:CLT}) have a number of interesting statistical consequences. Typically, statements are robust, i.e., they hold for a multitude of models. Moreover, they naturally lead to useful data analysis protocols such as confidence intervals and hypothesis tests with error margins that are quantifiable and whose accuracy provably increases at an explicit rate (e.g., $\sqrt{N}$ for the CLT). In the probability literature, the study of the asymptotic distribution of sums of functions of Gaussian random variables has been carried out over many decades now (see \cite{rosenblatt:1961,taqqu:1975,taqqu:1979,dobrushin:major:1979,major:1981,giraitis:surgailis:1985,guyon:leon:1989,giraitis:surgailis:1990,pipiras:taqqu:2017} for just a few references). In the context of anomalous diffusion modeling, in turn, the related asymptotic distribution of the $\MSD$ was established in \cite{didier:zhang:2017} for a broad class of Gaussian fractional stochastic processes. It was shown that the convergence in distribution of the $\MSD$ occurs at different rates, and that the limiting distribution may be Gaussian or non-Gaussian, all depending on the value of the diffusion exponent $\alpha$. This made it possible, for example, to construct \textit{asymptotically valid} confidence intervals for the anomalous diffusion parameters starting from a single observed particle path.

In this paper, we propose particle path-based statistical protocols for detecting fluid heterogeneity that builds upon the $\MSD$'s asymptotic distribution. The protocols test fluid heterogeneity in two different experimental situations, namely,
\begin{itemize}\vspace{-2mm}
\item [$(i)$] assuming local physical homogeneity, whether different regions of the fluid are heterogeneous;\vspace{-2mm}
\item [$(ii)$] assuming global physical homogeneity of each fluid sample, whether two samples from each fluid are heterogeneous.\vspace{-2mm}
\end{itemize}
Hereinafter, these two senses are referred to as \textit{intra- and interfluid heterogeneity}, respectively. The testing methodology is based on an improved single-path $\MSD$-based estimation technique. To construct the latter, we tackle two of the main issues involved in $\MSD$-based estimation, namely: \textbf{(a)} the presence of \textit{bias} in log-$\MSD$-based methods; and \textbf{(b)} the effect of \textit{correlated disturbances} $\{\varepsilon_k \}_{k = 1,\hdots,m}$ in \eqref{e:regression}. Starting from a concentration inequality \cite{boucheron:lugosi:massart:2013}, we address these issues by providing mathematical characterizations of the bias and finite sample estimation variance which are by themselves of interest, as well as by introducing procedures for bias-correction and nearly optimal estimation under intra-path correlation. Motivated by applications in viscoelastic diffusion, the single-path estimation and heterogeneity testing protocols are mathematically established for $0 < \alpha < 3/2$, which covers all the subdiffusive range and part of the superdiffusive regime, and are asymptotically valid. For the sake of completeness, we also discuss and provide computational studies on the strong superdiffusivity range $3/2 \leq \alpha < 2$ (see Remark \ref{r:on_3/2=<alpha<2} on the difficulties involved in dealing with the possibly non-Gaussian asymptotic distribution of the $\MSD$). To guide experimental practice under common technical constraints such as limited camera recording time, we also apply the proposed tools in investigating the difference between observing longer particle paths and using a larger number of particle paths of given length. To illustrate the use of the protocols in physical practice, we make inferences on fluid viscoelasticity with data from the Hill and Schoenfisch Labs (UNC-Chapel Hill) on biofilm eradication, as first reported and described in \cite{reighard:hill:dixon:worley:schoenfisch:2015}.

The paper is organized as follows. In Section \ref{s:pre}, we summarize the key mathematical results on the asymptotic distribution of the $\MSD$. In Section \ref{s:improved_pathwise_estim}, assuming a single observed path of realistic length, we characterize the bias and the variance in $\MSD$-based estimation to construct the improved single-path estimator and compare it with the standard $\MSD$-based estimator in terms of statistical performance. In Section \ref{s:test_heter}, assuming multiple observed paths, we use the estimator developed in Section \ref{s:improved_pathwise_estim} to construct statistical testing protocols for intra- and interfluid heterogeneity detection. In Section \ref{s:experiments}, we model and test fluid heterogeneity through experimental data. For the reader's convenience, Section \ref{s:asympt_dist_MSD} of the Appendix contains mathematically accurate statements of the results in Section \ref{s:pre} and \cite{didier:zhang:2017}. Sections \ref{sec:app2sec1}, \ref{s:bias_var_asympt_dist_stand_estimator} and \ref{s:pseudocode} contain all new mathematical results and their proofs. Newly designed \texttt{Matlab} routines containing the estimation and testing protocols will be made available on the authors' websites at the time of publication.

\section{Background}\label{s:pre}

Before we revisit the results in \cite{didier:zhang:2017} on the asymptotic behavior of the $\MSD$, for the sake of exposition we consider some classical results from probability theory.

\begin{example}\label{ex:CLT}
Consider independent and identically distributed random variables $X_1,\hdots,X_N$, each with mean $\langle X_1 \rangle = \mu$ and finite variance
$\Var X_1 := \langle X^2_1 \rangle - \langle X_1 \rangle^2 = \varphi^2 > 0$. If $\overline{X}_N = N^{-1} \sum^{N}_{i=1}X_i$ denotes the sample mean, then the celebrated \textit{central limit theorem }states that, for large $N$, the distribution of the standardized sample mean approaches that of a standard normal, i.e.,
\begin{equation}\label{e:CLT}
\sqrt{N} \hspace{1mm}\frac{(\overline{X}_N - \mu)}{\varphi} \stackrel{d}\rightarrow {\mathcal N}(0,1), \quad N \rightarrow \infty.
\end{equation}
%This means that, if $F_N$ and $\Phi$ are, respectively, the cumulative distribution functions of the standardized sample mean and the standard normal law, then $$
%F_N(x) \rightarrow \Phi(x), \quad N \rightarrow \infty,
%$$
%at every point $x$.
Apart from naturally leading to confidence intervals and hypothesis tests, the convergence \eqref{e:CLT} also implies that $\overline{X}_N$ is a \textit{consistent} estimator of $\mu$, namely, it converges in probability to $\mu$. This is so because
\begin{equation}\label{e:Xbar_consistency}
\overline{X}_N - \mu =  \Big(\frac{\varphi}{\sqrt{N}} \Big) \sqrt{N}\hspace{1mm}\frac{(\overline{X}_N - \mu)}{\varphi}  \stackrel{P}\rightarrow 0, \quad N \rightarrow \infty.
\end{equation}
The zero limit in probability in \eqref{e:Xbar_consistency} stems from the fact that the vanishing term $\varphi/\sqrt{N} \rightarrow 0$ multiplies a standardized sample mean that converges in distribution \eqref{e:CLT} (see \cite{shiryaev:2000}).
\end{example}

Apart from distinct assumptions on the observations, the claims in \cite{didier:zhang:2017} on the asymptotic behavior of the $\MSD$ are reminiscent of the classical statements \eqref{e:CLT} and \eqref{e:Xbar_consistency}, with two differences: $(i)$ the rate of convergence is not typically $\sqrt{N}$ in biophysical modeling; $(ii)$ the asymptotic distribution of the $\MSD$ is not necessarily Gaussian.\\

So, consider the random vector
\begin{equation}\label{e:MSD^_vector}
\Big(M_N( \tau_1  ) ,\hdots, M_N( \tau_m  )\Big ),
\end{equation}
namely, a vector of $\MSD$ terms \eqref{e:MSD^} at $m$ different lag values, obtained from one path of a Gaussian, stationary increment process. Fitting \eqref{e:regression} and \eqref{e:MSD^_vector} by means of ordinary least squares (OLS) regression is the most intuitive way of constructing an estimator of the diffusion parameter vector ${\boldsymbol \xi} = (\log \sigma^2,\alpha)$. This corresponds to the common practice in the biophysical literature, both in experimental and methodological work (e.g., \cite{valentine:kaplan:thota:crocker:gisler:prudhomme:beck:weitz:2001,lieleg:vladescu:ribbeck:2010,burnecki:kepten:janczura:bronshtein:garini:weron:2012,lysy:pillai:hill:forest:mellnik:vasquez:mckinley:2016} among many references). Throughout this paper,
\begin{equation}\label{e:xi_OLS}
{\boldsymbol E}_{\ols} = ( L_{\ols}, A_{\ols})
\end{equation}
denotes this standard estimator (see \eqref{e:lm_estimator} for a precise expression). In this framework, we need to make the lag sizes $\tau_1,\hdots,\tau_m$ themselves go to infinity, though no faster than the sample size $N$. This mathematically expresses the practical analysis of anomalous diffusion data: the lag size has to be
\begin{itemize}
\item [$({\mathcal L}1)$] \textit{large enough} for the $\MSD$ regime to become \textit{log-linear};
\item [$({\mathcal L}2)$] but, at the same time, \textit{not too large} because of the \textit{increased variance} of the $\MSD$.
\end{itemize}
For a generic lag value $\tau$, we can model this idea by writing
\begin{equation}\label{e:h->infty}
\infty \leftarrow \tau \ll N.
\end{equation}
The limit and inequality in \eqref{e:h->infty} express $({\mathcal L}1)$ and $({\mathcal L}2)$, respectively (the accurate mathematical statements are given by condition \eqref{e:h(n)}; see also Figure \ref{fig:msd_shortmemo}).

%Condition \eqref{e:h->infty} is typically not needed when the particle motion $X$ is exactly self-similar, e.g., if it follows a fBm (see Example \ref{ex:fBm}; c.f.\ ). However, in the absence of exact self-similarity, the $\MSD$ is generally \textit{biased} due to the presence of non-fractional high frequency behavior. This is illustrated in Figure \ref{fig:msd_shortmemo} for both simulated and experimental data, where estimation based on the first few lag values $\tau$ leads to strongly biased estimates of $\alpha$.

\vspace{3mm}
\begin{minipage}{\linewidth}
\setlength{\heavyrulewidth}{1.5pt}
\setlength{\abovetopsep}{4pt}
\centering
 \begin{tabular}{lcc}\toprule
\textnormal{parameter range}  &  \textnormal{rate of convergence} & asymptotic distribution\\\hline
$0 < \alpha < 3/2$              & $\sqrt{\frac{N}{\tau}} \frac{1}{\tau^{\alpha}}$ & Gaussian\\
 $\alpha = 3/2$                 &  $\sqrt{\frac{N}{\log N}}\frac{1}{\tau^2}$  & Gaussian \\
 $3/2 < \alpha < 2$             &  $\frac{N^{2- \alpha}}{\tau^2}$ & non-Gaussian\\
 \bottomrule
\end{tabular}
\captionof{table}{Asymptotic behavior of the $\MSD$ random vector \eqref{e:MSD^_vector} (see Theorem \ref{t:asympt_dist_MSD}).}\label{table:asympt_dist_MSD}
\par
\bigskip

\end{minipage}
\vspace{3mm}

\vspace{3mm}
\setlength{\heavyrulewidth}{1.5pt}
\setlength{\abovetopsep}{4pt}
\begin{table}[!htbp]
\centering
\begin{tabular}{*5c}
\toprule
parameter range &  \multicolumn{2}{c}{rate of convergence} & joint asymptotic distribution & consistency\\
%\midrule
{}   & $ L_{\ols} $   & $ A_{\ols}$    &    &  \\\hline
$0 < \alpha < 3/2$    &  $\sqrt{\frac{N}{\tau}} \frac{1}{\tau^{\alpha}} \frac{1}{\log \tau}$ & $\sqrt{\frac{N}{\tau}} \frac{1}{\tau^{\alpha}}$   & Gaussian  & yes\\
$\alpha = 3/2$    &  $\sqrt{\frac{N}{\log N}}\frac{1}{\tau^2}\frac{1}{\log \tau}$ & $\sqrt{\frac{N}{\log N}}\frac{1}{\tau^2}$   & Gaussian  & yes\\
 $3/2 < \alpha < 2$   &  $\frac{N^{2- \alpha}}{\tau^2}\frac{1}{\log \tau}$  &  $\frac{N^{2- \alpha}}{\tau^2}$   & non-Gaussian  & yes\\
\bottomrule
\end{tabular}
\caption{Asymptotic behavior of the standard $\MSD$-based estimator \eqref{e:xi_OLS} (see Corollary \ref{c:asympt_dist_MSD}).} \label{table:asympt_dist_OLS}
\end{table}
\vspace{3mm}

The asymptotic distribution of the $\MSD$ random vector \eqref{e:MSD^_vector} after centering is briefly described in Table \ref{table:asympt_dist_MSD}. This leads to the asymptotic behavior of the standard estimator \eqref{e:xi_OLS}, which is summarized in Table \ref{table:asympt_dist_OLS} in terms of convergence rate, asymptotic distribution and consistency. In both cases, the value of $\alpha$ determines the convergence rate and the nature of the asymptotic distribution. In particular, over almost the whole strong superdiffusivity range (i.e., over $3/2 < \alpha < 2$), the asymptotic distribution is non-Gaussian (Rosenblatt-type; see Theorem \ref{t:asympt_dist_MSD} and \cite{rosenblatt:1961,taqqu:1975,taqqu:2011,veillette:taqqu:2013}). For any instance, by an argument analogous to \eqref{e:Xbar_consistency}, the standard estimator is consistent, i.e.,
\begin{equation}\label{e:consistency}
{\boldsymbol E}_{\ols} \stackrel{P}\rightarrow \boldsymbol \xi.
\end{equation}
%where $\stackrel{P}\rightarrow$ denotes convergence in probability.

The family of stochastic processes for which the limits in distribution in Tables \ref{table:asympt_dist_MSD} and \ref{table:asympt_dist_OLS} hold is broad and contains a number of popular models. Three examples are fractional Brownian motion (fBm), fractional instances of the generalized Langevin equation (GLE) and the (integrated) fractional Ornstein-Uhlenbeck process (ifOU).

\begin{example}\label{ex:fBm}
Together with the continuous time random walk, fBm is one of the most popular models of anomalous diffusion \cite{taqqu:2003,barkai:garini:metzler:2012}. For some value of the so-named Hurst parameter $H \in (0,1)$ and a variance parameter $D > 0$, a fBm $B_H(t)$ is the only Gaussian, stationary increment process with covariance function
\begin{equation}\label{e:fBm_cov}
\langle B_H(s)B_H(t)\rangle = D \{|t|^{2H} + |s|^{2H} - |t-s|^{2H}\}, \quad s,t \in \bbR.
\end{equation}
The particular parameter value $H=1/2$ corresponds to the ordinary Brownian motion (Wiener process). In view of \eqref{e:fBm_cov}, which implies exact self-similarity, for fBm the MSD scaling relation \eqref{e:EX2(t)=Dt(alpha)} holds as an equality, i.e.,
\begin{equation}\label{e:MSD_fBm}
\langle B^2_H(t) \rangle = \sigma^2 t^{\alpha}, \quad t \in \bbR,
\end{equation}
where %$\sigma^2 = C^2_H$ and
\begin{equation}\label{e:alpha=2H}
\sigma^2 = 2D, \quad \alpha = 2H.
\end{equation}
\end{example}

\begin{example}
The GLE has been used as a universal model of anomalous diffusion in the biophysical field of microrheology \cite{mason:weitz:1995,zwanzig:2001,ottobre:pavliotis:2011,nguyen:mckinley:2017}. A subclass of interest of the GLE framework is the fractional GLE family \cite{kou:xie:2004,kou:2008,didier:fricks:2014}, which is obtained almost surely as the solution of the stochastic differential equation
\begin{equation}\label{e:fGLE}
m \hspace{1mm}dV(t) = - \lambda \int^{t}_{-\infty}\Gamma(t-s)V(s)dsdt + dB_{H}(t), \quad 1/2 < H < 1.
\end{equation}
In \eqref{e:fGLE}, $m,\lambda > 0$ and the memory kernel has the form $\Gamma(t) = 2H(2H-1)|t|^{2H-2}$, $t \neq 0$, which is a consequence of invoking the fluctuation-dissipation relation \cite{didier:mckinley:hill:fricks:2012,lysy:pillai:hill:forest:mellnik:vasquez:mckinley:2016}. The integrated fractional generalized Langevin process (ifGL) is given by $X(t) = \int^{t}_{0}V(s) ds$, $t > 0$, where $\{V(t)\}_{t \geq 0}$ is a solution of the fractional GLE. For the ifGL, relation \eqref{e:EX2(t)=Dt(alpha)} holds with $\alpha = 2(1 - H)$ (subdiffusive) as $t \rightarrow \infty$.
\end{example}

\begin{example}
The ifOU is given by $X(t) = \int^{t}_{0}V(s) ds$, $t > 0$, where the so-named fractional Ornstein-Uhlenbeck process $\{V(t)\}_{t \geq 0}$ is the almost surely continuous solution to the fBm-driven Langevin equation
\begin{equation} \label{e:fOU_SDE}
dV(t) = - \lambda V(t) dt + \varphi \hspace{0.5mm}dB_{H}(t), \quad t \geq 0, \quad \lambda>0, \quad 0 < H< 1
\end{equation}
(see \cite{cheridito:kawaguchi:maejima:2003,prakasarao:2010}). The ifOU process is a mathematically convenient model of anomalous diffusion. In the subdiffusive range, it displays a similar correlation structure to that of the ifGL process. For the ifOU, relation \eqref{e:EX2(t)=Dt(alpha)} holds with \eqref{e:alpha=2H} as $t \rightarrow \infty$.  %Furthermore, for the generic process $X$ in \eqref{x_spec_rep},
%\begin{equation}\label{e:ex(s)x(t)_B}
%\Big| \frac{\bbE X^2(t)}{\sigma^2 t^{\alpha}} - 1\Big| \leq \frac{C}{t^{\delta}}, \quad t > 0,
%\end{equation}
%for appropriate constants $\sigma^2$ and $C$ (see \cite{didier:zhang:2017}). In other words, and in light of \eqref{e:MSD_fBm}, the MSD of $X$ differs from that of a fBm by a factor that vanishes at large time scales.
\end{example}

\begin{remark}
The results in \cite{didier:zhang:2017} do not cover some important anomalous diffusion models such as continuous time random walks. For the latter family of models, limit theorems typically involve distinct nonstandard asymptotic distributions depending on the assumptions (see, for instance, \cite{meerschaert:scheffler:2004,meerschaert:nane:xiao:2009} and references therein; for general guidelines on the use of the $\MSD$, see \cite{kepten:weron:sikora:burnecki:garini:2015}).
\end{remark}

\section{Improved $\MSD$-based estimation}\label{s:improved_pathwise_estim}

The standard estimator ${\boldsymbol E}_{\ols}= ( L_{\ols}, A_{\ols})$ in \eqref{e:xi_OLS} has at least two significant shortcomings: \textit{finite sample bias} and suboptimal performance \textit{in the presence of correlation} among the regression disturbance terms $\{\varepsilon_{k}\}_{k=1,\hdots,N}$. We propose a single-path improved estimation protocol that addresses these issues. Accordingly, it involves two components, which we describe next. These two components involve asymptotically valid mathematical expressions for finite-sample bias and variance. Hereinafter, different lag values are expressed as
\begin{equation}\label{e:hk}
\tau_k = w_k \tau, \quad w_1 < \hdots < w_k,
\end{equation}
for fixed constants $w_\cdot$, where $\tau=\tau(N)$ grows as function of $N$.\\ %Recall that $ A_{\ols}$ is the second entry of the estimator ${\boldsymbol E}_{\ols}$ obtained from the original regression system \eqref{e:regression}. \\ %i.e., (a) bias correction and (b) a quasi-optimal regression procedure based on accounting for correlation among disturbance errors.
%\begin{itemize}
%\item [$(a)$] a finite sample bias correction in the regression system of equations;
%\item [$(ii)$] a nearly optimal regression procedure.% in the presence of correlated errors.
%\end{itemize}
%We now provide a description of each one of them.\\

\noindent \textbf{(a) Bias correction.} In $\MSD$-based scaling analysis, there at least two sources of bias. First, bias appears if the particle movement is not exactly self-similar (not a fBm), i.e.,
 $$
 \langle X^2(t) \rangle \neq \sigma^2 t^{\alpha} \quad \textnormal{over a range of $t$}.
 $$
In fact, the deviation of the MSD from exact self-similarity or power scaling is generally controlled by the relation
\begin{equation}\label{e:bound_ensemble_MSD}
\Big| \frac{ \langle X^{2}(t)\rangle}{\sigma^2 t^{\alpha}} - 1\Big| \leq \frac{C}{t^{\delta}} \quad \textnormal{for large $t$},
\end{equation}
for some constant $\sigma^2 > 0$, where the deviation parameter $\delta > 0$ mostly depends on the high frequency behavior of the particle motion (see Proposition \ref{p:bound_ensemble_MSD}). Second, even under self-similarity, bias stems from the elementary fact that the logarithm of the ensemble average and the ensemble average of the logarithm are distinct (e.g., \cite{veitch:abry:1999,moulines:roueff:taqqu:2007:spectral,moulines:roueff:taqqu:2007fractals,moulines:roueff:taqqu:2008}). In the context of \eqref{e:regression}, this means that
$$
\langle\log M_N(\tau)\rangle \neq \log \langle M_N(\tau)\rangle = \alpha \log \tau + \log \sigma^2, \quad \tau \in \bbN.
$$
So, by reinterpreting $\log M_N(\tau)$ itself as an estimator of $\alpha \log \tau + \log \sigma^2$, we can express the bias involved in $\MSD$-based estimation as
\begin{equation}\label{e:bias_log}
\langle\log M_N(\tau) \rangle - (\alpha \log \tau + \log \sigma^2)  = - \frac{\tau}{N}\beta_N(\alpha,\tau) + O\Big(\frac{1}{\tau^{\delta}}\Big) + O\Big(\frac{\tau}{N}\Big)
\end{equation}
for the same $\delta>0$ as in \eqref{e:bound_ensemble_MSD} (for $0 < \alpha < 3/2$ -- see Theorem \ref{t:logmubias}; see also Remark \ref{r:on_3/2=<alpha<2} on the range $3/2 \leq \alpha < 2$). The term of order $O(\tau^{-\delta})$, then, is mostly determined by the high frequency behavior of the anomalously diffusive particle (see Figure \ref{fig:msd_shortmemo} and expressions \eqref{e:s(x)}, \eqref{e:delta}). In \eqref{e:bias_log}, the main bias factor is given by the function
\begin{equation}\label{e:beta(alpha,h,n)}
\beta_N(\alpha,\tau) = \frac{1}{4 \tau} \sum_{i=-N+1}^{N-1} \bigg(1- \frac{\abs{i}}{N}\bigg)  \Big\{ \abs{ \frac{i}{\tau} + 1}^{\alpha} -2 \abs{ \frac{i}{\tau} }^{\alpha} + \abs{ \frac{i}{\tau} - 1}^{\alpha} \Big\}^2.
\end{equation}
% and
%\begin{equation}\label{e:varrhoe1}
%  \varrho(h,n,\alpha) = \left\{
%                                                   \begin{array}{ll}
%                                                     (\frac{\tau}{N})^{1/2}, & 0<\alpha<3/2; \\
%                                                     (\frac{h \log N}{N})^{1/2}, & \alpha = 3/2; \\
%                                                     (\frac{\tau}{N})^{2-\alpha}, & 3/2<\alpha <1.
%                                                   \end{array}
%                                                 \right.
%\end{equation}
Note that \eqref{e:beta(alpha,h,n)} depends on the unknown parameter $\alpha$. So, we use $ A_{\ols}$ and \eqref{e:beta(alpha,h,n)} to define an estimator of the bias vector by
\begin{equation}\label{e:beta(alphahat,h,n)}
\Big( \beta_N( A_{\ols},\tau_k) \Big)_{k=1,\hdots,m}.
\end{equation}

%One of them decays to zero as a function of the lag values $\tau$. This reflects the familiar notion that the regression procedure can only be conducted on lag values large enough that the associated $\MSD$ curve approaches its large scale behavior and becomes approximately flat. The other component is a (known) function of the unknown diffusion exponent $\alpha$. Therefore, a bias correction procedure consists of using $ A_{\ols}$ to generate a plug-in estimator of the latter and subtracting it from the regression equations.
%
%On the other hand, the error terms in \eqref{e:regression} are generally correlated. In any regression problem with correlated errors, a GLS procedure is expected to outperform OLS in terms of mean squared error (MSE) \cite{christensen:2011}.

\noindent \textbf{(b) Accounting for disturbance correlation.} In linear estimation theory, the method for dealing with correlated random errors is called generalized least squares (GLS). In fact, the resulting GLS estimator is the best linear unbiased estimator, since it outperforms its OLS counterpart in terms of mean squared error (MSE) (see  \cite{christensen:2011}).

In the context of $\MSD$-based estimation, to better understand the difference between the standard, OLS-based estimator and the related GLS-based estimator, recast the vector system \eqref{e:regression} as the regression model
\begin{equation}\label{e:MSD_regression}
{\mathbf z} = X \boldsymbol\xi + {\boldsymbol \varepsilon}.
\end{equation}
In \eqref{e:MSD_regression}, the term $\boldsymbol\xi$ is as in \eqref{e:EX2(t)=Dt(alpha)}, and the dependent variable and the regressor are given by, respectively,
\begin{equation}\label{e:lm_notation}
{\mathbf z} = \Big(\log M_N(\tau_k)\Big)_{k=1,\hdots,m}, \quad
X = \left(
      \begin{array}{cc}
        1 & \log \tau_1 \\
        \vdots & \vdots \\
        1 & \log \tau_m \\
      \end{array}
    \right).
\end{equation}
It is well known that the expression
\begin{equation}\label{e:lm_estimator}
{\boldsymbol E}_{\ols}:=(X^T X)^{-1}X^T {\mathbf z} = ( L_{\ols}, \hspace{1mm} A_{\ols} \hspace{0.5mm})^T
\end{equation}
gives the standard estimator \eqref{e:xi_OLS} generated by the OLS solution to the system \eqref{e:MSD_regression}. By contrast, let \begin{equation}\label{e:hetero_tildesigma}
  \Upsilon({\boldsymbol \xi}) = \Big(\upsilon_{k_1,k_2}({\boldsymbol \xi})\Big)_{k_1,k_2 = 1,\hdots,m},
\end{equation}
be the finite sample covariance matrix of the vector ${\mathbf z}$ as in \eqref{e:lm_notation}. The GLS solution is given by
\begin{equation}\label{e:GLS_solution}
(X^T\Upsilon({\boldsymbol \xi})^{-1} X)^{-1}X^T\Upsilon({\boldsymbol \xi})^{-1} {\mathbf z},
\end{equation}
which involves the unknown matrix \eqref{e:hetero_tildesigma}. In practice, then, one needs to estimate such matrix. For this purpose, we first establish the entrywise expansion
  \begin{equation}\label{e:tildesigmak1k2e1}
  \upsilon_{k_1,k_2}({\boldsymbol \xi}) = \frac{\tau}{N}  \hspace{1mm}\varsigma_N(\alpha,\tau_{k_1},\tau_{k_2}) + O\Big(\frac{\tau^{1-\delta}}{N}\Big) + o\Big( \frac{\tau}{N}\Big), \quad k_1,k_2 = 1,\hdots,m
  \end{equation}
  (for $0 < \alpha < 3/2$ -- see Theorem \ref{t:tildesigmak1k2}; see also Remark \ref{r:on_3/2=<alpha<2} on the range $3/2 \leq \alpha < 2$). In \eqref{e:tildesigmak1k2e1}, the main variance factor is given by
  $$
  \varsigma_N(\alpha,\tau_{k_1},\tau_{k_2}) = \frac{1}{2 \tau } \sum_{i=-N+1}^{N-1} \bigg( 1- \frac{\abs{i}}{N}\bigg)\bigg\{ \abs{ \frac{i}{\sqrt{\tau_{k_1} \tau_{k_2}}} + \sqrt{\frac{\tau_{k_1}}{\tau_{k_2}}} }^{\alpha}
  -\abs{ \frac{i}{\sqrt{\tau_{k_1} \tau_{k_2}}} + \sqrt{\frac{\tau_{k_1}}{\tau_{k_2}}} - \sqrt{\frac{\tau_{k_2}}{\tau_{k_1}}}}^{\alpha}
 $$
 \begin{equation}\label{e:varsigma-n(alpha,h1,h2)}
-\abs{ \frac{i}{\sqrt{\tau_{k_1} \tau_{k_2}}} }^{\alpha}  + \abs{ \frac{i}{\sqrt{\tau_{k_1} \tau_{k_2}}} - \sqrt{\frac{\tau_{k_2}}{\tau_{k_1}}}}^{\alpha} \bigg\}^2.
  \end{equation}
%Starting from \eqref{e:varsigma-n(alpha,h1,h2)}, each entry of the covariance matrix $\Upsilon({\boldsymbol \xi})$ in \eqref{e:hetero_tildesigma} can be approximated by a function of the unknown parameter $\alpha$.
Note that expression \eqref{e:varsigma-n(alpha,h1,h2)} does not involve the constant $\sigma^2$, but it is still a function of the unknown parameter $\alpha$. Second, and in view of this, we can use $ A_{\ols}$ and \eqref{e:varsigma-n(alpha,h1,h2)} to define an estimator of the covariance matrix by
\begin{equation}\label{e:Sigmatilde(alpha-hat)}
\Upsilon( A_{\ols}) := \Big( \frac{\tau}{N}\varsigma_N( A_{\ols},\tau_{k_1},\tau_{k_2})\Big)_{k_1,k_2=1,\hdots,m}.\vspace{1mm}
\end{equation}

Drawing upon \textbf{(a)} and \textbf{(b)}, we can further construct an improved estimator of ${\boldsymbol \xi}$ by a quasi-GLS procedure based on the estimator $\Upsilon( A_{\ols})$ and by replacing \eqref{e:regression} with the bias-corrected regression system
\begin{equation}\label{e:regression_bias-corrected}
\log M_N(\tau_k)  + \frac{\tau_k}{N} \beta_N( A_{\ols},\tau_k)
= \log \sigma^2 + \alpha \log \tau_k +\varepsilon_k, \quad k=1,\hdots,m.
\end{equation}
The resulting estimator can be expressed as
\begin{equation}\label{e:hetero_hatxi}
  {\boldsymbol E} = ( L, A)
    = (X^T \Upsilon^{-1}( A_{\ols}) X)^{-1} X^T \Upsilon^{-1}( A_{\ols}) {\mathbf{y}},
\end{equation}
where $X$ is again as in \eqref{e:lm_notation} and
\begin{equation}\label{e:y,X}
{\mathbf y} = \left(\begin{array}{c}
\log M_N(\tau_1)  + \frac{\tau_1}{N} \beta_N( A_{\ols},\tau_1)\\
\vdots \\
\log M_N(\tau_m)  + \frac{\tau_m}{N} \beta_N( A_{\ols},\tau_m)\\
\end{array}\right).
\end{equation}
For the reader's convenience, the construction of the estimator ${\boldsymbol E}$ is summarized in the form of pseudocode in Appendix \ref{s:pseudocode}.

To compare the performances of ${\boldsymbol E}$ and ${\boldsymbol E}_{\ols}$, we generated 1000 independent paths of length $2^{10}$ and estimated the diffusion exponent based on the two methods. Figure \ref{fig:ols_vs_wls} displays the results in terms of Monte Carlo bias, standard deviation and square root MSE. The improved estimator ${\boldsymbol E}$ outperforms the usual estimator ${\boldsymbol E}_{\ols}$ by any of the three criteria for different values of $\alpha$.\\

%\begin{remark}\label{r:residual_O(\tau^(-delta))}
%The term of order $O(\tau^{-\delta})$ on the right hand side of \eqref{e:bias_log} is mostly determined by the high frequency (short memory) behavior of the anomalous diffusion process. The small scale effect of short memory can be observed on plots of MSD curves vs lag values on a $\log_2-\log_2$ scale. The left and right plots in Figure \ref{fig:msd_shortmemo} display, respectively, MSD curves from simulated ifOU with $\alpha = 0.6$ and from experimental data provided by the David B.\ Hill Lab (COS2-NO at 8 mg ml${}^{-1}$). In both cases, using small lag values in the $\log_2 \MSD$ regression leads to conspicuously biased estimates of $\alpha$. With simulated data, we obtain the estimates $ A = 1.42$ and $0.70$ by using small and large lag values, respectively. With real data, the use of small lag values yields the estimate $ A = 1.10$. This would suggest that the observed particle movement is not subdiffusive, whereas the use of large lag values leads to the opposite conclusion ($ A=0.69$).\vspace{1mm} %On the other hand, the lag values cannot be chosen too large by comparison to the path length $N$. Otherwise, this leads to the inclusion of $\MSD$ terms $M_N(\tau)$ with large variance due to the small number of terms in the average.
%\end{remark}

\begin{remark}
Note that the main bias and variance factors $\beta_{N}(\alpha,\tau)$ and $\varsigma_N(\alpha,\tau_{k_1},\tau_{k_2})$ in \eqref{e:beta(alpha,h,n)} and \eqref{e:varsigma-n(alpha,h1,h2)}, respectively, converge as $N \rightarrow \infty$ (see Lemma \ref{l:bias_var_converge}). Moreover, after standardization, the estimator \eqref{e:hetero_hatxi} is provably asymptotically normal and consistent for $0 < \alpha < 3/2$ (see \eqref{e:xi_to_zeta_estimvar} in Section \ref{s:test_heter} and Proposition \ref{p:asympt_dist_estim}). See also Remark \ref{r:on_3/2=<alpha<2} on the range $3/2 \leq \alpha < 2$. \vspace{1mm}
\end{remark}

\begin{remark}\label{r:on_3/2=<alpha<2}
Although we do not provide proofs in this paper, the methods developed in this section and also in Section \ref{s:test_heter} can be extended to the strongly superdiffusive range $3/2 \leq \alpha < 2$. For example, due to nonstandard convergence rates, expressions \eqref{e:bias_log} and \eqref{e:tildesigmak1k2e1} hold after replacing $O(\frac{\tau}{N})$ with $O((\frac{\tau \log N}{N}))$ (for $\alpha = 3/2$) or $O((\frac{\tau}{N})^{4 - 2 \alpha})$ (for $3/2 < \alpha < 2$). Likewise, the asymptotic non-Gaussian distribution of the estimator \eqref{e:xi_to_zeta_estimvar} in Section \ref{s:test_heter}, with nonstandard convergence rates, can be established. However, inference involving the nonstandard limiting distribution can be cumbersome, and the computational studies in this section and in Section \ref{s:test_heter} show that the methods in the proposed format work reasonably well for realistic path lengths. See also \cite{didier:zhang:2017}, Remark 2, on how to construct asymptotically valid confidence intervals for ${\boldsymbol \xi}$ based on the standard estimator ${\boldsymbol E}_{\ols}$ assuming prior knowledge that $3/2 < \alpha < 2$.
\end{remark}

\begin{figure}[htbp]
\begin{center}
\includegraphics[scale=0.46]{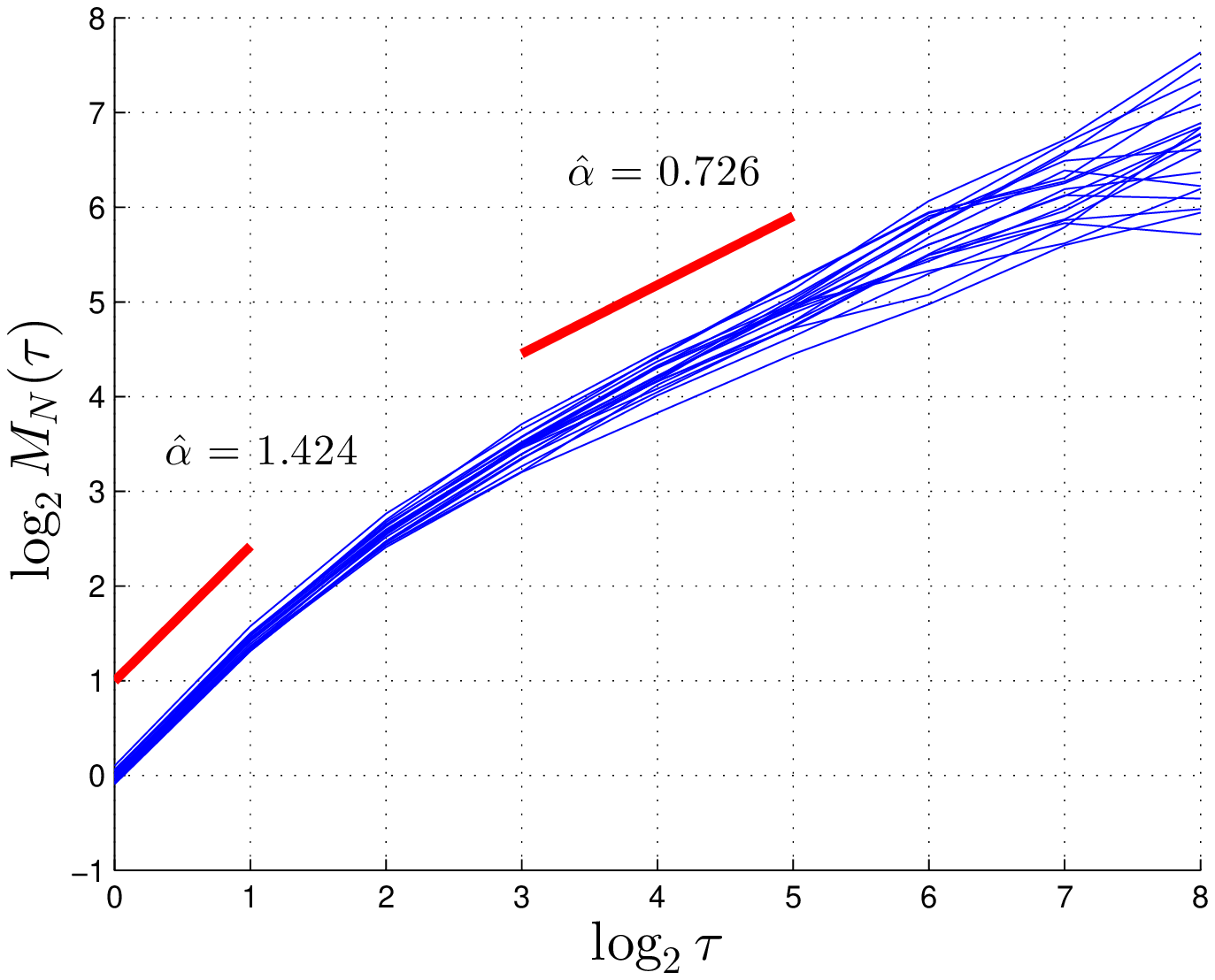}
\includegraphics[scale=0.5]{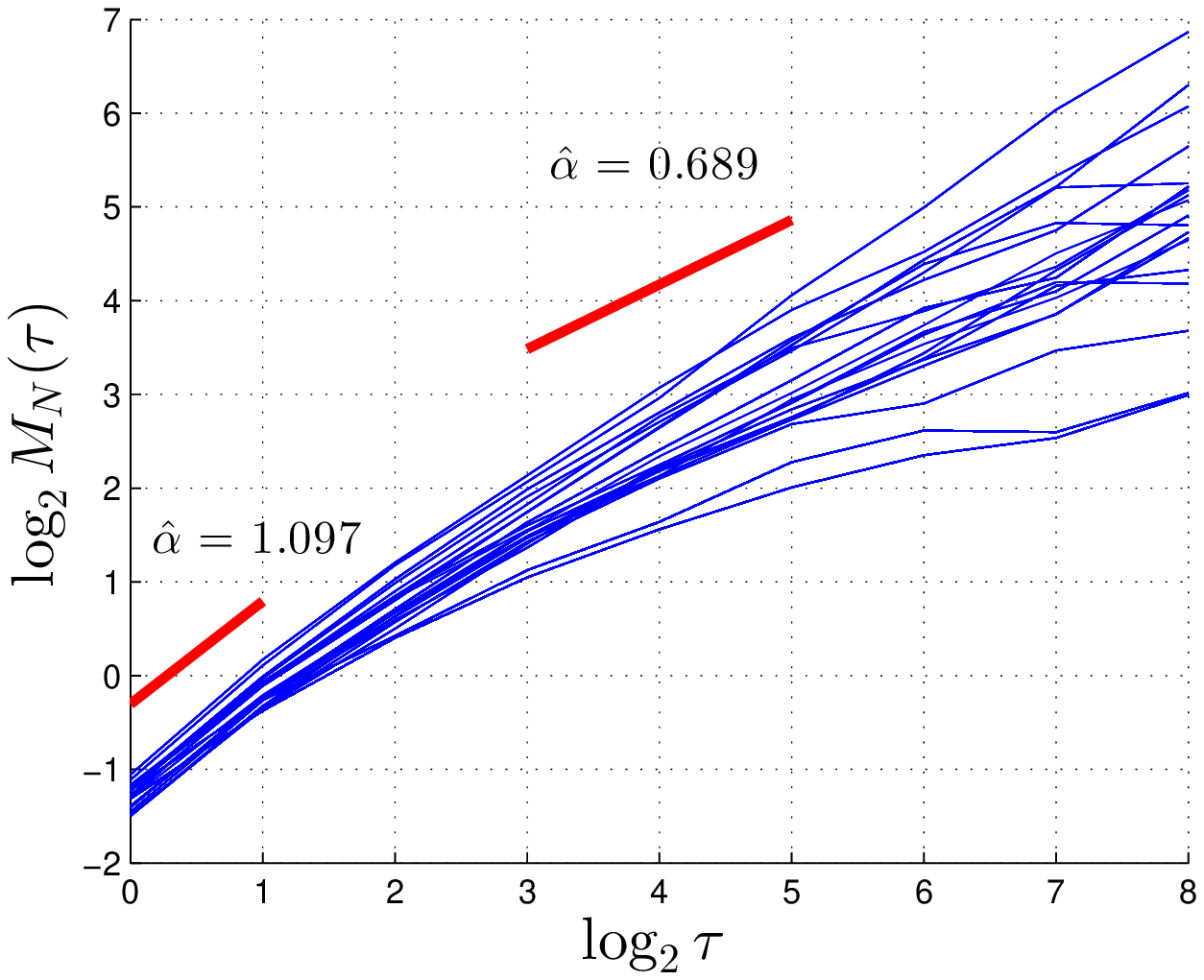}
\caption{\textbf{Bias and inconsistency over small scales $\tau$ vs vanishing bias and consistency as $\tau\rightarrow \infty$.} In general, taking the double limit $\tau, N \rightarrow \infty$ (see \eqref{e:h->infty}) is necessary. Over \textit{fixed} (``small") lag values $\tau$, $\MSD$-based estimation is biased and, for most anomalous diffusion models other than fBm, inconsistent. As mathematically characterized by expansion \eqref{e:bias_log}, estimation bias is fundamentally a consequence of the fact that $\langle \log \cdot \rangle \neq \log \langle \cdot \rangle$ and of the presence of the small scale factor $O(\tau^{-\delta})$, whereas, in turn, inconsistency generally appears as a consequence of this same factor. %Using small lag values $\tau$ often leads to conspicuously biased estimates of $\alpha$. To illustrate this,
The left and right plots show, respectively, 20 independent ifOU paths (length $2^{11}$, $\alpha = 0.6$) and 20 particle paths (length 1800) from \emph{P.\ aeruginosa} biofilm after COS2-NO treatment at concentration level 8 mg ml${}^{-1}$. The first and second red lines in each plot indicate, respectively, the fitted slope over small ($\tau = 1,2$) and large ($\tau = 8,32$) lag values. Based on the former lag values, $ A = 1.42 $ and 1.10 (evidence of superdiffusivity) for simulated and experimental data, respectively, whereas, by contrast, $ A = 0.70$ and 0.69 (evidence of subdiffusivity) based on the latter. This illustrates the fact that bias and inconsistency vanish when $\tau$ (and $N$) becomes large.}
\label{fig:msd_shortmemo}
\end{center}
\end{figure}

\begin{figure}[htbp]
\begin{center}
\includegraphics[scale=0.5]{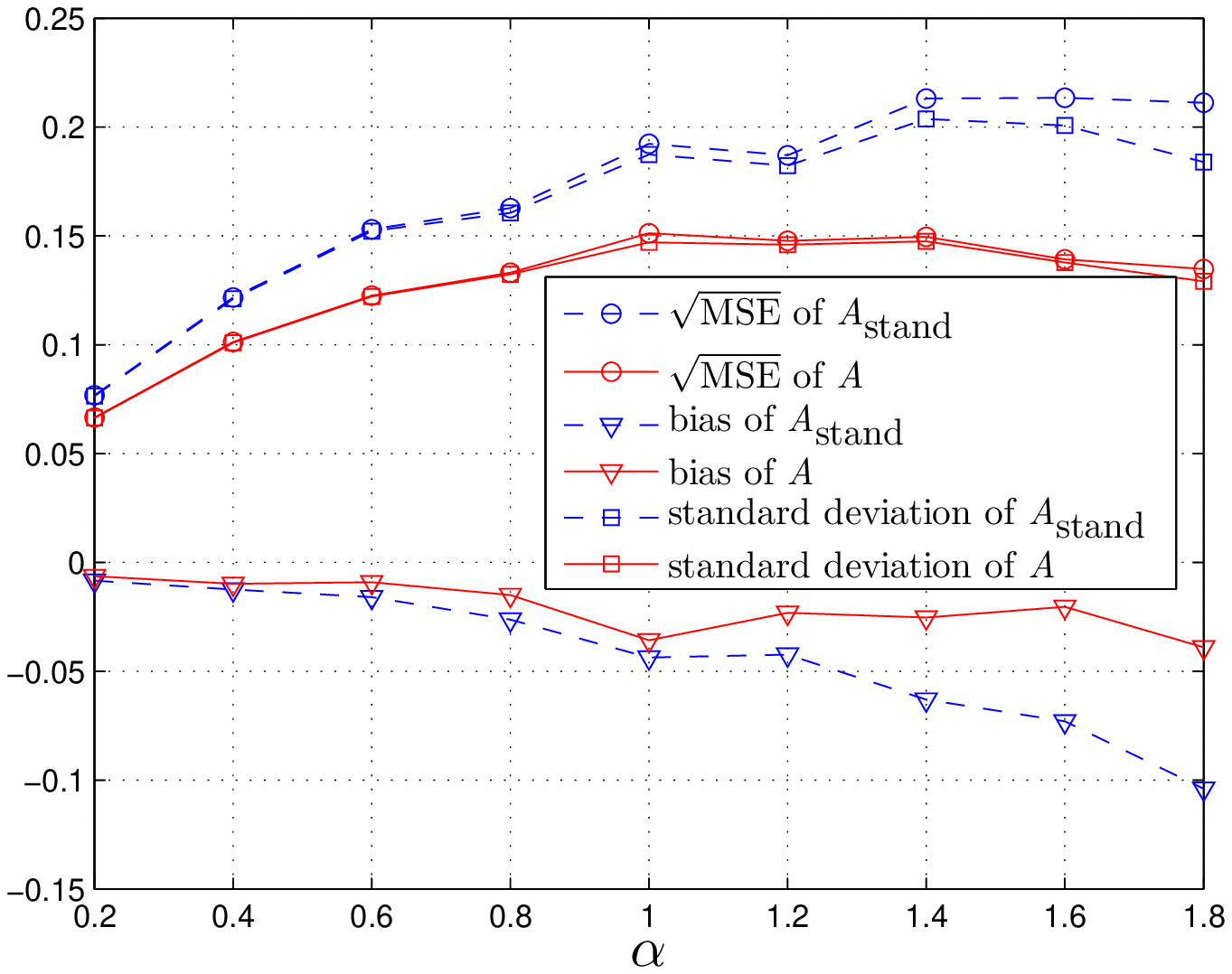}
\includegraphics[scale=0.5]{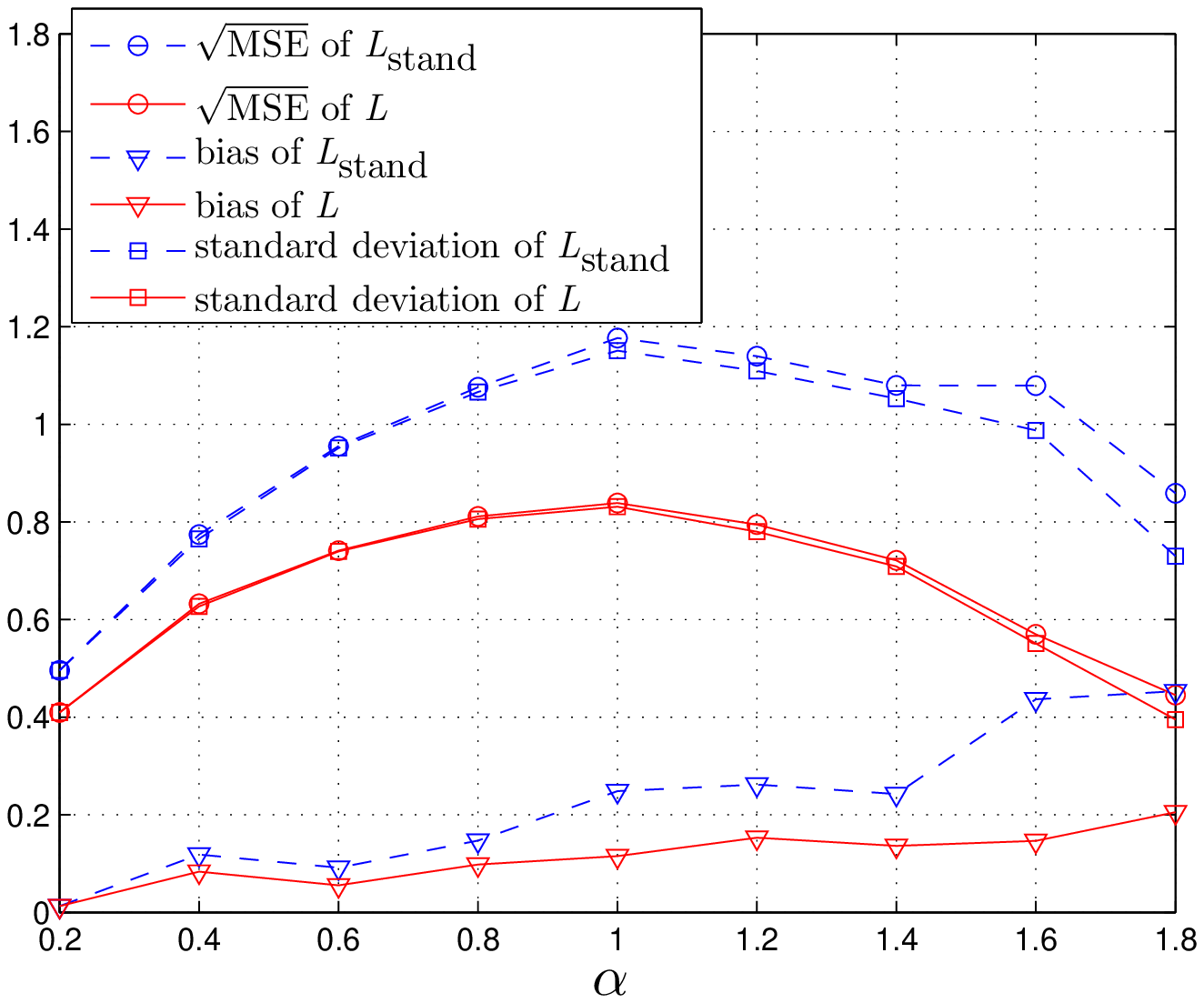}
\caption{\textbf{Comparative analysis of bias, standard deviation and MSE} of the estimators ${\boldsymbol E}$ (see \eqref{e:hetero_hatxi}) and ${\boldsymbol E}_{\ols}$ (see \eqref{e:xi_OLS}) as a function of diffusion exponents $\alpha$ ($x$-axis). Left plot: estimation of $\alpha$. Right plot: estimation of $\log \sigma^2$. Solid and dashed lines represent ${\boldsymbol E}$ and ${\boldsymbol E}_{\ols}$, respectively. For any parameter value $\alpha$, the proposed estimator ${\boldsymbol E}$ has smaller bias, standard deviation and square root MSE than ${\boldsymbol E}_{\ols}$. The total number of Monte Carlo runs is 1000 based on paths of length $2^{10}$.}
\label{fig:ols_vs_wls}
\end{center}
\end{figure}

%To study the effect of bias correction on the performance of the estimator ${\boldsymbol E} = ( L, A)$, we generated 1000 independent paths of length $2^{10}$ and estimated the diffusion exponent $\alpha$ in the presence or absence of bias correction. The results are shown in Figure \ref{fig:wls_bias}. As expected, the bias of $ A$ is smaller after bias correction. In addition, it is near zero for all values of the parameter $\alpha$, with the exception of the range of extreme superdiffusivity ($\alpha=1.8$).
%
%
%\begin{figure}[htbp]
%\begin{center}
%\includegraphics[scale=0.6]{gls_bias_alp.eps}\\
%\includegraphics[scale=0.6]{gls_bias_log2theta.eps}
%\caption{$x$-axis: $\alpha$. The effect of the bias correction procedure in the performance of the estimator ${\boldsymbol E}$ as a function of the diffusion exponent $\alpha$. For comparison, the latter is also considered without bias correction. The version of the estimator with bias correction has a significantly smaller bias than that without bias correction.}
%\label{fig:wls_bias}
%\end{center}
%\end{figure}

\section{Testing heterogeneity}\label{s:test_heter}

Single particle tracking experiments with viscoelastic diffusion often generate data in the form of multiple particle paths. As discussed in the Introduction, fluid heterogeneity can be tested in the \textit{intra-} and \textit{interfluid} senses. The pathwise framework constructed in Section \ref{s:improved_pathwise_estim} can be used in new testing protocols with good finite sample and asymptotic properties. We remind the reader that the mathematical statements cover the diffusion exponent range $0 < \alpha < 3/2$ (subdiffusive and mildly superdiffusive range), though we also include computational experiments for the strongly superdiffusive range $3/2 \leq \alpha < 2$.

Tables \ref{table:hypotheses} and \ref{table:tests} display the proposed framework. For each type of fluid heterogeneity, they show the appropriate hypotheses and testing procedures, respectively. In the remainder of this section, we provide a detailed description of the protocols. To set the notation, we recall that, for a given a hypothesis test, the conditional probability
\begin{equation}\label{e:test_size}
\bbP (H_0 \textnormal{ is rejected }| \hspace{0.5mm}{\boldsymbol \xi} \textnormal{ satisfies }H_0) =: \epsilon \in [0,1]
\end{equation}
is called the \textit{size} (or significance level) of the test, whereas the function
\begin{equation}\label{e:test_power}
{\boldsymbol \xi} \mapsto \bbP (H_0 \textnormal{ is rejected } | \hspace{0.5mm} {\boldsymbol \xi} \textnormal{ satisfies }H_1) \in [0,1]
\end{equation}
is called the \textit{power} of the test.\\

\begin{minipage}{\linewidth}
\centering
\setlength{\heavyrulewidth}{1.5pt}
\setlength{\abovetopsep}{4pt}
\begin{tabular}{lcc}\toprule
\textnormal{heterogeneity}  &  $H_{0}$ & $H_{a}$ \\\hline
\textnormal{intrafluid}   & ${\boldsymbol \xi}_{1} = \cdots = {\boldsymbol \xi}_{\nu}$ &$ \hspace{3mm}{\boldsymbol \xi}_{i} \neq {\boldsymbol \xi}_{j}$ \textnormal{for some }$1\leq i, j \leq \nu$ \\
 \textnormal{interfluid} & ${\boldsymbol \xi}_{\mathrm I} = {\boldsymbol \xi}_{\mathrm{II}}$ & ${\boldsymbol \xi}_{\mathrm I} \neq {\boldsymbol \xi}_{\mathrm{II}}$\\\bottomrule
\end{tabular}
\captionof{table}{Hypotheses}\label{table:hypotheses}
\par
\bigskip
%Should be a caption
\end{minipage}

%\begin{minipage}{\linewidth}
%\centering
%\captionof{table}{Test statistics}\label{table:test_statistics}
% \begin{tabular}{lc}\hline
%\textnormal{heterogeneity} & \textnormal{test statistic}            \\\hline
%\textnormal{intrafluid}   & $S^2_1, S^2_2 \textnormal{ (see \eqref{e:S-xihat})}$\\
%\textnormal{interfluid}   & $T_1, T_2 \textnormal{ (see \eqref{e:two_fluid_ts})}$\\\hline
%\end{tabular}
%\par
%\bigskip
%%Should be a caption
%\end{minipage}
$$
\vspace{-6mm}
$$
\begin{minipage}{\linewidth}
\centering
\setlength{\heavyrulewidth}{1.5pt}
\setlength{\abovetopsep}{4pt}
\begin{tabular}{lccc}\toprule
\textnormal{heterogeneity} & \textnormal{rejection region} & \textnormal{test statistic}   & \textnormal{number of paths} \\ \hline
\textnormal{intrafluid} & $R_{\textnormal{intra}} \textnormal{ (see \eqref{e:R-intra})}$ & $S^2_1, S^2_2 \textnormal{ (see \eqref{e:S-xihat})}$ & $\nu$\\
%\textnormal{intrafluid} & $R_{\textnormal{intra}} \textnormal{ (see \eqref{e:R-intra_large})}$ & $3/2 < \alpha < 2$ & large $\nu$\\
\textnormal{interfluid} & $R_{\textnormal{inter}} \textnormal{ (see \eqref{e:R-inter})}$ & $T_1, T_2 \textnormal{ (see \eqref{e:two_fluid_ts})}$& $\nu_{\mathrm{I}}$, $\nu_{\mathrm{II}}$ \\\bottomrule
\end{tabular}
\captionof{table}{Tests}\label{table:tests}
\par
\bigskip
%Should be a caption
\end{minipage}

%\begin{minipage}{\linewidth}
%\centering
%\captionof{table}{Tests}\label{table:tests}
% \begin{tabular}{lccc}\hline
%\textnormal{heterogeneity} & \textnormal{rejection} & \textnormal{parameter} & \textnormal{sample} \\ & \textnormal{region} & \textnormal{range} & \textnormal{size}\\\hline
%\textnormal{intrafluid} & $R_{\textnormal{intra}} \textnormal{ (see \eqref{e:R-intra})}$ & $0 < \alpha \leq 3/2$ & all $\nu$\\
%\textnormal{intrafluid} & $R_{\textnormal{intra}} \textnormal{ (see \eqref{e:R-intra_large})}$ & $3/2 < \alpha < 2$ & large $\nu$\\
%\textnormal{interfluid} & $R_{\textnormal{inter}} \textnormal{ (see \eqref{e:R-inter})}$ & $0 < \alpha < 2$ & large $\nu_{\mathrm{I}}$, $\nu_{\mathrm{II}}$ \\\hline
%\end{tabular}
%\par
%\bigskip
%%Should be a caption
%\end{minipage}
%\begin{tabular}{|r|l|}
%  \hline
%  7C0 & hexadecimal \\
%  3700 & octal \\ \cline{2-2}
%  11111000000 & binary \\
%  \hline \hline
%  1984 & decimal \\
%  \hline
%\end{tabular}

%\begin{equation}\label{e:H0_one_fluid}
%  H_{0}: {\boldsymbol \xi}_{1} = {\boldsymbol \xi}_{2}= \cdots = {\boldsymbol \xi}_{\nu}
%\end{equation}
%against the alternative heterogeneity hypothesis
%\begin{equation}\label{e:Ha_one_fluid}
%  H_{a}: \exists 1\leq i\neq j \leq \nu,\, {\mathrm{such \, that}}\, {\boldsymbol \xi}_{i} \neq {\boldsymbol \xi}_{j}.
%\end{equation}
%The distribution of ${\boldsymbol E}_i$ only depends on path length $N$, the chosen lags $\tau_k = w_k h$ ($k=1,\hdots,m$) and $\alpha$.

\noindent \textbf{Intrafluid heterogeneity}. Suppose $\nu \in \bbN$ bead diffusion paths of length $N$ from a single fluid sample are available. If the fluid is physically homogeneous, it is expected to generate particle paths with nearly identical parameter values ${\boldsymbol \xi}$. The alternative is that ${\boldsymbol \xi}_i \neq {\boldsymbol \xi}_j$ for some pair $i,j$, namely, their anomalous diffusion parameters differ. These two possibilities, labeled $H_0$ and $H_a$, respectively, are listed on the row ``intrafluid" in Table \ref{table:hypotheses}.

Starting from the $\nu $ particle paths, let
\begin{equation}\label{e:Ei,i=1,...,nu}
{\boldsymbol E}_i, \quad i=1,\hdots,\nu,
\end{equation}
be vector-valued estimators as in \eqref{e:hetero_hatxi}. For the purpose of constructing a test statistic, we need a normalized (standardized) estimator. Note that the variance of the GLS-type solution \eqref{e:GLS_solution} is given by
\begin{equation}\label{e:GLS_solution_variance}
(X^T \Upsilon^{-1}( {\boldsymbol \xi}) X)^{-1}
\end{equation}
(cf.\ \cite{christensen:2011}). So, define a standardized estimator by
\begin{equation}\label{e:xi_to_zeta_estimvar}
{\boldsymbol Z}_i = \left(
                                  \begin{array}{c}
                                    Z_{i,1} \\
                                    Z_{i,2} \\
                                  \end{array}
                                \right)
 = \Lambda^{-1/2}( A_{\ols,i}) {\boldsymbol E}_i,\quad i=1,\hdots,\nu,
\end{equation}
where
\begin{equation}\label{e:Lambda(Astand,i)}
\Lambda( A_{\ols,i}):= (X^T \Upsilon^{-1}( A_{\ols,i}) X)^{-1}
\end{equation}
and the variance estimator $\Upsilon( A_{\ols,i})$ is given by \eqref{e:Sigmatilde(alpha-hat)}. Then, \eqref{e:xi_to_zeta_estimvar} converges in distribution to $\nu$ independent and identically distributed normal random vectors with uncorrelated entries (see Proposition \ref{p:asympt_dist_estim}). So, for $\overline{Z}_{j} = \nu^{-1} \sum_{i=1}^{\nu} Z_{i,j}$, let
\begin{equation}\label{e:S-xihat}
S^2_j = \frac{1}{\nu-1} \sum_{i=1}^{\nu} (Z_{i,j} - \overline{Z}_{j})^2,\quad j=1,2,
\end{equation}
be the normalized and decorrelated sample variances of $\{Z_{i,j}\}_{i=1,\hdots,\nu}$, $j = 1,2$, as in \eqref{e:xi_to_zeta_estimvar}. Then, under $H_0$,
$$
\Big( (\nu-1)S^2_1, (\nu-1)S^2_2 \Big) \stackrel{d}\rightarrow ({\mathcal X}_1,{\mathcal X}_2), \quad  {\mathcal X}_j \sim \chi^2_{\nu-1}, \quad j = 1,2,
$$
as $N \rightarrow \infty$, where ${\mathcal X}_1$ and ${\mathcal X}_2$ are independent random variables. To test heterogeneity at significance level $\epsilon$, we can use Bonferroni-type correction (e.g., \cite{christensen:2011}, section 5.3) and reject the null hypothesis $H_0$ if
\begin{equation}\label{e:R-intra}
R_{\textnormal{intra}}: (\nu-1)S^2_1 > \chi^2_{\nu-1,\epsilon/2} \quad  \textnormal{or} \quad (\nu-1)S^2_2 > \chi^2_{\nu-1,\epsilon/2},
\end{equation}
where $\chi^2_{\nu-1,\epsilon/2}$ is a chi-square quantile (c.f.\ Table \ref{table:tests}, ``intrafluid" rows).

To check the size of the test \eqref{e:R-intra} over finite samples, we conducted a Monte Carlo study with 50 simulated paths of length $2^{12}$ and recorded whether or not the null hypothesis $H_0$ is rejected at $\epsilon = 0.05$ significance level. This procedure was repeated 2000 times. Since each outcome is a Bernoulli trial (reject or not $H_0$), the simulation rejection rate follows a binomial distribution with $n=2000$ and $p=0.05$. Thus, a normal approximation to the 95\% confidence interval of the rejection rate gives $(0.040,0.060)$. As shown in Figure \ref{fig:intra_inter_test sizes}, left plot, the observed simulation rejection rate was around 0.05 and within the 95\% confidence interval (for $0 < \alpha < 3/2$), as expected. Unreported computational experiments for different significance levels lead to analogous conclusions.

Figure \ref{fig:rej_rate_intra} displays Monte Carlo power curves for the intrafluid test. The study was conducted with a total of $\nu = \nu_1 + \nu_2$ paths, where $\nu_1$ and $\nu_2$ have diffusion exponents $\alpha_1 = 1$ and $\alpha_2$, respectively, and $\alpha_2 = 0.8$ (left plot) or $\alpha_2 = 0.7$ (right plot). In each plot, the $x$-axis represents the proportion of paths
\begin{equation}\label{e:nu2/nu1+nu2}
\frac{\nu_2}{\nu_1+\nu_2}
\end{equation}
with diffusion exponent $\alpha=\alpha_2$. The power curves quickly converge to 1 as a function of the ratio \eqref{e:nu2/nu1+nu2}, especially for the more distinguishable value $\alpha_2 = 0.7 < 1 = \alpha_1$.

%if $\nu$ is large enough (heuristically, $\nu \geq 30$),
%\begin{equation}\label{e:s2_Norm_asymp}
%  S^2_1 \overset{\mathrm{asymp}}{\sim} N\bigg( 1,\frac{\kappa_1 - 1}{\nu} \bigg),\quad
%  S^2_2 \overset{\mathrm{asymp}}{\sim} N\bigg( 1,\frac{\kappa_2 - 1}{\nu} \bigg),
%\end{equation}
%where the parameter $\kappa_1 := \bbE[\widehat{\zeta}_{\cdot,1} - \bbE \widehat{\zeta}_{\cdot,1}]^4$, $\kappa_2 := \bbE[\widehat{\zeta}_{\cdot,2} - \bbE \widehat{\zeta}_{\cdot,2}]^4$ can be numerically approximated (see Table \ref{table:test_statistics} and Appendix \ref{s:fourth_moment}). Thus, at significance level $\epsilon$, we reject the null hypothesis $H_0$ if
%\begin{equation}\label{e:R-intra_large}
%R_{\textnormal{intra}}: \frac{S^2_1 - 1}{\sqrt{(\kappa_1 - 1)/\nu}}> z_{\epsilon/2} \quad  \textnormal{or} \quad \frac{S^2_2 - 1}{\sqrt{(\kappa_2 - 1)/\nu}}> z_{\epsilon/2}
%\end{equation}
%(see Table \ref{table:tests}).

When $\alpha >3/2$, under $H_0$ the estimators \eqref{e:Ei,i=1,...,nu} converge in distribution to $\nu$ independent and identically distributed non-Gaussian random vectors. Hence, so do the estimators \eqref{e:xi_to_zeta_estimvar}. In this case, the marginal distributions of the decorrelated vector $((\nu-1)S^2_1,(\nu-1)S^2_2)$ do not approach chi-squared distributions. In computational experiments, the size of the intrafluid test \eqref{e:R-intra} did not significantly deviate from the 0.05 target for $\alpha = 1.6$, indicating that the nonstandard asymptotic behavior is not a concern for paths of length $2^{12}$. Deviation was significant for the extreme value $\alpha = 1.8$, suggesting that approximating the test size by simulation may be generally recommendable for greater accuracy (see Figure \ref{fig:intra_inter_test sizes}, left plot).\\

\begin{figure}[htbp]
\begin{center}
\includegraphics[scale=0.45]{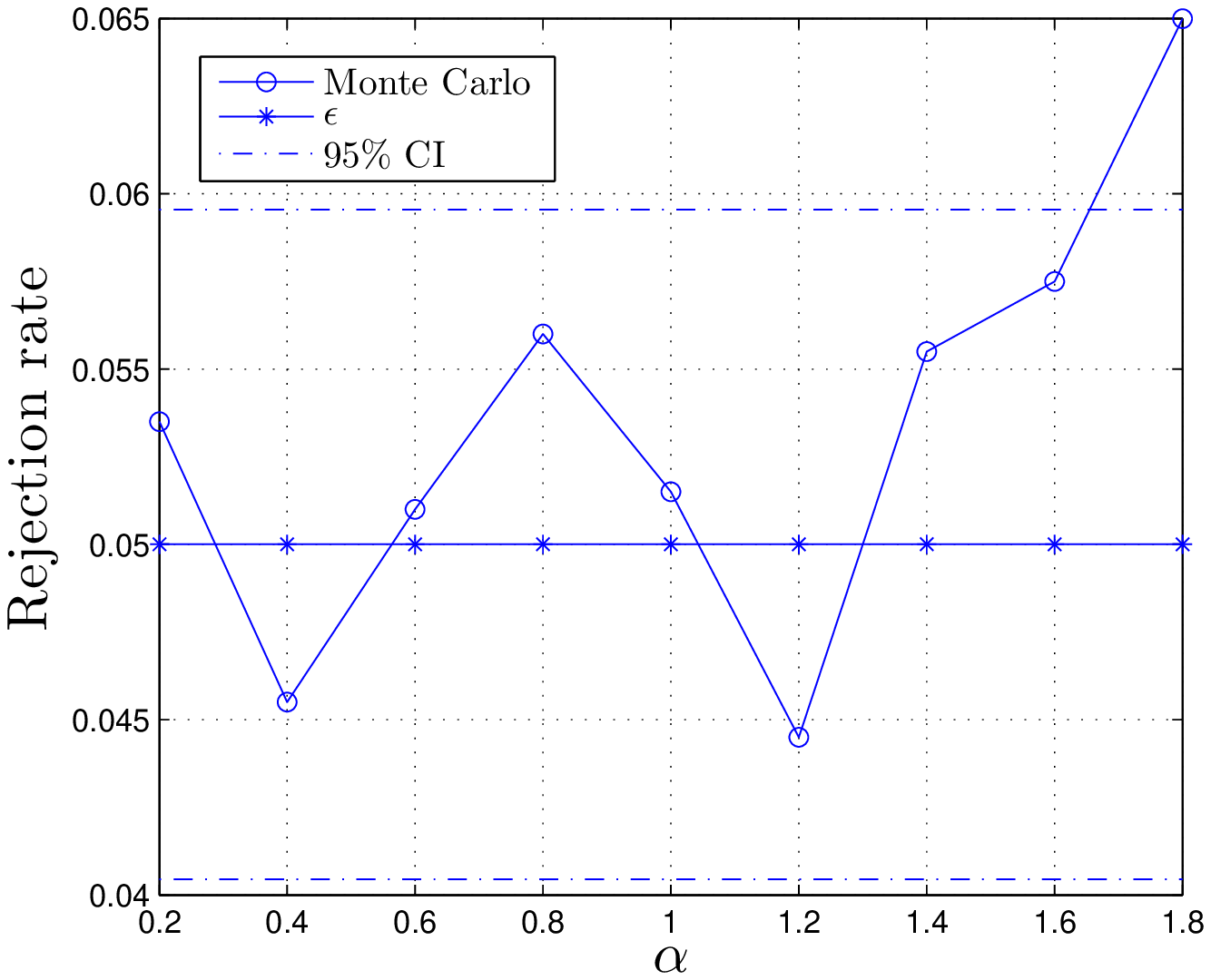}
\includegraphics[scale=0.48]{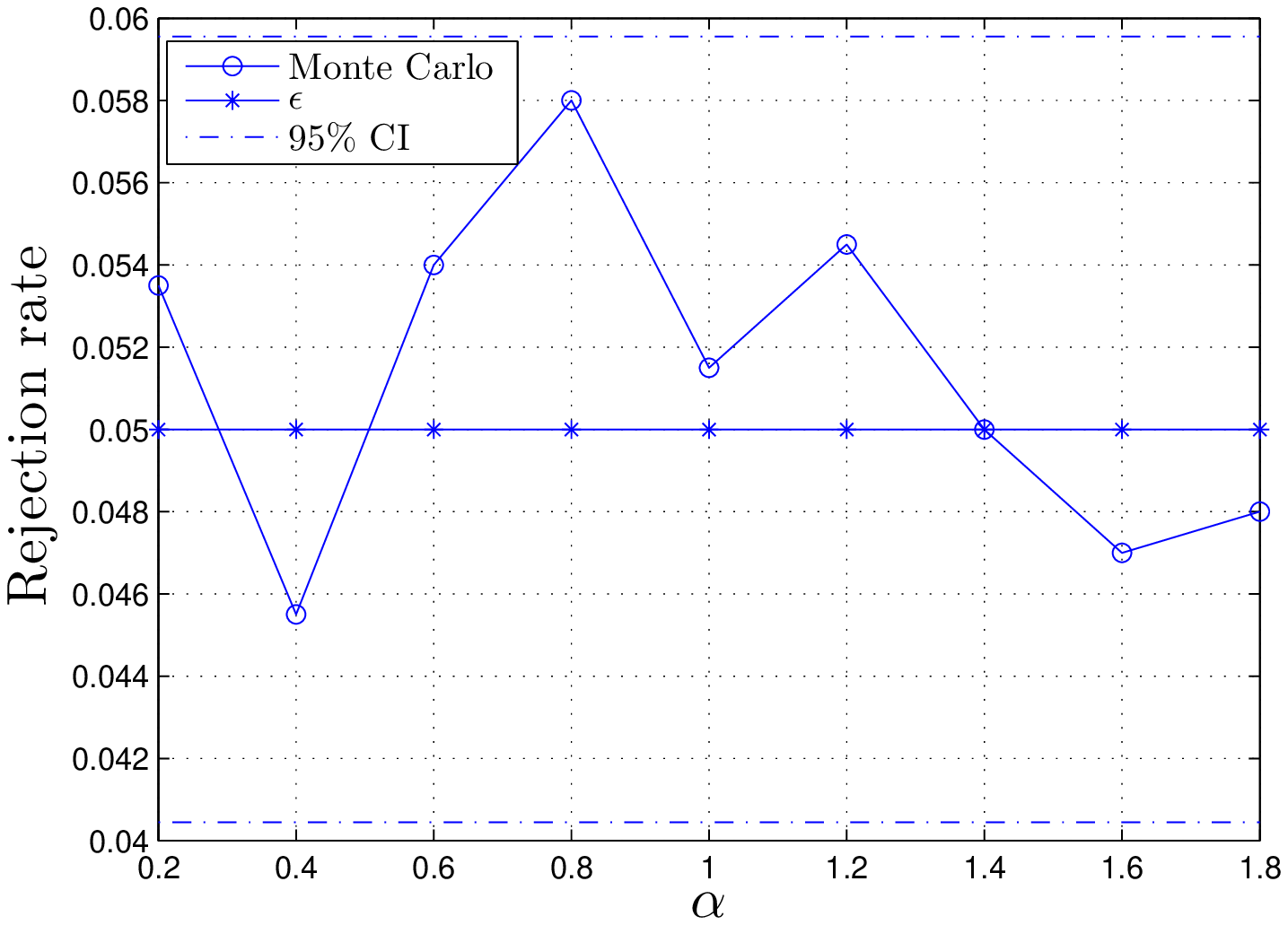}
\caption{\textbf{Intra- and interfluid heterogeneity: test sizes.} Monte Carlo sizes of the intrafluid test (left plot; see \eqref{e:R-intra}) and interfluid test (right plot; see \eqref{e:R-inter}) as a function of the diffusion exponent $\alpha$ ($x$-axis). For every value of $\alpha$, each of 2000 Monte Carlo runs consisted of generating 50 independent paths of length $2^{12}$ and conducting a test at $\epsilon = 0.05$ (see \eqref{e:test_size}) or, equivalently, 95\% confidence level. %The '-o' line is the Monte Carlo simulation result and the '-*' is the theoretical value, respectively. The 95\% confidence intervals for the simulated rejection rate are presented in two '-.' lines.
The Monte Carlo rejection rate is very close to the theoretical value of $\epsilon = 0.05$ for almost all values of $\alpha$.}
\label{fig:intra_inter_test sizes}
\end{center}
\end{figure}

\begin{figure}[htbp]
\begin{center}
\includegraphics[scale=0.45]{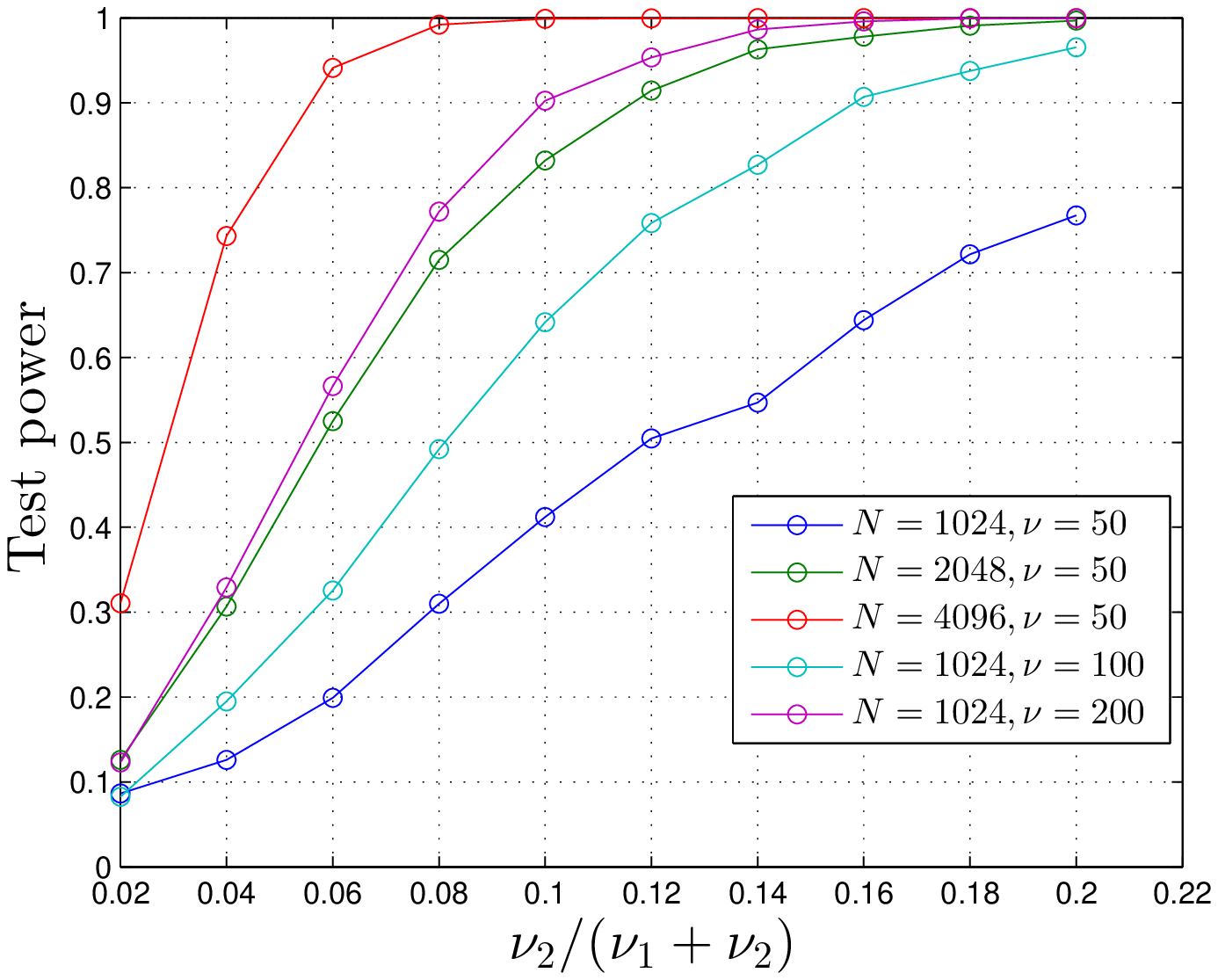}
\includegraphics[scale=0.5]{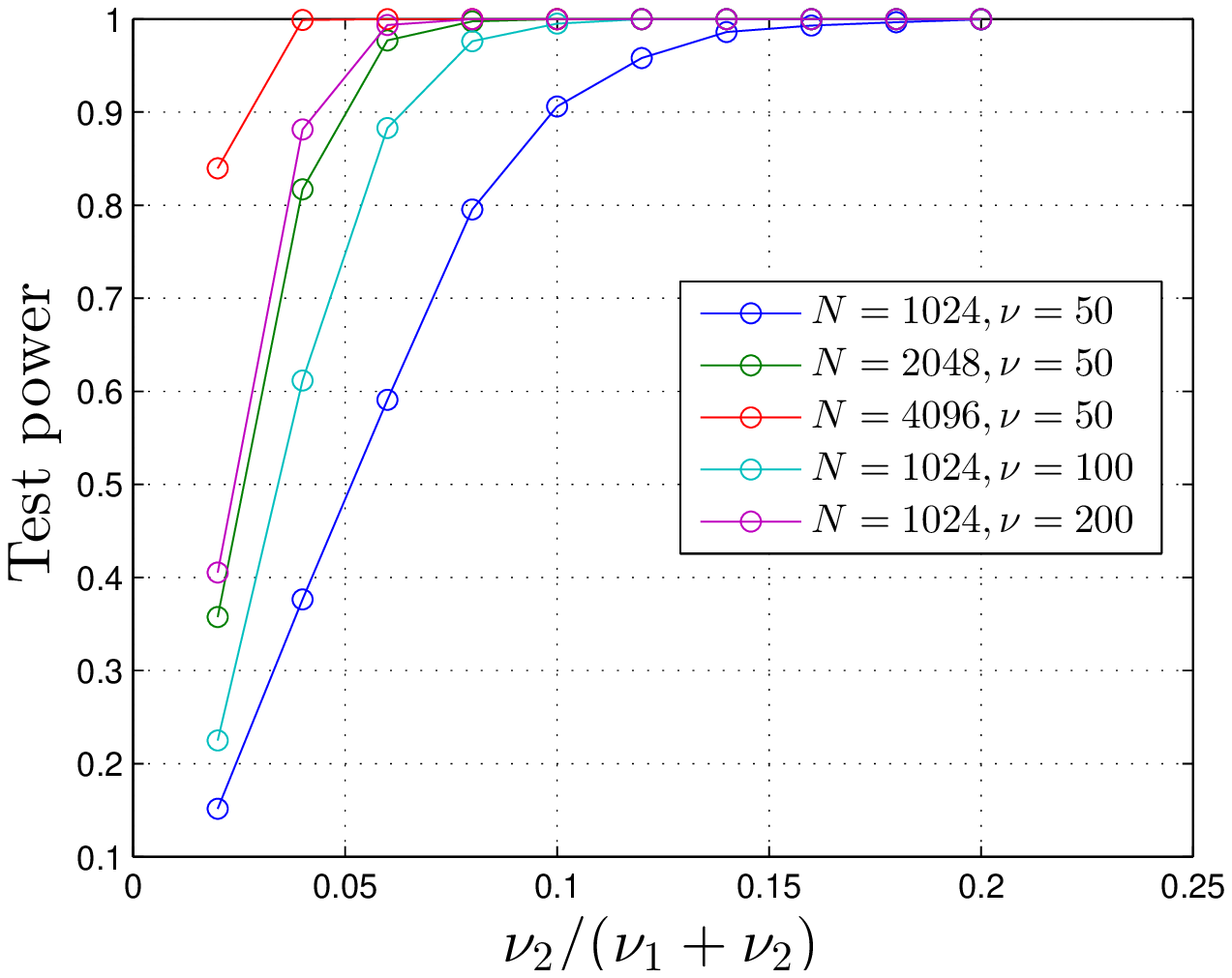}
\caption{\textbf{Intrafluid heterogeneity: test power.} Simulations were conducted with a total of $\nu = \nu_1 + \nu_2$ paths, where $\nu_1$ and $\nu_2$ of them displayed diffusion exponents $\alpha=\alpha_1 = 1$ and $\alpha = \alpha_2$, respectively. In each plot, the $y$-axis represents the observed test power, or rejection rates (see \eqref{e:test_power}), and the $x$-axis is the proportion of paths with $\alpha=\alpha_2$ (see \eqref{e:nu2/nu1+nu2}), starting at 0.02. The total number of Monte Carlo runs is 2000. Left plot: $\alpha_2 = 0.8$. Right plot: $\alpha_2 = 0.7$.}
\label{fig:rej_rate_intra}
\end{center}
\end{figure}

\noindent \textbf{Interfluid heterogeneity}. Now suppose $\nu_{\mathrm{I}}$ and $\nu_{\mathrm{II}}$ paths,
\begin{equation}\label{e:nuI,nuII}
\nu_{\mathrm{I}},\nu_{\mathrm{II}} \in \bbN,
\end{equation}
are obtained from two physically homogeneous fluid samples I and II, respectively. We are interested in testing whether the samples I and II are homogeneous, namely, whether or not particle diffusion in the fluid samples displays the same underlying parameter value $\boldsymbol \xi$. These two possibilities, labeled $H_0$ and $H_a$, respectively, are described on the row ``interfluid" in Table \ref{table:hypotheses}.

Since multiple (independent) particle paths are assumed available for each fluid sample, we can construct an estimator involving all available $\MSD$ terms. In fact, first define the overall average mean squared displacement over $\nu$ $\MSD$ terms (AMSD) by
\begin{equation}\label{e:hetero_emsd}
   M^*_N(\tau)=   \frac{1}{\nu} \sum^{\nu}_{\ell=1} M_N(\tau)_{\ell}.
\end{equation}
By independence,
\begin{equation}\label{e:EMSD_mean_variance}
\langle M^*_N(\tau)\rangle = \langle M_N(\tau)\rangle = \langle X^2(\tau)\rangle, \quad \var \hspace{0.5mm} M^*_N(\tau) = \frac{1}{\nu} \var \hspace{0.5mm}M_N(\tau).
\end{equation}
Then, $\AMSD$-type estimators
\begin{equation}\label{e:hetero_xi_emsd}
  {\boldsymbol E}^* = ( L^*, A^*)
\end{equation}
can be obtained by applying the pseudocode in Appendix \ref{s:pseudocode} after replacing $\MSD$ terms $M_N(\tau_k)$ with their $\AMSD$ counterparts $ M^*_N(\tau_k)$, $k = 1,\hdots,m$. Given two fluid samples I and II, let ${\boldsymbol E}^*_{\mathrm I}$ and ${\boldsymbol E}^*_{\mathrm{II}}$ be their respective $\AMSD$-type estimators. Their finite sample covariance matrices are given by $\nu^{-1}_{\mathrm{I}} \Lambda ({\boldsymbol \xi}_{\mathrm I})$ and $\nu^{-1}_{\mathrm{II}} \Lambda({\boldsymbol \xi}_{\mathrm{II}})$, respectively (cf.\ \eqref{e:GLS_solution_variance}). By analogy with \eqref{e:tildesigmak1k2e1} and \eqref{e:Lambda(Astand,i)}, we can define their $\AMSD$-type estimators
\begin{equation}\label{e:Sigma_xi(alpha_EMSD)}
\frac{1}{\nu_{\mathrm{I}}}\Lambda ( A^*_{\mathrm{I}}), \frac{1}{\nu_{\mathrm{II}}}\Lambda ( A^*_{\mathrm{II}}).
\end{equation}
Figure \ref{fig:var_alphat} displays a study of the accuracy of $\Lambda ( A^*_{\bullet})$ as an estimator. It plots Monte Carlo variances of the estimator ${\boldsymbol E}$ as well as their estimates $\Lambda ( A^*_{\bullet})$ for several values of the parameter $\alpha$. The latter nearly perfectly match the former in the subdiffusive range. A slight deviation appears in the strongly superdiffusive range, but still within an acceptable margin.

\begin{figure}[htbp]
\begin{center}
\includegraphics[scale=0.46]{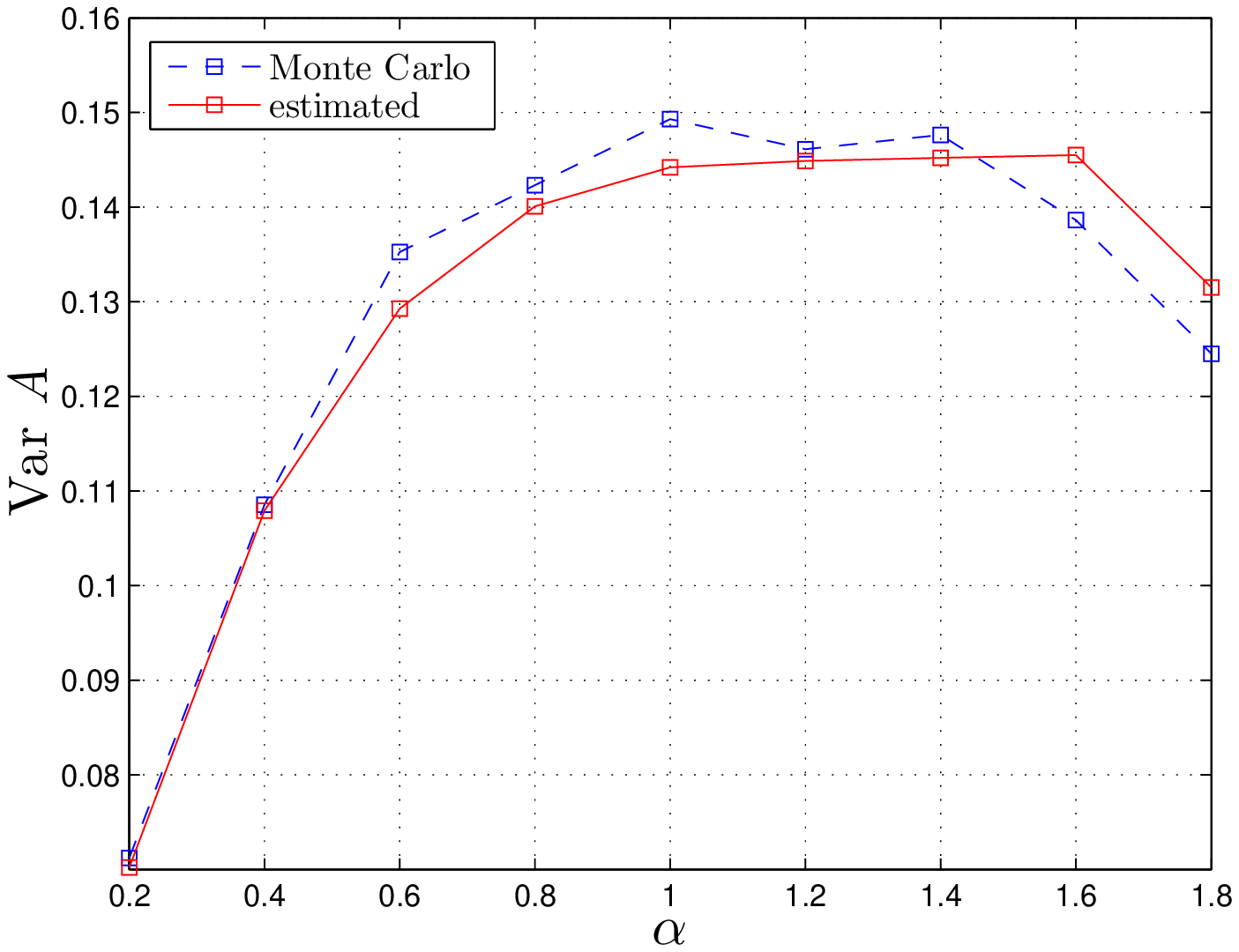}
\includegraphics[scale=0.5]{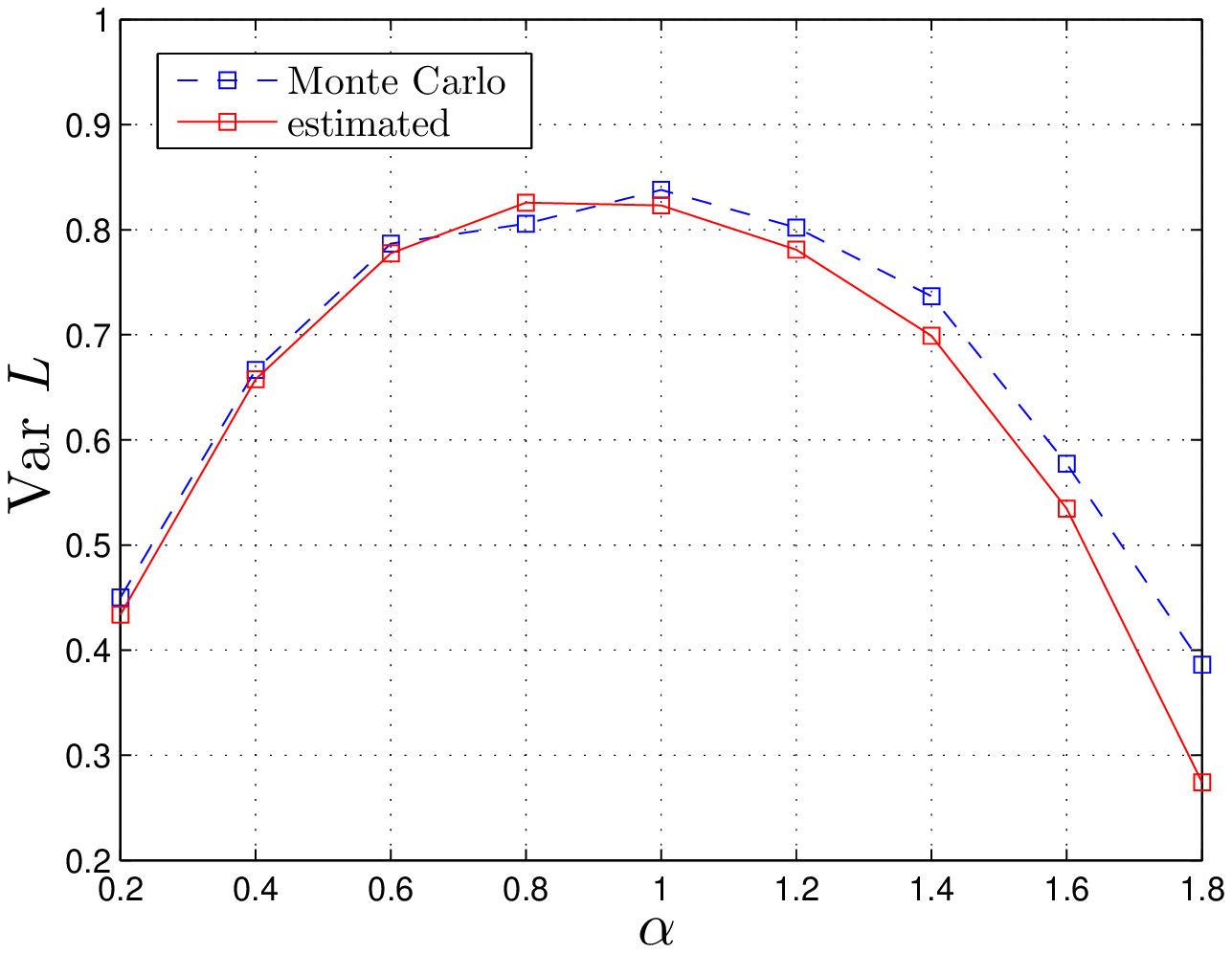}
\caption{\textbf{Comparative analysis} of the standard error of the estimator ${\boldsymbol E}=( L, A)$ (dashed line; see \eqref{e:hetero_hatxi}) and its $\AMSD$-type estimator $\Lambda ( A^*)$ (solid line; see \eqref{e:Sigma_xi(alpha_EMSD)}) as a function of the diffusion exponent $\alpha$. The latter closely matches the former, especially in the subdiffusive range $\alpha<1$. The number of Monte Carlo runs is 1000 based on particle paths of length $2^{10}$.}
\label{fig:var_alphat}
\end{center}
\end{figure}

Hence, we define the standardized estimators
$$
{\boldsymbol Z}^*_{j} = \left(
                                  \begin{array}{c}
                                    Z^*_{j,1} \\
                                    Z^*_{j,2} \\
                                  \end{array}
                                \right)
 = \sqrt{\nu_j}\Lambda^{-1/2}( A^*_{j}) {\boldsymbol E}^*_{j}, \quad j=\mathrm{I,II}.
$$
In view of Proposition \ref{p:asympt_dist_estim}, these estimators are also asymptotically normal for $0 < \alpha < 3/2$. Hence, let
\begin{equation}\label{e:two_fluid_ts}
T_1 = \frac{Z^*_{\mathrm{I},1} - Z^*_{\mathrm{II},1}}{\sqrt{2}},\quad
T_2 = \frac{Z^*_{\mathrm{I},2} - Z^*_{\mathrm{II},2}}{\sqrt{2}}
\end{equation}
be the associated test statistics. The rejection region is given by
\begin{equation}\label{e:R-inter}
R_{\textnormal{inter}}: \abs{T_1} >z_{\epsilon/4} \quad \mathrm{or} \quad \abs{T_2}>z_{\epsilon/4},
\end{equation}
where $z_{\epsilon/4}$ is a standard Normal quantile (c.f.\ Table \ref{table:tests}, row ``interfluid"). In \eqref{e:R-inter}, the probability $\epsilon/4$ stems, first, from applying a Bonferroni-type correction to a double testing region (hence yielding $\epsilon/2$ significance level in each), and second, from the fact that in each region the test statistic distribution is two-sided.

To check the test's size over finite samples, we produced a 2000-run Monte Carlo study based on two sets of 50 paths with the same diffusion exponent, where tests were conducted at significance level $\epsilon = 0.05$. As shown in Figure \ref{fig:intra_inter_test sizes}, right plot, the rejection rate was close to 0.05, as expected.

In Figure \ref{fig:test_power}, we investigate the interfluid test power as a function of the path lengths and number of paths. The $x$-axis represents the difference between the diffusion exponents from two fluids, namely,
\begin{equation}\label{e:delta-alpha}
\delta_{\alpha} = \abs{\alpha_{\mathrm I} - \alpha_{\mathrm{II}}} ,
\end{equation}
whereas the $y$-axis is the test power at $\epsilon = 0.05$. From top to bottom, the three plots correspond to $\alpha_{\min} = \min\{\alpha_{\mathrm I},\alpha_{\mathrm{II}}\} = 0.2, 1.0, 1.8$, respectively, for various combinations of realistic values of $N$ and $\nu = \nu_\textnormal{I} = \nu_\textnormal{II}$. In all cases, the power curves start at around 0.05, as expected, and quickly approach 1 as a function of $\delta_{\alpha}$ as defined in \eqref{e:delta-alpha}. Larger path lengths, larger number of particle paths as well as not very large values of $\alpha_{\min}$ are associated with faster convergence of power curves to 1.\\

\begin{figure}[htbp]
\begin{center}
\includegraphics[scale=0.5]{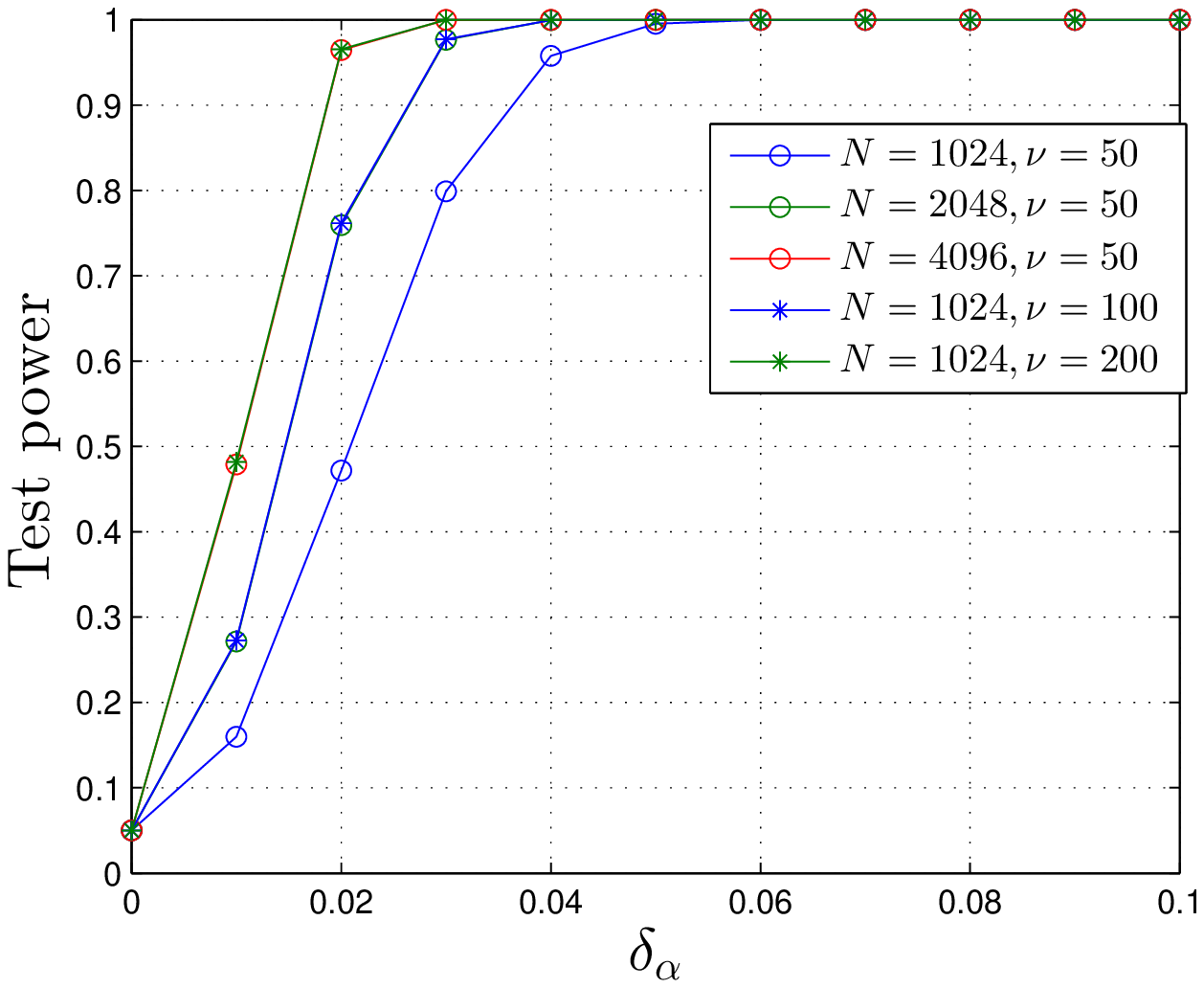}\\
\includegraphics[scale=0.5]{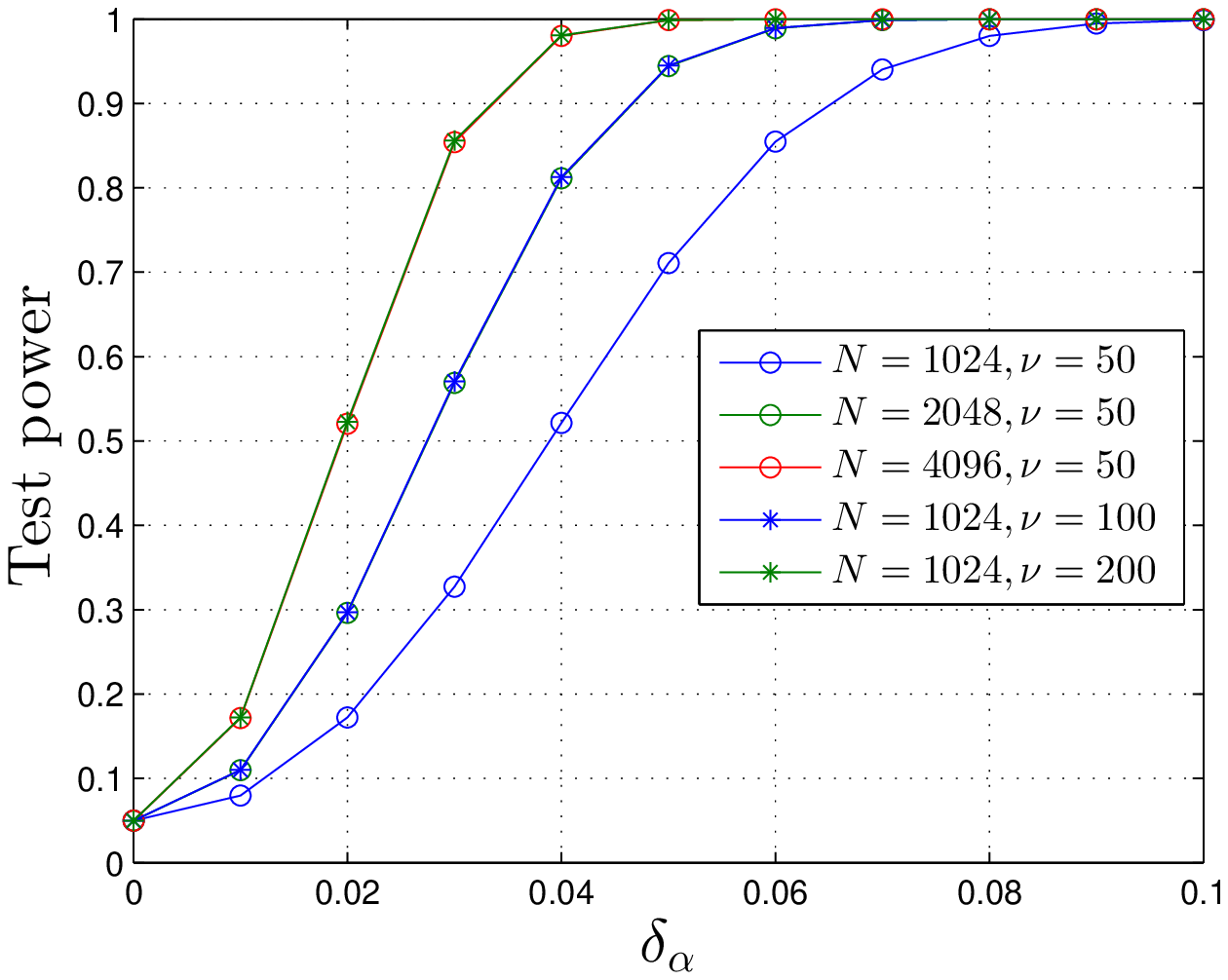}\\
\includegraphics[scale=0.5]{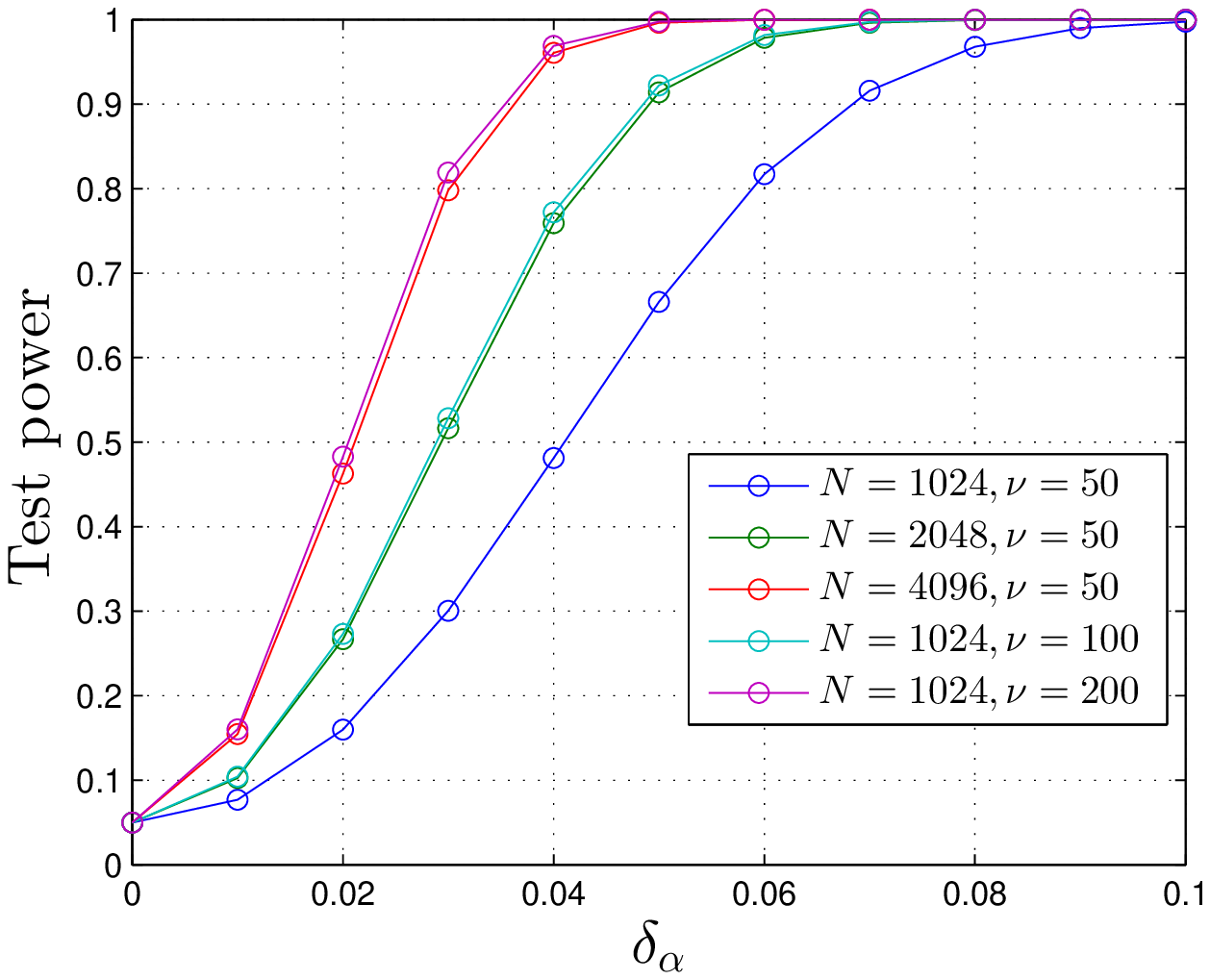}
\caption{\textbf{Interfluid heterogeneity: test power.} For various values of $N$ (see \eqref{e:MSD^}) and $\nu = \nu_\textnormal{I} = \nu_\textnormal{II}$ (see \eqref{e:nuI,nuII}), we plot test power ($y$-axis; see \eqref{e:test_power}) as a function of $\delta_{\alpha}$ ($x$-axis; see \eqref{e:delta-alpha}). The parameter values are given by $\alpha_{\min} := \min\{\alpha_{\mathrm I},\alpha_{\mathrm{II}}\}$ in the range $\alpha_{\min} = 0.2, 1.0, 1.8$ (top, middle and bottom, respectively).}
\label{fig:test_power}
\end{center}
\end{figure}

\noindent \textbf{Interfluid heterogeneity: more or longer paths under technical--experimental constraints?} Figure \ref{fig:test_power} also illustrates the following phenomenon. For the subdiffusive and diffusive cases, there is no visible difference between doubling the path lengths or the number of paths. However, in the strongly superdiffusive range, doubling the number of paths increases the test power more than doubling the path lengths.

In real world lab conditions, conducting single particle experiments involves coping with technical restrictions. For example, there may be limited camera recording time, tracer particles may slip out of the field of view or there may be a limit on the total number of tracer particles per fluid sample while still assuming that particles diffuse independently. So, assuming technical--experimental restrictions are in place, it is relevant to ask: what is the difference between
\begin{itemize}\vspace{-2mm}
\item [] \textit{Method I}: recording the movement of a larger number of particles ($\nu$) over a fixed period of time (hence, keeping constant the average sample path length $N$); and
\item [] \textit{Method II}: recording the same number of particles $\nu$ over a longer period of time (hence, yielding a larger average $N$)?
\vspace{-2mm}\end{itemize}
We answer this question in the framework of interfluid heterogeneity testing.

In the regimes of Methods I and II, we investigate the performance of the $\AMSD$-type estimator \eqref{e:hetero_xi_emsd} in terms of bias, standard deviation and square root MSE. Bearing in mind expression \eqref{e:EMSD_mean_variance}, by a similar reasoning to the one leading to expression \eqref{e:bias_log} for a single observed path, the bias of $\log_2 M^*_N(\tau)$ as an estimator of $\alpha \log \tau + \log \sigma^2$ is given by
\begin{equation}\label{e:ensemble_bias}
O(\tau^{-\delta}) - \bigg\langle  \frac{(M^*_N(\tau) - \langle X^2(\tau)\rangle)^2}{2 \langle X^2(\tau)\rangle^2} \bigg\rangle
 = O(\tau^{-\delta}) - \frac{1}{\nu}\bigg\langle \frac{(M_N(\tau) - \langle X^2(\tau)\rangle )^2}{2\langle X^2(\tau)\rangle^2}\bigg\rangle.
\end{equation}
Moreover, by the independence of particle paths, we can approximate the variance of $\log_2 M^*_N(\tau)$ by
\begin{equation}\label{e:ensemble_variance}
\Var \Big( \log M^*_N(\tau) \Big) \approx \frac{\var \hspace{1mm} M_N(\tau)}{\nu\langle X^2(\tau)\rangle^4}.
\end{equation}
The performance of the estimator \eqref{e:hetero_xi_emsd} in the two regimes depends on the interplay between the bias and variance components \eqref{e:ensemble_bias} and \eqref{e:ensemble_variance}, respectively. In a computational experiment, we applied the following procedure.
\begin{itemize}
  \item [1.] Start out in the same setting: $2^4$ paths of length $2^8$ for each method, run 500 Monte Carlo simulations to get the bias, standard deviation and square root MSE of $ A$ for Method I and II;
  \item [2.] for Method I, fix the path length and at each step generate $2^4$ times the previous number of paths and redo the Monte Carlos experiments;
  \item [3.] for Method II, fix the number of paths and at each step generate paths of length $2^4$ times the previous length, multiply all lags by 2 and redo the Monte Carlos experiments;
  \item [4.] repeat 2.\ and 3.\ three times.
\end{itemize}
For ease of comparison, Table \ref{t:MethodI_vs_MethodII} displays the multiple instances generated. Note that, at each step, the total number of points recorded
\begin{equation}\label{e:nu*N}
\nu \times N
\end{equation}
is identical for the two methods.

\begin{table}[!hbp]
\begin{center}
\begin{tabular}{ccccccc}\toprule
 &\multicolumn{3}{c}{Method I} & \multicolumn{3}{c}{Method II}\\
step & $N$ & $\nu$ &$ \nu \times N$ & $N$ & $\nu$ &$ \nu \times N$\\
\hline
1 & $2^8$ & $2^4$ & $2^{12}$ & $2^8$ &$2^4$ & $2^{12}$ \\
2 & $2^8$ & $2^8$ & $2^{16}$ & $2^{12}$ &$2^4$ & $2^{16}$ \\
3 & $2^8$ & $2^{12}$ & $2^{20}$ & $2^{16}$ &$2^4$ & $2^{20}$ \\
4 & $2^8$ & $2^{16}$ & $2^{24}$ & $2^{20}$ &$2^4$ & $2^{24}$ \\
\bottomrule
\end{tabular}
\caption{ Methods I and II. }
\label{t:MethodI_vs_MethodII}
\end{center}
\end{table}
%\begin{itemize}
%  \item [1.]  $2^4$ (paths, or $\nu$) $\times 2^8$ (nbp, or $N$) $= 2^{12}$ for both methods
%  \item [1.]  $2^8 \times 2^8 = 2^{16}$ for method I, $2^4 \times 2^{12} = 2^{16}$ for methods II (meanwhile multiply all lags by 2)
%  \item [1.]  $2^{12} \times 2^8 = 2^{20}$ for method I, $2^4 \times 2^{16} = 2^{20}$ for methods II (meanwhile multiply all lags by 4)
%  \item [1.]  $2^{16} \times 2^8 = 2^{24}$ for method I, $2^4 \times 2^{20} = 2^{24}$ for methods II (meanwhile multiply all lags by 8)
%\end{itemize}
%In step 1, both methods make use of a total of $16\times 256 = 4096$ observations. Then, in step 2, $4\times 16 = 64$ paths are generated under Method I, which yields a total of $64 \times 256 = 16384$ observations. In step 3, under Method II, 16 paths of length $4\times 256 = 1024$ are generated. Therefore, Method II also draws upon $16\times 1024 = 16384$ observations.

We compare the results in Figure \ref{fig:ensemble}, top and middle plots, where the diffusion exponent is set to $\alpha = 0.6$ and 1.0, respectively. Method II has smaller bias and square root MSE. The reason is that, when $\nu$ is large enough, the term $O(\tau^{-\delta})$ dominates the bias. Thus, increasing the number of paths $\nu$ does not reduce the bias. However, increasing the path length $N$ means that the $\MSD$ terms $M_N(\tau)$ with larger lag values $\tau$ can be used in the regression procedure. This implies a reduction in magnitude of the term $O(\tau^{-\delta})$, and hence, smaller bias. Method I displays smaller standard deviation because a 16-fold increase in $\nu$ reduces the standard error by a factor of $1/4$. Meanwhile, noting that $\langle X^2(\tau)\rangle \sim \sigma^2 \tau^{\alpha}$, expression \eqref{e:eta_zeta} implies that for Method II the standard deviation is proportional to $\sqrt{\tau/N}$. By multiplying $N$ by 16 and $\tau$ by 2, the standard error is reduced by a factor of $1/2\sqrt{2}$.

In Figure \ref{fig:ensemble}, bottom plot, we set $\alpha = 1.8$. For this parameter value, the convergence rate of the $\MSD$-based estimators is slower than $\frac{1}{\sqrt{N}}$. Method II still shows a smaller bias by comparison to Method I, as expected. However, since $\delta$ increases as a function of $\alpha$ (see expression \eqref{e:delta}), then $O(\tau^{-\delta})$ shrinks with $\alpha$. Therefore, the component $O(\tau^{-\delta})$ carries less weight in the estimator's bias for the superdiffusive case than for the subdiffusive case. Since $\alpha > 3/2$, again by expression \eqref{e:eta_zeta} (see also Remark \ref{r:on_3/2=<alpha<2}) the standard deviation for Method II is proportional to $(\tau/N)^{2-\alpha} = (\tau/N)^{0.2}$. Thus, again assuming a 16-fold increase in $N$ and a 2-fold increase in $\tau$, the standard error is reduced by a factor of $1/8^{0.2}$, which is much slower than the standard error reduction factor of $1/2\sqrt{2}$ for Method I. These are the two reasons why Method I displays smaller square root MSE than Method II.

\begin{figure}[htbp]
\begin{center}
\includegraphics[scale=0.5]{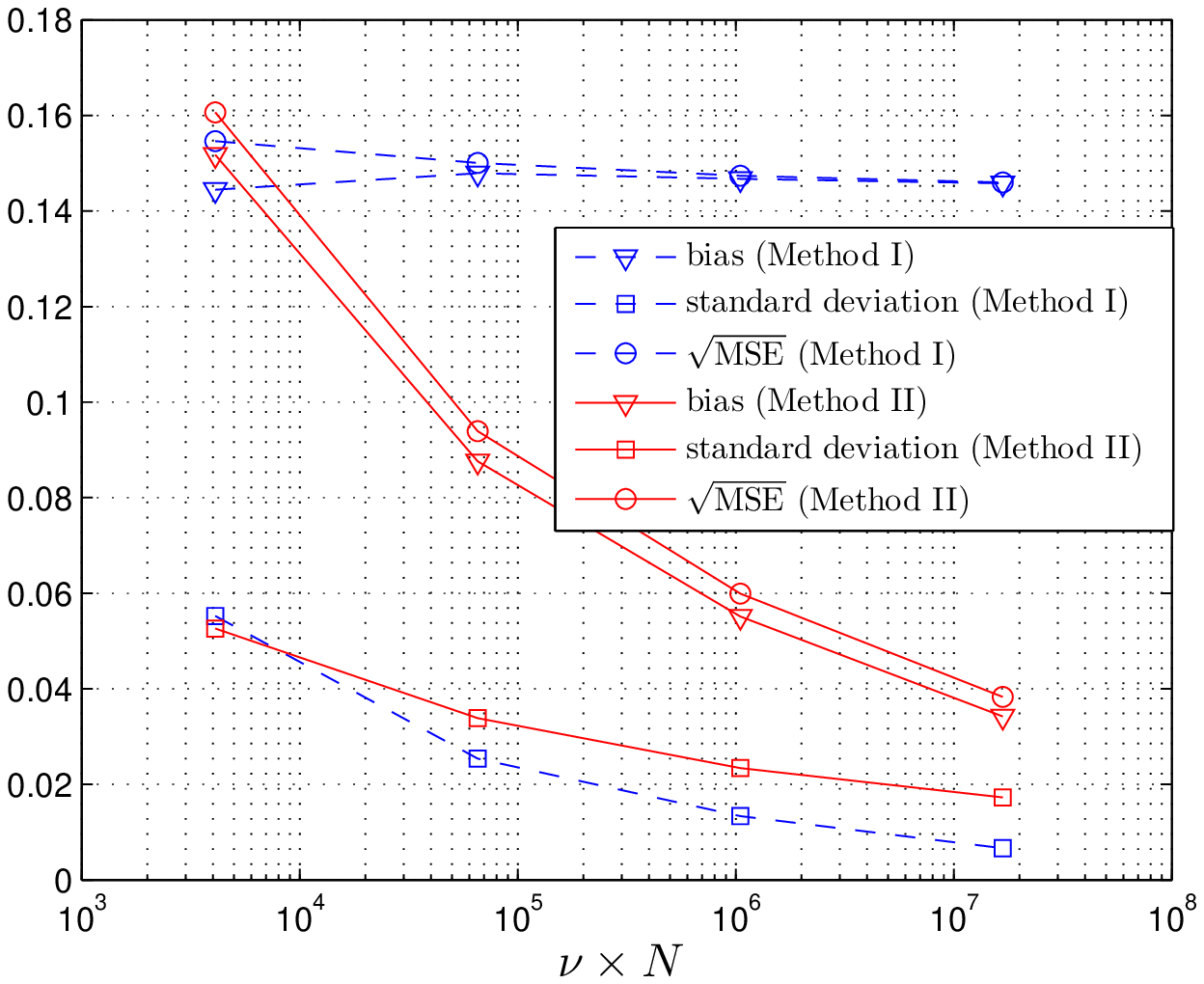}\\
\includegraphics[scale=0.5]{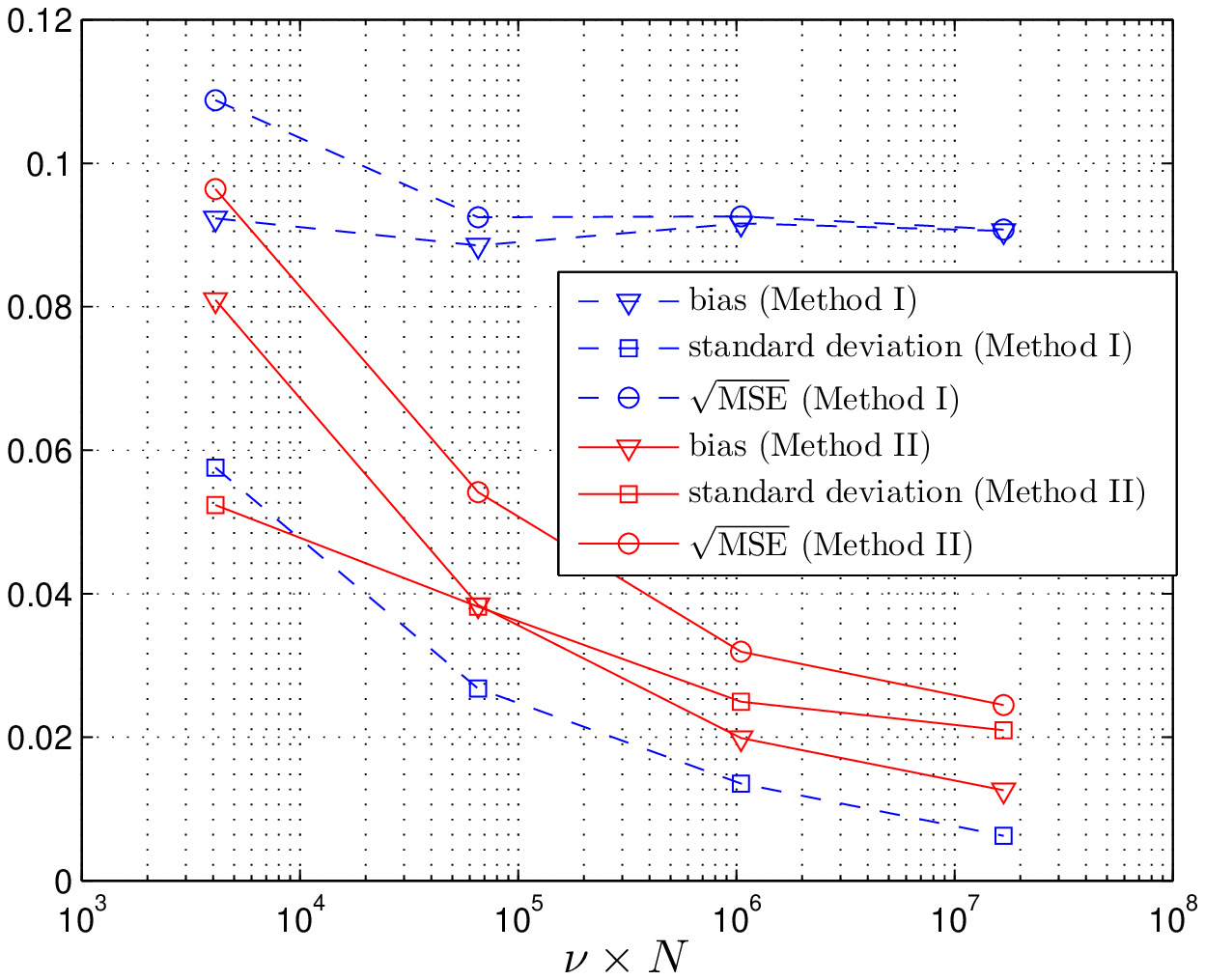}\\
\includegraphics[scale=0.5]{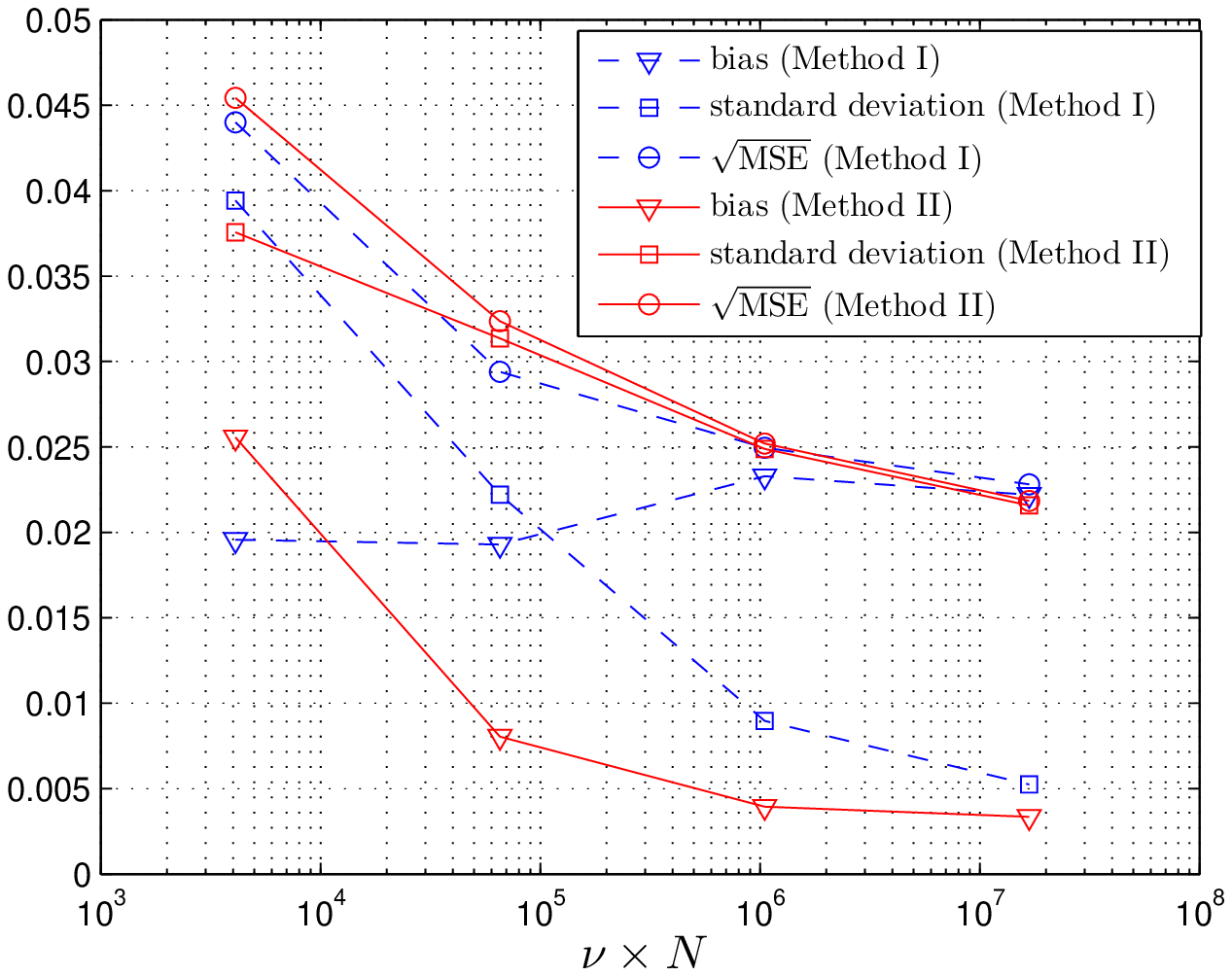}
\caption{\textbf{Bias, standard deviation, and square root MSE for Methods I (blue) and II (red).} The $x$-axis denotes the total number of recorded data points (see \eqref{e:nu*N}). Top: $\alpha = 0.6$. Middle: $\alpha = 1.0$. Bottom: $\alpha = 1.8$.}
\label{fig:ensemble}
\end{center}
\end{figure}

\section{Analysis of experimental data: heterogeneity of treated \emph{P.\ aeruginosa} biofilms}\label{s:experiments}

The Hill and Schoenfisch Labs at UNC-Chapel Hill produced data from experiments on disruption and eradication of \emph{P.\ aeruginosa} biofilms using nitric oxide-releasing chitosan oligosaccharides \cite{reighard:hill:dixon:worley:schoenfisch:2015}. For the reader's convenience, we provide a brief description of the experiments.

Cystic fibrosis (CF) lung disease is caused by defective chloride transport, resulting in thickened, dehydrated mucus. The latter restricts bacterial motility and promotes \emph{P.\ aeruginosa} biofilm formation. Inhaled tobramycin is currently the only antibiotic recommended for the treatment of both initial and chronic \emph{P.\ aeruginosa} infections in patients with CF. While inhaled tobramycin is effective at eradicating bacteria within biofilms, it fails to physically remove the structural remnants of the biofilm from the airways. This may lead to biofilm regrowth and the development of antibiotic-resistant infections. Therefore, an ideal anti-biofilm therapeutic for CF would both eradicate bacteria and physically degrade the biofilm, facilitating clearance from the airways.

On the other hand, nitric oxide (NO) is an endogenously produced diatomic free radical with significant antibacterial activity against \emph{P.\ aeruginosa} biofilms. Atomic force microscopy revealed that NO exposure causes structural damage to the membranes of planktonic Gram-negative bacteria, including \emph{P.\ aeruginosa}. The interest lies in the utility of NO-releasing chitosan oligosaccharides to both eradicate and physically alter \emph{P.\ aeruginosa} biofilms and in comparing its effect with tobramycin. In order to measure the physical changes to bacterial biofilms caused by NO-releasing chitosan oligosaccharides, movements of fluorescent tracer particles embedded in \textit{P.\ aeruginosa} biofilms were imaged at 60 frames per second for 30 seconds on an inverted microscope at $40 \times$ magnification.  The tracer particle displacement as a function of time was quantified using Video Spot Tracker software (Center for Computer Integrated Systems for Microscopy and Manipulation, University of North Carolina at Chapel Hill); see \cite{reighard:hill:dixon:worley:schoenfisch:2015} for more details.

Fluid heterogeneity has been correlated with increased viscoelasticity for complex biological materials such as sputum \cite{dawson:wirtz:hanes:2003}. In the experiments we describe, the effect of antibacterial treatment on biofilm heterogeneity was thus evaluated at different concentrations based on tracer particle displacement data. In Table \ref{tab:tobramycin}, we use the data to test the intrafluid heterogeneity of \emph{P.\ aeruginosa} biofilms after tobramycin treatments at concentrations levels 25, 50, 100, 200, and 400 $\mu$g ml${}^{-1}$. From each of these fluid samples, we randomly select 100 paths of length $N=1800$. An application of the intrafluid test \eqref{e:R-intra} produces strong evidence (negligible $p$-values) of intrafluid heterogeneity in every sample. This conclusion matches those reported in \cite{reighard:hill:dixon:worley:schoenfisch:2015}. Since no homogeneous fluid samples are detected from any of these five samples, we do not perform the interfluid heterogeneity test \eqref{e:R-inter}.

In Table \ref{tab:COS2NO}, we apply \eqref{e:R-intra} in the testing of intrafluid heterogeneity of \emph{P.\ aeruginosa} biofilms after COS2-NO treatment at concentration levels 1, 2, 4, 8, and 16 mg ml${}^{-1}$. COS2 releases NO, which cause the physical disruption and eradication of biofilms \cite{reighard:hill:dixon:worley:schoenfisch:2015}, and also reduces the viscoelastic properties of mucus \cite{reighard:ehre:rushton:ahonen:hill:schoenfisch:2017}. As before, 100 paths of length 1800 were randomly selected for each concentration level. At concentrations 1 or 2 mg ml${}^{-1}$, the $p$-values are still less than machine error, which indicates strongly significant heterogeneity. As the concentration level increases to 4 and 8 mg ml${}^{-1}$, the $p$-values also increase. At concentration level 16 mg ml${}^{-1}$, the $p$-value reaches 0.18. Hence, we fail to reject the null hypothesis of intrafluid homogeneity. This provides evidence that the COS2-NO treatment is effective at eradicating \emph{P.\ aeruginosa} biofilms. Once again, this analysis confirms the conclusions reported in \cite{reighard:hill:dixon:worley:schoenfisch:2015}. In Table \ref{tab:COS2NO_inter}, by applying \eqref{e:R-inter}, we test the interfluid heterogeneity of \emph{P.\ aeruginosa} biofilms after COS2-NO treatment at concentration level 16 mg ml${}^{-1}$. From each fluid sample (A, B and C), we selected 100 paths of length 1800 and conducted the test. It turns out that there is no evidence whatsoever of heterogeneity among fluid samples A, B and C.

\begin{table}
  \centering
  \begin{tabular}{cccc}
  \hline
  % after \\: \hline or \cline{col1-col2} \cline{col3-col4} \hdots
  Tobramycin ($\mu$g ml${}^{-1}$) &  $p$-value \\
  \hline
25  & $< 10^{-16}$\\
50  & $< 10^{-16}$\\
100 & $< 10^{-16}$\\
200 & $< 10^{-16}$\\
400 & $< 10^{-16}$\\
  \hline
    \end{tabular}
  \caption{Intrafluid biofilm heterogeneity testing after treatment with tobramycin at concentration levels 25, 50, 100, 200, 400 $\mu$g ml${}^{-1}$. 100 independent paths of length 1800 were randomly selected for each concentration level.}\label{tab:tobramycin}
\end{table}

\begin{table}
  \centering
  \begin{tabular}{cccc}
  \hline
  % after \\: \hline or \cline{col1-col2} \cline{col3-col4} \hdots
  COS2-NO (mg ml${}^{-1}$) & $p$-value \\
  \hline
1  & $< 10^{-16}$ \\
2 &  $< 10^{-16}$ \\
4 & $2 \times 10^{-13}$\\
8 & $3 \times 10^{-9}$\\
16 & 0.18\\
  \hline
    \end{tabular}
  \caption{Intrafluid biofilm heterogeneity testing after treatment with COS2-NO at concentration levels 1, 2, 4, 8, 16 mg ml${}^{-1}$. 100 independent paths of length 1800 were randomly selected for each concentration level.}\label{tab:COS2NO}
\end{table}

\begin{table}
  \centering
  \begin{tabular}{ccc}
  \hline
  % after \\: \hline or \cline{col1-col2} \cline{col3-col4} \hdots
  COS2-NO 16 mg ml${}^{-1}$    & $p$-value \\
  \hline
Group A vs Group B  & 0.9996 \\
Group A vs Group C  & 0.9998 \\
Group B vs Group C  & 0.9999 \\
  \hline
    \end{tabular}
  \caption{Interfluid biofilm heterogeneity testing after treatment with COS2-NO at concentration level 16 mg ml${}^{-1}$. 100 independent paths of length 1800 were randomly selected for each group.}\label{tab:COS2NO_inter}
\end{table}

\section{Conclusion}\label{s:conclusion}

%In this paper, we start from the asymptotic distribution of the TAMSD of nanometric tracer particles \cite{didier:zhang:2017} to construct statistical protocols for detecting physical fluid heterogeneity. The assumptions on particle motion cover a broad family of fractional Gaussian processes, including fractional Brownian motion and many instances of the generalized Langevin equation framework. The testing protocols allowed providing more accurate quantitative analysis of experimental data from the Hill and Schoenfisch Labs (UNC-Chapel Hill), and the results reported in \cite{reighard:hill:dixon:worley:schoenfisch:2015} were generally confirmed.

%Nevertheless, until recently, the limiting distribution of the $\MSD$ was unknown. This difficulty was overcome in \cite{didier:zhang:2017} for a broad class of Gaussian fractional stochastic processes. It was shown that the convergence in distribution of the pathwise $\MSD$ occurs at different rates, and that the limiting distribution may be Gaussian or non-Gaussian, all depending on the value of the diffusion exponent $\alpha$.

Motivated by applications in viscoelastic diffusion, in this paper we start from the $\MSD$'s asymptotic distribution for a broad class of Gaussian fractional stochastic processes \cite{didier:zhang:2017} to propose statistical protocols that make use of single-particle tracking data in the detection of fluid heterogeneity.

The testing methodology is based on an improved $\MSD$-type estimator. To construct this estimator, we tackle two of the main issues involved in $\MSD$-based estimation, namely, we mathematically characterize: \textbf{(a)} the finite-sample bias in log-$\MSD$-based methods; and \textbf{(b)} the effect of disturbance correlation. The theoretical results on \textbf{(a)} and \textbf{(b)} allow us to propose a nearly optimal estimator by combining a bias-correction procedure and a generalized least squares-type regression solution. The improved $\MSD$-based estimator \eqref{e:hetero_hatxi} is asymptotically normal for $0 < \alpha < 3/2$, and computational experiments show that the new estimator outperforms the standard $\MSD$-based estimator both in terms of bias and square root MSE for values of $\alpha$ over the whole parameter range $(0,2)$.

The estimator \eqref{e:hetero_hatxi} is used in the construction of protocols for fluid heterogeneity detection in two different experimental situations, namely, when testing: $(i)$ whether different regions of the same fluid are heterogeneous (\textit{intrafluid} heterogeneity); or $(ii)$ whether two samples from each homogenous fluid are heterogeneous (\textit{interfluid} heterogeneity). Reflecting the asymptotic behavior of \eqref{e:hetero_hatxi}, for $0 < \alpha < 3/2$ the test statistics \eqref{e:S-xihat} and \eqref{e:two_fluid_ts} for intra- and interfluid heterogeneity are asymptotically chi-square and asymptotically normally distributed, respectively. This ensures the tests and associated quantiles are asymptotically valid. Computational experiments confirm that the tests' significance levels are accurate over finite samples, and that the tests display high power even for relatively small deviations from the null hypotheses.

In all cases, for the sake of completeness, we discuss and provide computational studies on the strong superdiffusivity range $3/2 \leq \alpha < 2$. Although this may not affect physical areas of application where subdiffusion is prevalent, this research points to one difficulty involved in $\MSD$-based modeling, namely, a potentially non-Gaussian (Rosenblatt-type) limiting distribution with an intricate cumulant structure (cf.\ expression \eqref{e:3/4<H<1_multivariate_Rosen_limit}).

The constructed framework helps to shed light on the effect of common technical experimental constraints such as limited camera recording time: we characterize the difference between observing longer particle paths and using a larger number of particle paths of given length.

We apply the protocols in physical practice by making inferences on fluid viscoelasticity with data from the Hill and Schoenfisch Labs (UNC-Chapel Hill), as first reported and described in \cite{reighard:hill:dixon:worley:schoenfisch:2015}. The testing protocols reveal that COS2-NO treatment is effective in eradicating \textit{P.\ aeruginosa} biofilms, since greater concentration levels of the treatment clearly lead to greater fluid homogeneity as detected by tracer particle displacement data.

The research contained in this paper points to a number of interesting questions. From a modeling standpoint, it would be convenient to construct a heterogeneity testing framework for single particle experiments that, unlike $\MSD$-based methods, mathematically covers the full anomalous diffusion parametric range $0 < \alpha < 2$ under the same limit parametric family of distributions. Moreover, ideally a heterogeneity testing framework should be robust with respect to nuisance trends and added experimental noise. In another research direction, the same questions can be asked for classes of anomalous diffusion models not covered by the results in \cite{didier:zhang:2017} such as, for example, continuous time random walks and related stochastic processes.

%The proposed methods are asymptotically valid for $0 < \alpha < 3/2$, while computational experiments show that they also work reasonably well for $3/2 \leq \alpha < 2$ for realistic path lengths. The testing protocols allowed investigating and providing more accurate quantitative analysis of experimental data from the Hill and Schoenfisch Labs (UNC-Chapel Hill), and the results reported in \cite{reighard:hill:dixon:worley:schoenfisch:2015} were generally confirmed.

%
%(7) The conclusions section is probably the weakest part of this manuscript. Here, the authors should clearly list their main results and conclusions, as well as the methods how they were achieved. After this, the novelty of the current approach should be explained in detail, as compared to other (possible similar?) algorithms for SPT data analysis available in the literature. Lastly, some weaknesses of the algorithm and its applicability regimes (definitely, those exist) should be described and some perspectives for future research should be outlined.

\appendix

\section{The asymptotic distribution of the $\MSD$}\label{s:asympt_dist_MSD}

In Theorem \ref{t:asympt_dist_MSD}, we provide the asymptotic distribution of the $\MSD$ random vector after centering, which in Corollary \ref{c:asympt_dist_MSD} allows developing the asymptotic distribution of the OLS estimator \eqref{e:xi_OLS} of the diffusivity coefficient and diffusion exponent.

\begin{theorem}\label{t:asympt_dist_MSD}(Didier and Zhang \cite{didier:zhang:2017}, Theorem 1)
Suppose the particle motion is a Gaussian, stationary increment process whose covariance function admits the harmonizable (Fourier domain) representation
\begin{equation}\label{x_spec_rep}
\langle X(s) X(t)\rangle  =C^2_{\alpha} \int_{\bbR} \frac{(e^{isx}-1)(e^{-itx}-1)}{x^2}  f(x) dx, \quad \alpha\in (0,2), \quad C_{\alpha} \neq 0.
\end{equation}
In \eqref{x_spec_rep}, the spectral density has the generic form $f(x) = \frac{s(x)}{|x|^{\alpha/2-1/2}}$, where the $\bbC$-valued high frequency function $s(x)$ is bounded and satisfies the relations
\begin{equation}\label{e:s(x)}
\abs{s(0)}^2 = 1, \quad |\abs{s(x)}^2-1|\leq C_0 |x|^{\delta_0},\quad x\in (-\varepsilon_0,\varepsilon_0),
\end{equation}
for constants $C_0,\delta_0,\varepsilon_0>0$. Suppose the growth of the lag value term $\tau=\tau(N) \in \bbN $ as in \eqref{e:hk} with respect to the sample size $N$ is given by
\begin{equation}\label{e:h(n)}
\frac{\tau(N) \log^2(N)}{N}+\frac{N}{\tau(N)^{1+\delta/2}}\rightarrow 0,\quad N\rightarrow \infty,
\end{equation}
where
\begin{equation}\label{e:delta}
\delta=\min\{\alpha/2,\delta_0/2\}.
\end{equation}
Also consider the rates of convergence
\begin{equation}\label{e:eta_zeta}
\left\{\begin{array}{ccc}
0 < \alpha < 3/2: & \eta(N) = \sqrt{N},\, \zeta(\tau)=\tau^{\alpha+1/2} ;\\
\alpha = 3/2: & \eta(N) = \sqrt{N\log(N)},\, \zeta(\tau)=\tau^{2};\\
3/2 < \alpha < 2: & \eta(N) = N^{\alpha-1},\, \zeta(\tau)=\tau^{2}.
\end{array}\right.
\end{equation}
Then, as $N\rightarrow\infty$,
\begin{equation}\label{e:MSD_asymptotic_dist}
\bigg( \frac{N_k}{\eta(N_k)\zeta(\tau_k)}(M_N(\tau_k)- \langle X^2(\tau_k)\rangle )\bigg)_{k=1,\hdots,m}\stackrel{d}{\rightarrow} {\mathbf{Z}},
\end{equation}
where, for $k=1,\hdots,m$, $N_k = N - \tau_k$ is the number of available terms in each $\MSD$ sum \eqref{e:MSD^} and $\tau_k$ is given by \eqref{e:hk}. In \eqref{e:MSD_asymptotic_dist}, the distribution of the random vector ${\mathbf{Z}}$ can be described as follows.
\begin{itemize}
\item [($i$)] If $0 < \alpha < 3/2$, then ${\mathbf{Z}} \sim {\mathcal N}(0,\Sigma)$ (a $m$-variate Gaussian distribution), where the entry $k_1,k_2$ of the matrix $\Sigma = \Sigma(\alpha)$ is given by
\begin{equation}\label{e:Sigma_0<H<3/4}
\Sigma_{k_1, k_2} = 2 w_{k_1}^{-\alpha-1/2}w_{k_2}^{-\alpha-1/2} \Big(\frac{C_\alpha}{C_H}\Big)^4 \norm{\widehat{G}(y;w_{k_1},w_{k_2})}^2_{L^2(\bbR)},
\end{equation}
$k_1,k_2 = 1,\hdots,m$. In \eqref{e:Sigma_0<H<3/4}, we define
$$
C_H = \sqrt{\pi^{-1} H \Gamma(2H)\sin(H\pi)}
$$
and
\begin{equation}\label{e:G_hat}
\widehat{G}(y;w_{k_1},w_{k_2}) = C^2_H \frac{(e^{iw_{k_1} y}-1)(e^{-iw_{k_2}y}-1)}{\abs{y}^{\alpha+1}};
\end{equation}
\item [($ii$)] if $\alpha = 3/2$, then ${\mathbf{Z}} \sim {\mathcal N}(0,\Sigma)$, where the entry $k_1,k_2$ of the matrix $\Sigma= \Sigma(\alpha)$ is given by
\begin{equation}\label{e:Sigma_H=3/4}
\Sigma_{k_1,k_2}=4\vartheta^2, \quad k_1,k_2 = 1,\hdots,m,
\end{equation}
and
$$
\vartheta = \bigg(\frac{C_{\alpha}}{C_H}\bigg)^2 \frac{\alpha(\alpha-1)}{2};
$$
\item [($iii$)] if $3/2 < \alpha < 2$, $\mathbf{Z}$ follows a multivariate Rosenblatt-type distribution whose characteristic function is given by
\begin{equation}\label{e:3/4<H<1_multivariate_Rosen_limit}
    \phi_{\mathbf{Z}}({\mathbf{t}})=
        \exp\bigg\{\frac{1}{2}\sum_{s=2}^{\infty} \frac{[ 2i \vartheta \hspace{0.5mm}\sum^{m}_{k=1}t_k ]^s}{s}\hspace{1mm} c_s(\alpha)\bigg\}
\end{equation}
around the origin. In \eqref{e:3/4<H<1_multivariate_Rosen_limit}, for $s \geq 2$, $c_s(\alpha)$ is given by
\begin{equation}\label{e:cs}
     \int_{[0,1]^{s}} |x_1-x_2|^{\alpha-2}|x_2-x_3|^{\alpha-2}\cdots|x_s-x_1|^{\alpha-2}dx_1\cdots dx_s.
\end{equation}
\end{itemize}
\end{theorem}

Consider again the regression system \eqref{e:MSD_regression} and recall that expression \eqref{e:lm_estimator} gives the standard estimator generated by the OLS solution to the system \eqref{e:MSD_regression}. The following corollary describes the asymptotic distribution of the standard estimator \eqref{e:lm_estimator}.

\begin{corollary}(Didier and Zhang \cite{didier:zhang:2017}, Corollary 1)\label{c:asympt_dist_MSD} Suppose the assumptions of Theorem \ref{t:asympt_dist_MSD} hold. Then, as $N\rightarrow\infty$,
    \begin{equation}\label{e:lm_estimator_asymptotic}
    \frac{N \tau^{\alpha}}{\eta(N)\zeta(\tau)}
    \left(
      \begin{array}{c}
        \frac{1}{\log \tau }( L_{\ols}-{\log \sigma^2}) \\
         A_{\ols}-\alpha \\
      \end{array}
    \right)
    \stackrel{d}{\rightarrow}
    \left(
      \begin{array}{c}
        U^T \\
        -U^T \\
      \end{array}
    \right)
    A{\mathbf{Z}}.
    \end{equation}
In \eqref{e:lm_estimator_asymptotic},
\begin{equation}\label{e:A(theta,alpha)}
    A =A(\sigma^2,\alpha)=\diag( \zeta(w_1)/(\sigma^2 w_1^{\alpha}),\hdots,{\zeta(w_m)}/(\sigma^2 w_m^{\alpha})),
    \end{equation}
$\eta(\cdot)$, $\zeta(\cdot)$ and ${\mathbf{Z}}$ are as in Theorem \ref{t:asympt_dist_MSD}, and
    \begin{equation}\label{e:u_and_v}
    U^T=\frac{1}{c_w}\Big(\sum_{k=1}^{m} \log(w_k/w_1) , \hdots , \sum_{k=1}^{m} \log(w_k/w_m)\Big)
    \end{equation}
    with constant
    \begin{equation}\label{e:c_det_of_MTM}
    c_w=m \sum_{k=1}^{m} \log^2(w_k) - \Big(\sum_{k=1}^{m} \log(w_k)\Big)^2.
    \end{equation}
In particular, the standard estimator \eqref{e:lm_estimator} is consistent, namely, relation \eqref{e:consistency} holds.
\end{corollary}

\begin{proposition}\label{p:bound_ensemble_MSD}(Didier and Zhang \cite{didier:zhang:2017}, Proposition 1, $(i)$) Under the assumptions of Theorem \ref{t:asympt_dist_MSD}, there is a constant $\sigma^2 >0$ such that \eqref{e:bound_ensemble_MSD} holds for some $C > 0$, where $\delta>0$ is given by \eqref{e:delta}.
\end{proposition}

\section{Some lemmas}\label{sec:app2sec1}

In this section, we present some lemmas that are used to prove Theorems \ref{t:logmubias} and \ref{t:tildesigmak1k2} in Section \ref{s:bias_var_asympt_dist_stand_estimator}. Throughout this section, we assume $0 < \alpha < 3/2$ and the conditions of Theorem \ref{t:asympt_dist_MSD}. In proofs, whenever convenient $C$ denotes a constant that may change from one line to the next.\\

In light of \eqref{e:MSD_asymptotic_dist}, define the standardized statistic
\begin{equation}\label{e:varpi}
  \varpi(\tau) = \frac{M_N(\tau)}{\langle  X^2(\tau)\rangle }.
\end{equation}
In particular,
\begin{equation}\label{e:appB_varpi_def}
\varpi(\tau) \overset{P}{\rightarrow} 1, \quad N \rightarrow \infty,
\end{equation}
so that a Taylor expansion can be applied to $\log \varpi(\tau)$ around 1. Meanwhile, we define the standardized increment
\begin{equation}\label{e:w_j(h)}
  W_{j}(\tau) = \frac{X(j+\tau)-X(j)}{\sqrt{\langle X^2(\tau)\rangle }}.
\end{equation}
We will use the following results in our proofs. The first one is the classical Isserlis theorem, which reduces the higher moments of a multivariate normal vector to its second moments. The second one is a concentration inequality that will allow us to establish sharp bounds on the tails of centered quadratic forms.
\begin{theorem}[Isserlis, \cite{isserlis1916}]\label{t:Isserlis}
  Let $(Z_1,Z_2,\hdots,Z_{2N})$ be a zero mean, multivariate normal random vector. Then,
  $$
  \langle Z_1Z_2 \hdots Z_{2N}\rangle  = \sum \prod \langle Z_i Z_j \rangle ,
  $$
  where the notation $\sum \prod$ stands for summing over all distinct ways of partitioning $Z_1,\hdots, Z_{2N}$ into pairs $Z_i, Z_j$ and each summand is a product of these $N$ pairs.
\end{theorem}

\begin{theorem}\cite{laurent:massart:2000,boucheron:lugosi:massart:2013}\label{thm:laurent:massart:2000}
  Let $Z_1,\hdots,Z_N \overset{\mathrm{i.i.d.}}{\sim} {\mathcal N}(0,1)$ and consider constants $\eta_1,\hdots,\eta_N \geq 0$, not all zero. Let $\norm{\boldsymbol \eta}_2$ and $\norm{\boldsymbol \eta}_{\infty}$ be the Euclidean square and sup norms of the vector $\boldsymbol \eta = (\eta_1,\hdots,\eta_N)^T$. Also, define the random variable $X= \sum_{i=1}^{N} \eta_{i,N}(Z_i^2-1)$. Then, for every $x>0$,
  $$
  \bbP(X\geq 2\norm{\boldsymbol \eta}_2 \sqrt{x} + 2\norm{\boldsymbol \eta}_{\infty} x) \leq \exp(-x),
  $$
  $$
  \bbP(X \leq -2 \norm{\boldsymbol\eta}_{2} \sqrt{x}) \leq \exp(-x).
  $$
\end{theorem}

The following lemma describes some basic properties of the central moments of \eqref{e:varpi}.
\begin{lemma}\label{lem:kappa>=2}
As $N\rightarrow\infty$,
\begin{equation}\label{e:varpikappa>=2}
   \langle (\varpi(\tau) - 1)^{2}\rangle  = O\Big(\frac{\tau}{N}\Big).
\end{equation}
Moreover, any moment of $\varpi(\tau) - 1$ is bounded in $N$, i.e.,
\begin{equation}\label{e:bounded_moments}
   \Big| \langle (\varpi(\tau) - 1)^{\kappa}\rangle  \Big| = O(1), \quad \kappa \in \bbN.
\end{equation}
%\GDcomment{Do we need to assume that $\alpha < 3/2$?}
\end{lemma}
\begin{proof}
Expression \eqref{e:bounded_moments} (for $\kappa \geq 3$) can be proved by adapting the argument for establishing expression (C.24) in \cite{wendt:didier:combrexelle:abry:2017}, while making use of the bound \eqref{e:|eta|_sup=<|eta|_2=<Ch/n} and Lemma \ref{lem:P(varpi<r)} (expressions \eqref{e:pvarpi<r} and \eqref{e:pvarpi>r'}). So, for the reader's convenience, we establish \eqref{e:varpikappa>=2} (for $\kappa = 2$). The left-hand side of \eqref{e:varpikappa>=2} can be rewritten as
  $$
  \frac{1}{N^2} \bigg\langle  \sum_{k_1 = 1}^{N} \sum_{k_2 = 1}^{N} \bigg(\frac{(X(\tau+k_1)-X(k_1))^2}{\langle X^2(\tau)\rangle } - 1\bigg)
  \bigg(\frac{(X(\tau+k_1)-X(k_2))^2}{\langle  X^2(\tau)\rangle } - 1\bigg) \bigg\rangle
  $$
  \begin{equation}\label{e:varpikappa=2:2}
 = \frac{1}{N^2} \bigg\langle \sum_{k_1 = 1}^{N} \sum_{k_2 = 1}^{N} (W_{k_1}^2(\tau)-1)(W_{k_2}^2(\tau)-1) \bigg\rangle .
  \end{equation}
  By applying the Isserlis theorem (Theorem \ref{t:Isserlis}),
  $$
  \langle (W_{k_1}^2(\tau)-1)(W_{k_2}^2(\tau)-1)\rangle  = \langle W_{k_1}^2(\tau) W_{k_2}^2(\tau)\rangle  - \langle W_{k_1}^2(\tau)\rangle  - \langle W_{k_2}^2(\tau)\rangle  + 1
  $$
  $$
  = 2 \langle W_{k_1}(\tau) W_{k_2}(\tau)\rangle^2 = \frac{2}{\langle X^2(\tau)\rangle^2} \gamma_\tau^2(k_1 - k_2),
  $$
  where
  $$
  \gamma_h(k_1 - k_2) = \langle (X(k_1 + h) - X(k_1))(X(k_2 + h) - X(k_2))\rangle .
  $$
  Thus, \eqref{e:varpikappa=2:2} can be recast as
  \begin{equation}\label{e:varpikappa=2:3}
    \frac{2}{N^2 \langle X^2(\tau)\rangle^2} \sum_{k_1,k_2 = 1}^{N} \gamma_\tau^2(k_1 - k_2),
  \end{equation}
  Note that
  $$
  \frac{2}{N^2 \langle X^2(\tau)\rangle ^2} = O\Big(  \frac{1}{N^2 \tau^{2\alpha}} \Big)  = O \Big( \frac{\tau}{N}\zeta^{-2}(\tau) \eta^{-2}(N)  \Big),
  $$
  where $\zeta(\tau),\eta(N)$ are defined by \eqref{e:eta_zeta}. Then, by \eqref{e:varpikappa=2:3}, Lemmas C.1 -- C.4 in \cite{didier:zhang:2017}, expression \eqref{e:varpikappa=2:2} is of the order $O\Big(\frac{\tau}{N}\Big)$, as claimed.
\end{proof}

The next lemma draws upon Theorem \ref{thm:laurent:massart:2000} and Lemma \ref{lem:kappa>=2} to construct a concentration inequality for \eqref{e:varpi} (see also \cite{gajda:wylomanska:kantz:chechkin:sikora:2018}).
\begin{lemma}\label{lem:P(varpi<r)}
  Fix $-\infty < r <1/2 < 3/2 < r'$. Then, for any $0 < \xi < 1/2$ and some $C>0$,
  \begin{equation}\label{e:pvarpi<r}
  \bbP(\varpi(\tau)\leq r) \leq \exp \Big\{  -C \Big(\frac{N}{\tau}\Big)^{\xi}\Big\}
  \end{equation}
and
\begin{equation}\label{e:pvarpi>r'}
  \bbP(\varpi(\tau)\geq r') \geq \exp \Big\{  -C \Big(\frac{N}{\tau}\Big)^{\xi}\Big\}.
  \end{equation}
\end{lemma}
\begin{proof}
Let $W_j(\tau), j=1,\hdots,N$ be as in \eqref{e:w_j(h)}. Then, for $\varpi(\tau)$ as in \eqref{e:appB_varpi_def}, we can write
$$
\varpi(\tau) = \frac{1}{N} \sum_{j=1}^{N} W^2_{j}(\tau) = \frac{1}{N} \mathbf{W}_N^T \mathbf{W}_N,
$$
where $\mathbf{W}_N = (W_{1}(\tau),\hdots,W_{N}(\tau))^T$ is a multivariate normal vector with covariance matrix $\Gamma$. Consider the spectral decomposition $Q \Lambda Q^T = \Gamma$, where $Q$ is a $N\times N$ orthogonal matrix and $\Lambda=\diag\{\lambda_1,\hdots,\lambda_N\}$ is a $N\times N$ diagonal matrix. Then, $\mathbf{W}_N \overset{d}{=} Q \Lambda^{1/2} \mathbf{Z}_N$, where $\mathbf{Z}_N = (Z_1,\hdots,Z_N)^T {\sim} {\mathcal N}(0,I_N)$, $I_N$ is the $N\times N$ identity matrix, and $\overset{d}{=}$ denotes equality in distribution. Therefore,
\begin{equation}\label{e:sumetaZ^2}
  \varpi(\tau) \overset{d}{=}\frac{1}{N} ( Q \Lambda^{1/2} \mathbf{Z}_N)^T  Q \Lambda^{1/2} \mathbf{Z}_N = \frac{1}{N} \mathbf{Z}_N^T \Lambda \mathbf{Z}_N = \sum_{j=1}^{N} \eta_{j,N} Z_j^2,
\end{equation}
where $\eta_{j,N} = \frac{\lambda_j}{N}$. Let $\boldsymbol \eta_N = (\eta_{i,N})_{i=1,\hdots,N}$ be the vector of coefficients $\eta_{\cdot,N}$. By expression \eqref{e:varpikappa>=2} in Lemma \ref{lem:kappa>=2},
\begin{equation}\label{e:|eta|_sup=<|eta|_2=<Ch/n}
\norm{\boldsymbol \eta_N}_\infty^2 \leq \norm{\boldsymbol \eta_N}_2^2 = \var\, \varpi(\tau) = \langle (\varpi(\tau) - 1)^2 \rangle  = O\Big(\frac{\tau}{N}\Big).
\end{equation}
By Theorem \ref{thm:laurent:massart:2000}, by using the same argument as in the proof of Lemma C.3 in \cite{wendt:didier:combrexelle:abry:2017}, and applying the bound \eqref{e:|eta|_sup=<|eta|_2=<Ch/n},
$$
\bbP(\varpi(\tau)\leq r) = \bbP\Big(\sum_{j=1}^{N} \eta_{j,N} (Z_j^2-1)\leq r-1 \Big) \leq \exp\bigg\{-\frac{C}{\norm{\boldsymbol \eta_N}_2^2}\bigg\}
\leq \exp\bigg\{-C \hspace{0.5mm}\frac{N}{\tau}\bigg\}
$$
for some $C>0$. Thus, \eqref{e:pvarpi<r} follows. To show \eqref{e:pvarpi>r'}, it suffices to adapt the proof of expression (C.34) in \cite{wendt:didier:combrexelle:abry:2017}. In fact, by \eqref{e:|eta|_sup=<|eta|_2=<Ch/n}, $0 < \xi < 1/2$ and Theorem \ref{thm:laurent:massart:2000},
$$
\bbP(\varpi(\tau)\geq r') = \bbP\Big(\sum_{j=1}^{N} \eta_{j,N} (Z_j^2-1)\geq r'-1 \Big)
$$
$$
\leq
\bbP\Big(\sum_{j=1}^{N} \eta_{j,N} (Z_j^2-1)\geq 2 \norm{\boldsymbol \eta_N}_2 \Big( \frac{N}{\tau}\Big)^{\frac{\xi}{2}} +  2 \norm{\boldsymbol \eta_N}_\infty \Big(\frac{N}{\tau}\Big)^{\xi} \Big) \leq \exp\Big\{ - C \Big(\frac{N}{\tau}\Big)^{\xi} \Big\}.
$$
\end{proof}

%The next statement is a consequence of Lemma \ref{lem:P(varpi<r)}.
%\begin{corollary}\label{coro:P(varpi<r)}
%Fix $-\infty < r <1/2 < 3/2 < r'$. For large enough $n\in \bbN$,
%$$
%\bbP(\varpi(\tau)\leq r) = o\Big[\Big(\frac{\tau}{N}\Big)^2 \Big], \quad \bbP(\varpi(\tau)\geq r') = o\Big[\Big(\frac{\tau}{N}\Big)^2 \Big].
%$$
%\end{corollary}

The following lemma is used in the proof of Lemma \ref{lem:s1s2s3}.
\begin{lemma}\label{lem:logvarpimoment}
Let $p\geq 1$, there is a constant $K_p$ only depending on $p$ such that
\begin{equation}\label{e:E|log_varpi(h)|^p=<Kp}
\langle \abs{\log \varpi(\tau)}^p \rangle \leq K_p
\end{equation}
\end{lemma}
\begin{proof}
By \eqref{e:sumetaZ^2}, $\varpi(\tau)$ is a nonnegative weighted sum of independent chi-squared random variables, where not all weights are zero. Then, relation \eqref{e:E|log_varpi(h)|^p=<Kp} is a consequence of expression (96) in \cite{moulines:roueff:taqqu:2007:spectral}, p.\ 184.
\end{proof}

The following lemma can be shown based on Lemma \ref{lem:P(varpi<r)} and an adaptation of the proof of expressions (C.38) and (C.39) in \cite{wendt:didier:combrexelle:abry:2017}, which pertains to higher order (cross)moments of wavelet variance terms. %\GDcomment{In the Hadamard paper, expressions (C.43) and (C.44), we use the facts that $\sum (\cdot)(\cdot)$ is convergent.}
\begin{lemma}\label{lem:kappa1kappa2}
  Let $\kappa_1,\kappa_2 \in \bbN \cup \{0\},\kappa_1+\kappa_2\geq 3$, and fix $0<r<1/2$. Then, as $N\rightarrow\infty$,
$$
  \langle (\varpi(\tau_1) - 1)^{\kappa_1}(\varpi(\tau_2) - 1)^{\kappa_2}\rangle  = O\Big[\Big(\frac{\tau}{N}\Big)^2\Big].
$$
\begin{equation}\label{e:varpivarpi}
\langle (\varpi(\tau_1) - 1)^{\kappa_1}(\varpi(\tau_2) - 1)^{\kappa_2} 1_{\{\min\{\varpi(\tau_1),\varpi(\tau_2)\}>r\}}\rangle  = O\Big[\Big(\frac{\tau}{N}\Big)^2\Big].
\end{equation}
%\GDcomment{do we need to assume $\alpha < 3/2?$}
\end{lemma}

Lemmas \ref{lem:Evarpi-1varpi-1}, \ref{lem:s1s2s3} and \ref{lem:elogvarpi}, stated and shown next, are used in the proofs of Theorems \ref{t:logmubias} and \ref{t:tildesigmak1k2}. The lemmas provide expressions for (cross)moments and (cross)moments of logarithms of the random variables \eqref{e:varpi} at different lag values.
\begin{lemma}\label{lem:Evarpi-1varpi-1}
$$
\langle (\varpi(\tau_{k_1}) - 1)(\varpi(\tau_{k_2}) - 1)\rangle  = \frac{1}{2 n} \sum_{i=-N+1}^{N-1} \bigg(1- \frac{\abs{i}}{N}\bigg) \times
$$
$$
\times \bigg\{ \abs{ \frac{i}{\sqrt{\tau_{k_1} \tau_{k_2}}} + \sqrt{\frac{\tau_{k_1}}{\tau_{k_2}}} }^{\alpha}- \abs{ \frac{i}{\sqrt{\tau_{k_1} \tau_{k_2}}} + \sqrt{\frac{\tau_{k_1}}{\tau_{k_2}}} - \sqrt{\frac{\tau_{k_2}}{\tau_{k_1}}}}^{\alpha}-
$$
\begin{equation}\label{e:varpi-1varpi-1}
  -\abs{ \frac{i}{\sqrt{\tau_{k_1} \tau_{k_2}}} }^{\alpha}
+ \abs{ \frac{i}{\sqrt{\tau_{k_1} \tau_{k_2}}} - \sqrt{\frac{\tau_{k_2}}{\tau_{k_1}}}}^{\alpha} \bigg\}^2(1 + O(\tau^{-\delta})).
\end{equation}
\end{lemma}
\begin{proof}
  For notational simplicity, assume $k_1=1$ and $k_2 = 2$. By \eqref{e:w_j(h)}, the left-hand side of \eqref{e:varpi-1varpi-1} can be rewritten as
  \begin{equation}\label{e:varpi-1varpi-1e2}
    \frac{1}{N^2} \sum_{j_1,j_2 = 1}^{N} \langle  (W^2_{j_1}(\tau_1) - 1)(W^2_{j_2}(\tau_2) - 1)\rangle  = \frac{1}{N^2} \sum_{j_1,j_2 = 1}^{N} \langle  W^2_{j_1}(\tau_1)W^2_{j_2}(\tau_2)\rangle  - 1.
  \end{equation}
  By Theorem \ref{t:Isserlis} (Isserlis),
  $$
  \langle  W^2_{j_1}(\tau_1)W^2_{j_2}(\tau_2) \rangle = \langle W^2_{j_1}(\tau_1) \rangle \langle W^2_{j_2}(\tau_2)\rangle  + 2 \langle W_{j_1}(\tau_1)W_{j_2}(\tau_2) \rangle ^2
  $$
  \begin{equation}\label{e:varpi-1varpi-1e3}
  = 1 + 2 \langle  W_{j_1}(\tau_1)W_{j_2}(\tau_2) \rangle^2.
  \end{equation}
By Lemma A.1 in \cite{didier:zhang:2017},
$$
\bigg\langle  \frac{(X(j_1 + \tau_1) - X(j_1))(X(j_2 + \tau_2)- X(j_2))}{\sqrt{\langle X^2(\tau_1)\rangle } \sqrt{\langle X^2(\tau_2)\rangle }} \bigg\rangle
$$
\begin{equation}\label{e:ex(s)x(t)}
   = \bigg\langle  \frac{(B_H(j_1 + \tau_1) - B_H(j_1))(B_H(j_2 + \tau_2) - B_H(j_2))}{\sqrt{\langle B_H^2(\tau_1)\rangle }\sqrt{\langle B_H^2(\tau_2)\rangle }}\bigg\rangle  (1+O(\tau^{-\delta})),
\end{equation}
  where $B_H$ is a standard fBm with Hurst parameter given by the relation \eqref{e:alpha=2H}. By \eqref{e:w_j(h)}, \eqref{e:ex(s)x(t)} and expression \eqref{e:fBm_cov} for the covariance function of fBm,
  $$
  \langle W_{j_1}(\tau_1)W_{j_2}(\tau_2)\rangle  = \frac{1}{2}\Big\{\abs{ \frac{j_1 - j_2}{\sqrt{\tau_{1} \tau_{2}}} + \sqrt{\frac{\tau_{1}}{\tau_{2}}} }^{\alpha}- \abs{ \frac{j_1 - j_2}{\sqrt{\tau_{1} \tau_{2}}} + \sqrt{\frac{\tau_{1}}{\tau_{2}}} - \sqrt{\frac{\tau_{2}}{\tau_{1}}}}^{\alpha}-
  $$
  \begin{equation}\label{e:varpi-1varpi-1e4}
    -\abs{ \frac{j_1 - j_2}{\sqrt{\tau_{1} \tau_{2}}} }^{\alpha}
    + \abs{ \frac{j_1 - j_2}{\sqrt{\tau_{1} \tau_{2}}} - \sqrt{\frac{\tau_{2}}{\tau_{1}}}}^{\alpha}\Big\}(1+O(\tau^{-\delta})).
  \end{equation}
Since $\langle  W_{j_1}(\tau_1)W_{j_2}(\tau_2)\rangle  = \langle W_{j_1+k}(\tau_1)W_{j_2+k}(\tau_2)\rangle $, then by expression \eqref{e:varpi-1varpi-1e3} we can rewrite \eqref{e:varpi-1varpi-1e2} as
  \begin{equation}\label{e:varpi-1varpi-1e22}
    \frac{1}{2N} \sum_{i=-N+1}^{N-1} \frac{1}{N}\sum_{j_1-j_2 = i, j_1,j_2 = 1}^{N} (2\langle W_{j_1}(\tau_1)W_{j_2}(\tau_2)\rangle )^2.
  \end{equation}
  Relation \eqref{e:varpi-1varpi-1} is now a consequence of \eqref{e:varpi-1varpi-1e4} and \eqref{e:varpi-1varpi-1e22}.
\end{proof}

\begin{lemma}\label{lem:s1s2s3}
\begin{equation}\label{e:s1s2s3}
  \langle \log\varpi(\tau_{k_1}) \log\varpi(\tau_{k_2}) \rangle = \langle  (\varpi(\tau_{k_1}) - 1)(\varpi(\tau_{k_2}) - 1)\rangle  + o\Big(\frac{\tau}{N}\Big).
\end{equation}
\end{lemma}
\begin{proof}
For notational simplicity, assume $k_1=1$ and $k_2 = 2$. Let
$$
S_1 = \langle \log \varpi(\tau_1) \log \varpi(\tau_2)\rangle  - \langle  \log \varpi(\tau_1) \log \varpi(\tau_2)1_{\{ \min\{\varpi(\tau_1),\varpi(\tau_2)\}>r\}}\rangle ,
$$
$$
S_2 = \langle \log \varpi(\tau_1) \log \varpi(\tau_2)1_{\{\min\{\varpi(\tau_1),\varpi(\tau_2)\}>r\}} \rangle -
$$
$$
-\langle  (\varpi(\tau_1) - 1)(\varpi(\tau_2) - 1) 1_{\{\min\{\varpi(\tau_1),\varpi(\tau_2)\}>r\}}\rangle ,
$$
$$
S_3 = \langle  (\varpi(\tau_1) - 1)(\varpi(\tau_2) - 1) 1_{\{\min\{\varpi(\tau_1),\varpi(\tau_2)\}>r\}}\rangle  - \langle  (\varpi(\tau_1) - 1)(\varpi(\tau_2) - 1)\rangle .
$$
Note that
$$
 \langle \log \varpi(\tau_1) \log \varpi(\tau_2)\rangle  = S_1 + S_2 + S_3.
$$
Therefore, establishing \eqref{e:s1s2s3} is equivalent to showing that $S_1 + S_2 + S_3 = o\Big(\frac{\tau}{N}\Big)$. It suffices to show that
\begin{equation}\label{e:max(|S1|,|S2|,|S3|)=bound}
\max\{|S_1|,|S_2|,|S_3|\} = o\Big(\frac{\tau}{N}\Big).
\end{equation}
Let $0<r<1/2$. We start off with $S_2$ by writing out the almost sure Taylor expansion
\begin{equation}\label{e:logvarpitaylor}
  \log \varpi(\tau) 1_{\{ \varpi(\tau)>r \}} = \bigg\{(\varpi(\tau) - 1) - \frac{1}{2}\bigg( \frac{\varpi(\tau) - 1}{\sigma^2_+(\varpi(\tau))}\bigg)^2 \bigg\} 1_{\{\varpi(\tau)>r\}},
\end{equation}
where $\sigma^2_+(\varpi(\tau)) \in [\min\{\varpi(\tau),1\}, \max\{\varpi(\tau),1\}]$. Then,
$$
\langle \log \varpi(\tau_1) \log \varpi(\tau_2) 1_{\{ \min\{\varpi(\tau_1),\varpi(\tau_2)\} >r \}} \rangle
$$
$$
= \langle (\varpi(\tau_1) - 1)(\varpi(\tau_2) - 1)1_{\{ \min\{\varpi(\tau_1),\varpi(\tau_2)\} >r \}}\rangle
$$
$$
-\frac{1}{2} \bigg\langle  (\varpi(\tau_1) - 1)\bigg( \frac{\varpi(\tau_2) - 1}{\sigma^2_+(\varpi(\tau_2))}\bigg)^2 1_{\{ \min\{\varpi(\tau_1),\varpi(\tau_2)\} >r \}} \bigg\rangle
$$
$$
-\frac{1}{2} \bigg\langle  \bigg( \frac{\varpi(\tau_1) - 1}{\sigma^2_+(\varpi(\tau_1))}\bigg)^2(\varpi(\tau_2) - 1) 1_{\{ \min\{\varpi(\tau_1),\varpi(\tau_2)\} >r \}} \bigg\rangle
$$
\begin{equation}\label{e:s2}
  +\frac{1}{4} \bigg\langle  \bigg( \frac{\varpi(\tau_1) - 1}{\sigma^2_+(\varpi(\tau_1))}\bigg)^2 \bigg( \frac{\varpi(\tau_2) - 1}{\sigma^2_+(\varpi(\tau_2))}\bigg)^2 1_{\{ \min\{\varpi(\tau_1),\varpi(\tau_2)\} >r \}} \bigg\rangle .
\end{equation}
The second, third and fourth terms can be bounded by a similar argument, so we only develop the latter. Recast
$$
\bigg( \frac{\varpi(\tau) - 1}{\sigma^2_+(\varpi(\tau))}\bigg)^2 1_{\{ \varpi(\tau)>r \}}
= \bigg( \frac{\varpi(\tau) - 1}{\sigma^2_+(\varpi(\tau))}\bigg)^2 \bigg( 1_{\{ 1/2>\varpi(\tau)>r \}} + 1_{\{ \varpi(\tau)\geq 1/2 \}}\bigg)
$$
\begin{equation}\label{e:varpitheta+}
  \leq \bigg( \frac{\varpi(\tau) - 1}{r}\bigg)^2 1_{\{ 1/2>\varpi(\tau)>r \}} + \bigg( \frac{\varpi(\tau) - 1}{1/2}\bigg)^2 1_{\{ \varpi(\tau)\geq 1/2 \}}.
\end{equation}
Therefore, we can rewrite the fourth term in \eqref{e:s2} as
$$
\bigg\langle  \bigg( \frac{\varpi(\tau_1) - 1}{\sigma^2_+(\varpi(\tau_1))}\bigg)^2 \bigg( \frac{\varpi(\tau_2) - 1}{\sigma^2_+(\varpi(\tau_2))}\bigg)^2 1_{\{ \min\{\varpi(\tau_1),\varpi(\tau_2)\} >r \}} \bigg\rangle
$$
$$
\leq \frac{1}{r^4} \langle (\varpi(\tau_1) - 1)^2 1_{\{ 1/2>\varpi(\tau_1)>r \}} (\varpi(\tau_2) - 1)^2 1_{\{ 1/2>\varpi(\tau_2)>r \}} \rangle
$$
$$
+ \frac{1}{(r/2)^2} \langle (\varpi(\tau_1) - 1)^2 1_{\{ \varpi(\tau_1)\geq 1/2 \}} (\varpi(\tau_2) - 1)^2 1_{\{ 1/2>\varpi(\tau_2)>r \}} \rangle
$$
$$
+ \frac{1}{(r/2)^2} \langle (\varpi(\tau_1) - 1)^2 1_{\{ 1/2>\varpi(\tau_2)>r \}} (\varpi(\tau_2) - 1)^2 1_{\{ \varpi(\tau_2)\geq 1/2 \}} \rangle
$$
\begin{equation}\label{e:s24t}
  + \frac{1}{(1/2)^4} \langle (\varpi(\tau_1) - 1)^2 1_{\{ \varpi(\tau_2)\geq 1/2 \}} (\varpi(\tau_2) - 1)^2 1_{\{ \varpi(\tau_2)\geq 1/2 \}} \rangle
\end{equation}
By \eqref{e:varpivarpi}, the fourth term in \eqref{e:s24t} is bounded by
\begin{equation}\label{e:s24t4t}
  O\Big[\Big(\frac{\tau}{N}\Big)^2\Big].
\end{equation}
By the Cauchy-Schwarz inequality, (\ref{e:pvarpi<r}) and (\ref{e:varpivarpi}), the first term in the sum (\ref{e:s24t}) is bounded by
$$
\frac{1}{r^4} \sqrt{\langle (\varpi(\tau_1)-1)^4(\varpi(\tau_1)-1)^4\rangle }\sqrt{\langle 1_{\{ 1/2>\varpi(\tau_1)>r \}}1_{\{ 1/2>\varpi(\tau_2)>r \}}\rangle }
$$
$$
\leq \frac{1}{r^4} O\Big(\frac{\tau}{N}\Big) \sqrt{\bbP(1/2>\varpi(\tau_1)>r)\bbP(1/2>\varpi(\tau_2)>r)}
$$
\begin{equation}\label{e:s24t1t}
  \leq \frac{1}{r^4} O\Big(\frac{\tau}{N}\Big) \exp\Big\{ - C \Big(\frac{N}{\tau}\Big)^{1-\xi}\Big\} = o\Big(\frac{\tau}{N}\Big).
\end{equation}
Again by the Cauchy-Schwarz inequality, (\ref{e:pvarpi<r}) and (\ref{e:varpivarpi}), the second term in the sum (\ref{e:s24t}) is bounded by
$$
\frac{4}{r^2} \sqrt{\langle (\varpi(\tau_1)-1)^4(\varpi(\tau_1)-1)^4\rangle }\sqrt{\langle 1_{\{ \varpi(\tau_1)\geq 1/2 \}}1_{\{ 1/2>\varpi(\tau_2)>r \}}\rangle }
$$
$$
\leq \frac{4}{r^2} O\Big(\frac{\tau}{N}\Big) \sqrt{\bbP(\varpi(\tau_1)\geq 1/2)\bbP(1/2>\varpi(\tau_2)>r)}
$$
\begin{equation}\label{e:s24t2t}
\leq \frac{4}{r^2} O\Big(\frac{\tau}{N}\Big) \exp\Big\{ - C \Big(\frac{N}{\tau}\Big)^{1-\xi}\Big\} = o\Big(\frac{\tau}{N}\Big).
\end{equation}
An analogous bound holds for the third term in the sum (\ref{e:s24t}). Therefore, $|S_2|$ is bounded by the right-hand side of \eqref{e:max(|S1|,|S2|,|S3|)=bound}. To tackle $S_3$, rewrite it as
$$
-\langle (\varpi(\tau_1) - 1)(\varpi(\tau_2) - 1)( 1_{\{ \varpi(\tau_1)>r\}} 1_{\{ \varpi(\tau_2)\leq r\}}
$$
\begin{equation}\label{e:s3}
  + 1_{\{ \varpi(\tau_1)\leq r\}} 1_{\{ \varpi(\tau_2)> r\}} + 1_{\{ \varpi(\tau_1)\leq r\}} 1_{\{ \varpi(\tau_2)\leq r\}})\rangle .
\end{equation}
By the Cauchy-Schwarz inequality, (\ref{e:pvarpi<r}) and (\ref{e:varpivarpi}), the first term on the right-hand side of (\ref{e:s3}) is bounded by
$$
\sqrt{\langle (\varpi(\tau_1) - 1)^2(\varpi(\tau_2) - 1)^2\rangle } \sqrt{\bbP(\varpi(\tau_2) \leq r)}
$$
$$
\leq O\Big(\frac{\tau}{N}\Big) \exp\Big\{ - C \Big(\frac{N}{\tau}\Big)^{1-\xi}\Big\}= o\Big(\frac{\tau}{N}\Big).
$$
Similar bounds hold for the remaining terms on the right-hand side of (\ref{e:s3}). Therefore, $|S_3|$ is also bounded by the right-hand side of \eqref{e:max(|S1|,|S2|,|S3|)=bound}. As for $S_1$, it can be reexpressed as
\begin{equation}\label{e:s1}
  \langle \log \varpi(\tau_1)\log \varpi(\tau_2)\Big( 1_{\{ \varpi(\tau_1)>r\}} 1_{\{ \varpi(\tau_2)\leq r\}}+ 1_{\{ \varpi(\tau_1)\leq r\}} 1_{\{ \varpi(\tau_2)> r\}} + 1_{\{ \varpi(\tau_1)\leq r\}} 1_{\{ \varpi(\tau_2)\leq r\}}\Big)\rangle .
\end{equation}
Note that, by Lemma \ref{lem:logvarpimoment}, $\langle \log^4 \varpi(\tau)\rangle $ is bounded. Then, by applying the Cauchy-Schwarz inequality twice, the first term on the right-hand side of (\ref{e:s1}) is bounded by
$$
\sqrt{\langle \log^2 \varpi(\tau_1)\log^2 \varpi(\tau_2)\rangle } \sqrt{\bbP(\varpi(\tau_2) \leq r)}
$$
$$
\leq \Big(\langle \log^4 \varpi(\tau_1)\rangle \langle \log^4 \varpi(\tau_2)\rangle \Big)^{1/4} \exp\Big\{ - C \Big(\frac{N}{\tau}\Big)^{1-\xi}\Big\}
= o\Big(\frac{\tau}{N}\Big).
$$
Similar bounds hold for the remaining terms on the right-hand side of (\ref{e:s1}). Therefore, $|S_1|$ is bounded by the right-hand side of \eqref{e:max(|S1|,|S2|,|S3|)=bound}. This shows \eqref{e:s1s2s3}.
\end{proof}

\begin{lemma}\label{lem:elogvarpi}
\begin{equation}\label{e:logvarpi}
  \langle \log \varpi(\tau) \rangle + \frac{1}{2} \langle  (\varpi(\tau) - 1)^2\rangle  = O\Big(\frac{\tau}{N}\Big).
\end{equation}
\end{lemma}
\begin{proof}
Fix $0 < r <1/2$. Let
$$
T_1 = \langle  \log \varpi(\tau) \rangle - \langle \log \varpi(\tau)1_{\{\varpi(\tau)>r\}}\rangle ,
$$
$$
T_2 = \langle \log\varpi(\tau)1_{\{\varpi(\tau)>r\}}\rangle  + \frac{1}{2} \langle (\varpi(\tau) - 1)^2 1_{\{\varpi(\tau)>r\}}\rangle ,
$$
$$
T_3 = \frac{1}{2} \langle  (\varpi(\tau) - 1)^2\rangle  - \frac{1}{2} \langle  (\varpi(\tau) - 1)^2 1_{\{\varpi(\tau)>r\}}\rangle .
$$
Recall that, by Lemma \ref{lem:logvarpimoment}, $\langle \log^2 \varpi(\tau)\rangle $ is bounded. Thus, by the Cauchy-Schwarz inequality and by Lemma \ref{lem:P(varpi<r)},
$$
T_1 = \langle \log\varpi(\tau)1_{\{\varpi(\tau)\leq r\}} \rangle  \leq \sqrt{\langle \log^2 \varpi(\tau)\rangle } \sqrt{\bbP(\varpi(\tau)\leq r)}
$$
\begin{equation}\label{e:appB_T1_bound}
  \leq \sqrt{\langle \log^2 \varpi(\tau)\rangle } \exp\Big\{ -C \Big(\frac{N}{\tau}\Big)^{1-\xi}\Big\} = o\Big(\frac{\tau}{N}\Big).
\end{equation}
By a similar reasoning, we can further prove that
\begin{equation}\label{e:appB_T3_bound}
T_3 = o\Big(\frac{\tau}{N}\Big).
\end{equation}
Now, we turn to $T_2$. By an almost sure Taylor expansion,
$$
\log \varpi(\tau)1_{\{\varpi(\tau)>r\}} = \bigg\{(\varpi(\tau) - 1) - \frac{1}{2} (\varpi(\tau) - 1)^2
+ \frac{1}{3}\bigg(\frac{\varpi - 1}{\sigma^2_{+}(\varpi)}\bigg)^3  \bigg\}1_{\{\varpi(\tau)>r\}},
$$
where $\sigma^2_+(\varpi(\tau)) \in [\min\{\varpi(\tau),1\}, \max\{\varpi(\tau),1\}]$. Then, $T_2$ is bounded by
\begin{equation}\label{e:logvarpis2}
  \abs{\langle (\varpi(\tau) - 1)1_{\{\varpi(\tau)>r\}} \rangle} + \frac{1}{3}\abs{\Big\langle\bigg(\frac{\varpi - 1}{\sigma^2_{+}(\varpi)}\bigg)^3 1_{\{\varpi(\tau)>r\}}\Big\rangle}.
\end{equation}
Since $\langle\varpi(\tau) - 1\rangle = 0$, by the Cauchy-Schwarz inequality and Lemmas \ref{lem:kappa>=2} and \ref{lem:P(varpi<r)}, the first term in \eqref{e:logvarpis2} can be bounded by
$$
\abs{\langle\varpi(\tau) - 1)1_{\{\varpi(\tau)>r\}}\rangle} = \abs{\langle (\varpi(\tau) - 1)1_{\{\varpi(\tau)>r\}}\rangle - \langle \varpi(\tau) - 1\rangle}
$$
$$
 = \abs{\langle (\varpi(\tau) - 1)1_{\{\varpi(\tau)\leq r\}}\rangle} \leq \sqrt{\langle (\varpi(\tau) - 1)^2\rangle} \sqrt{\bbP(\varpi(\tau)\leq r)} = o\Big(\frac{\tau}{N}\Big).
$$
Meanwhile, by the Cauchy-Schwarz inequality and Lemmas \ref{lem:P(varpi<r)} and \ref{lem:kappa1kappa2}, the second term in \eqref{e:logvarpis2} is bounded by
\begin{equation}\label{e:appB_T2t2_bound}
\frac{1}{3r^3} \abs{\langle (\varpi - 1)^3 1_{\{1/2> \varpi(\tau)>r\}}\rangle} + \frac{1}{3(1/2)^3} \abs{\langle (\varpi - 1)^3 1_{\{ \varpi(\tau)\geq 1/2\}}\rangle}
\leq O\Big(\frac{\tau}{N}\Big)
\end{equation}
Thus,
\begin{equation}\label{e:appB_T2_bound}
T_2 = O\Big(\frac{\tau}{N}\Big).
\end{equation}
Relations \eqref{e:appB_T1_bound}, \eqref{e:appB_T3_bound} and \eqref{e:appB_T2_bound} imply \eqref{e:logvarpi}.
\end{proof}

\section{Bias and variance of ${\boldsymbol E}_{\ols}$ and the asymptotic distribution of ${\boldsymbol E}$}\label{s:bias_var_asympt_dist_stand_estimator}

We are now in a position to prove Theorems \ref{t:logmubias} and \ref{t:tildesigmak1k2} and Proposition \ref{p:asympt_dist_estim}, which give, respectively, asymptotically valid characterizations of the bias and variance involved in $\MSD$-based estimation, and the asymptotic distribution of the standardized estimator \eqref{e:xi_to_zeta_estimvar}.\\

The proof of Theorem \ref{t:logmubias} is a consequence of a Taylor expansion, followed by using estimates of the decay of $\MSD$ moments. Constructing the latter requires using a concentration inequality (e.g., \cite{ledoux:2005,boucheron:lugosi:massart:2013}), which was done in Section \ref{sec:app2sec1}.

\begin{theorem}\label{t:logmubias}
For $0< \alpha < 3/2$, under the assumptions of Theorem \ref{t:asympt_dist_MSD}, \eqref{e:bias_log} holds.
\end{theorem}
\begin{proof}
The left-hand side of \eqref{e:bias_log} can be rewritten as
\begin{equation}\label{e:biaslogmu22}
  \Big\langle\log \frac{M_N(\tau)}{\langle X^2(\tau)\rangle}\Big\rangle + \log \frac{\langle X^2(\tau)\rangle}{\sigma^2 \tau^{\alpha}}
  = \langle \log \varpi(\tau) \rangle + \log \frac{\langle X^2(\tau)\rangle}{\sigma^2 \tau^{\alpha}} .
\end{equation}
By Proposition 1 in \cite{didier:zhang:2017}, we can rewrite the second sum term on the right-hand side of \eqref{e:biaslogmu22} as
$$
\log (1 + O(\tau^{-\delta})) = O(\tau^{-\delta}), \quad N \rightarrow \infty.
$$
By Lemmas \ref{lem:Evarpi-1varpi-1} and \ref{lem:elogvarpi}, we can recast the first sum term on the right-hand side of \eqref{e:biaslogmu22} as
$$
- \frac{1}{2} \langle(\varpi(\tau) - 1)^2\rangle + O\Big(\frac{\tau}{N}\Big)
$$
$$
= -\frac{1}{4 N} \sum_{i=-N+1}^{N-1} \bigg(1- \frac{\abs{i}}{N}\bigg)\bigg\{ \abs{ \frac{i}{\tau} + 1}^{\alpha}
  -2 \abs{ \frac{i}{\tau} }^{\alpha} + \abs{ \frac{i}{\tau} - 1}^{\alpha} \bigg\}^2
  + O(\tau^{-\delta}) + O\Big(\frac{\tau}{N}\Big).
$$
Thus, \eqref{e:bias_log} follows.
\end{proof}

Next, the proof of Theorem \ref{t:tildesigmak1k2} relies on Taylor expansions of the moments of the logarithm of the $\MSD$.
\begin{theorem}\label{t:tildesigmak1k2}
For $0< \alpha < 3/2$, under the assumptions of Theorem \ref{t:asympt_dist_MSD}, expression \eqref{e:tildesigmak1k2e1} holds.
\end{theorem}
\begin{proof}
For $k_1,k_2 = 1,\hdots,m$, rewrite
$$
\upsilon_{k_1,k_2} = \cov(\log M_N(\tau_{k_1}), \log M_N(\tau_{k_2}))
$$
$$
=\langle [\log M_N(\tau_{k_1}) - \langle\log M_N(\tau_{k_1})\rangle][\log M_N(\tau_{k_1}) - \langle\log M_N(\tau_{k_1})\rangle]\rangle
$$
$$
= \langle [\log \varpi(\tau_{k_1}) - \langle\log \varpi(\tau_{k_1})\rangle][\log \varpi(\tau_{k_2}) - \langle\log \varpi(\tau_{k_2})\rangle]\rangle
$$
\begin{equation}\label{e:tildesigmaleme1}
  = \langle \log \varpi(\tau_{k_1})\log \varpi(\tau_{k_2})\rangle - \langle\log \varpi(\tau_{k_1}) \rangle\langle\log \varpi(\tau_{k_2})\rangle.
\end{equation}
By Lemmas \ref{lem:kappa>=2} and \ref{lem:elogvarpi},
$$
\langle\log \varpi(\tau_{k_1})\rangle = O\Big(\frac{\tau}{N}\Big).
$$
Therefore, \eqref{e:tildesigmaleme1} can be reexpressed as
\begin{equation}\label{e:tildesigmaleme2}
  \langle \log \varpi(\tau_{k_1})\log \varpi(\tau_{k_2})\rangle + o\Big(\frac{\tau}{N}\Big).
\end{equation}
By Lemmas \ref{lem:Evarpi-1varpi-1}, \ref{lem:s1s2s3} and \ref{l:bias_var_converge} (expression \eqref{e:varsigma(alpha-hat,h,h)_asympt_equiv}), expression \eqref{e:tildesigmak1k2e1} holds.
\end{proof}

The proof of Proposition \ref{p:asympt_dist_estim} builds upon Taylor expansions and characterizing the asymptotic behavior of the standardization term in the definition of the estimator \eqref{e:xi_to_zeta_estimvar}.
\begin{proposition}\label{p:asympt_dist_estim}
Under the assumptions of Theorem \ref{t:asympt_dist_MSD}, suppose $0 < \alpha < 3/2$. Then, the estimator \eqref{e:xi_to_zeta_estimvar} satisfies
\begin{equation}\label{e:asympt_dist_estim}
\Lambda^{-1/2}( A_{\ols})({\boldsymbol E}  - {\boldsymbol \xi})\stackrel{d}\rightarrow {\mathcal N}(0,I), \quad N \rightarrow \infty,
\end{equation}
where the vector ${\boldsymbol \xi}$ is given by \eqref{e:EX2(t)=Dt(alpha)}. In particular, the estimator is consistent, i.e,
$$
{\boldsymbol E}  \stackrel{P}\rightarrow {\boldsymbol \xi}.
$$
\end{proposition}
\begin{proof}
Recast the estimator \eqref{e:xi_to_zeta_estimvar} as
\begin{equation}\label{e:standardized_zetahat}
{\boldsymbol Z} = (X^T \Upsilon^{-1}( A_{\ols}) X)^{1/2}(X^T \Upsilon^{-1}( A_{\ols}) X)^{-1}X^T \Upsilon^{-1}( A_{\ols}) \textbf{y}.
\end{equation}
Rewrite $\Upsilon(\alpha) = \Upsilon(\alpha,N)$ as to express the dependence of the latter matrix on $N$. Define
\begin{equation}\label{e:S(alphahat-stand,n)}
\Big(\frac{N}{\tau} \hspace{1mm}\Upsilon( A_{\ols},N) \Big)^{-1} =: S( A_{\ols},N) = \Big(s_{k_1,k_2}( A_{\ols},N)\Big)_{k_1,k_2 = 1,\hdots,m},
\end{equation}
$$
s_N( A_{\ols}) = \sum^{m}_{k_1 = 1}\sum^{m}_{k_2 = 1}s_{k_1,k_2}( A_{\ols},N) \in \bbR.
$$
By \eqref{e:hetero_tildesigma}, \eqref{e:tildesigmak1k2e1} and \eqref{e:varsigma(alpha-hat,h,h)_asympt_equiv}, we can write
\begin{equation}\label{e:S(xi,n)->S(xi)}
S( A_{\ols},N) \stackrel{P}\rightarrow S(\alpha)= \Big(s_{k_1,k_2}(\alpha)\Big)_{k_1,k_2 = 1,\hdots,m}, \quad N \rightarrow \infty,
\end{equation}
and
\begin{equation}\label{e:s(alpha)}
s(\alpha):=\sum^{m}_{k_1=1}\sum^{m}_{k_2=1}s_{k_1,k_2}(\alpha) \in \bbR
\end{equation}
for some constant matrix $S(\alpha)$. For notational simplicity, write $s_N = s_N( A_{\ols},N)$ and $s_{k_1,k_2}(N) = s_{k_1,k_2}( A_{\ols},N)$. Then,
$$
X^T S({\boldsymbol \xi})X=
\left(\begin{array}{cc}
s_N   & \sum^{m}_{k_1=1}\sum^{m}_{k_2=1}(\log \tau_{k_1})s_{k_1,k_2}(N)\\
\sum^{m}_{k_1=1}\sum^{m}_{k_2=1}(\log \tau_{k_1})s_{k_1,k_2}(N) & \sum^{m}_{k_1=1}\sum^{m}_{k_2=1}(\log \tau_{k_1}\log \tau_{k_2})s_{k_1,k_2}(N)
\end{array}\right).
$$
By a simple calculation and relation \eqref{e:S(xi,n)->S(xi)},
$$
c_w(N) := \det (X^T S( A_{\ols},N) X)
$$
$$
= s_N \sum^{m}_{k_1=1}\sum^{m}_{k_2=1} \log w_{k_1}\log w_{k_2} s_{k_1,k_2}(N) - \Big( \sum^{m}_{k_1=1}\sum^{m}_{k_2=1} \log w_{k_1}  s_{k_1,k_2}(N)\Big)^2
$$
$$
\stackrel{P}\rightarrow s(\alpha) \sum^{m}_{k_1=1}\sum^{m}_{k_2=1} \log w_{k_1}\log w_{k_2} s_{k_1,k_2}(\alpha) - \Big( \sum^{m}_{k_1=1}\sum^{m}_{k_2=1} \log w_{k_1}  s_{k_1,k_2}(\alpha)\Big)^2 = c_w(\alpha).
$$
Moreover, by \eqref{e:S(alphahat-stand,n)},
\begin{equation}\label{e:GLS-adjusted_term_asympt}
(X^T \Upsilon^{-1}( A_{\ols},N) X)^{-1}X^T \Upsilon^{-1}( A_{\ols},N)= (X^T S( A_{\ols},N) X)^{-1}X^T S( A_{\ols},N)
\end{equation}
$$
= \frac{1}{c_w(N)}\left(\begin{array}{c}
\log \tau \Big\{\sum^{m}_{k_1=1}\sum^{m}_{k_2=1} \log w_{k_1}  s_{k_1,k_2}(N) \Big( \sum^{m}_{k=1}s_{k,j}(N)\Big) - s_N \sum^{m}_{k=1}\log w_k s_{k,j}(N)\Big\}\\
+ \Big(\sum^{m}_{k_1=1}\sum^{m}_{k_2=1} \log w_{k_1}\log w_{k_2} s_{k_1, k_2}(N)\sum^{m}_{k=1}s_{k,j}(N) \\
- \sum^{m}_{k_1=1}\sum^{m}_{k_2=1} \log w_{k_1} s_{k_1, k_2}(N)\sum^{m}_{k=1}\log w_{k} s_{k,j}(N) \Big)\\
\\
s_N \sum^{m}_{k=1}\log w_k s_{k,j}(N) - \sum^{m}_{k_1=1}\sum^{m}_{k_2=1} \log w_{k_1}  s_{k_1,k_2}(N) \Big( \sum^{m}_{k=1}  s_{k,j}(N)\Big)
\end{array}\right)_{j=1,\hdots,m}
$$
\begin{equation}\label{e:(XTSigma^(-1)X)^(-1)XTSigma^(-1)ksi}
=: \left(\begin{array}{c}
\log \tau \hspace{1mm}a_{n,j} + b_{n,j}\\
-a_{n,j}
\end{array}\right)_{j=1,\hdots,m},
\end{equation}
where the sequences of constants $\{a_{n,j}\}_{N \in \bbN}$ and $\{b_{n,j}\}_{N \in \bbN}$, converge to constants $a_{j}$ and $b_{j}$, respectively, for $j = 1,\hdots,m$.

Recall that, for a symmetric positive definite matrix
$$
M = \left(\begin{array}{cc}
m_{11} & m_{12}\\
m_{12} & m_{22}
\end{array}\right),
$$
we can write its square root in closed form as
$$
M^{1/2} = \frac{1}{\sqrt{\textnormal{tr}(M) + 2\sqrt{\det(M)}}}\left(\begin{array}{cc}
m_{11} + \sqrt{\det(M)} & m_{12}\\
m_{12} & m_{22} + \sqrt{\det(M)}
\end{array}\right).
$$
Therefore,
$$
(X^T S( A_{\ols})X)^{1/2}= \frac{1}{\sqrt{\textnormal{tr}(X^T S( A_{\ols})X) + 2\sqrt{c_w(N)}}}
$$
\begin{equation}\label{e:M^1/2}
\left(\begin{array}{cc}
s_N  + \sqrt{c_w(N)} & \sum^{m}_{k_1=1}\sum^{m}_{k_2=1}(\log \tau_{k_1})s_{k_1,k_2}(N)\\
\sum^{m}_{k_1=1}\sum^{m}_{k_2=1}(\log \tau_{k_1})s_{k_1,k_2}(N) & \sum^{m}_{k_1=1}\sum^{m}_{k_2=1}(\log \tau_{k_1}\log \tau_{k_2})s_{k_1,k_2}(N) + \sqrt{c_w(N)}
\end{array}\right).
\end{equation}
Note that
\begin{equation}\label{e:tr(XTS(alpha-hat)X)}
\textnormal{tr}(X^T S( A_{\ols})X) + 2\sqrt{c_w(N)} \sim \log^2 \tau \hspace{1mm}s_N.
\end{equation}
By expressions \eqref{e:S(xi,n)->S(xi)}, \eqref{e:s(alpha)}, \eqref{e:(XTSigma^(-1)X)^(-1)XTSigma^(-1)ksi}, \eqref{e:M^1/2} and \eqref{e:tr(XTS(alpha-hat)X)},
$$
(X^T \Upsilon^{-1}( A_{\ols},N) X)^{1/2}(X^T \Upsilon^{-1}( A_{\ols},N) X)^{-1}X^T \Upsilon^{-1}( A_{\ols},N)
$$
$$
= \sqrt{\frac{N}{\tau}}(X^T S( A_{\ols},N) X)^{1/2}(X^T S( A_{\ols},N) X)^{-1}X^T S( A_{\ols},N)
$$
$$
= \sqrt{\frac{N}{\tau}} \frac{1}{\sqrt{\textnormal{tr}(X^T S( A_{\ols},N) X)+ 2 \sqrt{c_{w}(N)}}} \frac{1}{c_{w}(N)}
$$
$$
\left(\begin{array}{c}
a_{n,j}\Big[(\log \tau) (s_N + \sqrt{c_{w}(N)}) - \sum^{m}_{k_1=1}\sum^{m}_{k_2=1}(\log \tau_{k_1})s_{k_1 k_2}(N)\Big] + b_{n,j}(s_N + \sqrt{c_{w}(N)})\\
\\
a_{n,j}\Big[(\log \tau)\sum^{m}_{k_1=1}\sum^{m}_{k_2=1}(\log \tau_{k_1})s_{k_1 k_2}(N) \\ - \sum^{m}_{k_1=1}\sum^{m}_{k_2=1}(\log \tau_{k_1}\log \tau_{k_2})s_{k_1 k_2}(N)
 - \sqrt{c_{w}(N)} \Big]\\
 + b_{n,j}\Big[\sum^{m}_{k_1=1}\sum^{m}_{k_2=1}(\log \tau_{k_1})s_{k_1 k_2}(N) \Big]
\end{array}\right)_{j=1,\hdots,m}
$$
$$
= \sqrt{\frac{N}{\tau}} \frac{1}{\sqrt{\textnormal{tr}(X^T S( A_{\ols},N) X)+ 2 \sqrt{c_{w}(N)}}} \frac{1}{c_{w}(N)}
$$
$$
\left(\begin{array}{c}
a_{n,j}\Big[(\log \tau) \sqrt{c_{w}(N)} - \sum^{m}_{k_1=1}\sum^{m}_{k_2=1}(\log w_{k_1})s_{k_1k_2}(N)\Big] + b_{n,j}(s_N + \sqrt{c_{w}(N)})  \\
\\
a_{n,j}\Big[-(\log \tau)\sum^{m}_{k_1=1}\sum^{m}_{k_2=1}(\log w_{k_1})s_{k_1k_2}(N) - \sum^{m}_{k_1=1}\sum^{m}_{k_2=1}(\log w_{k_1}\log w_{k_2})s_{k_1 k_2}(N)\\
 - \sqrt{c_{w}(N)} \Big]
 + b_{n,j}\Big[(\log \tau)s_N + \sum^{m}_{k_1=1}\sum^{m}_{k_2=1}(\log w_{k_1})s_{k_1 k_2}(N)\Big]
\end{array}\right)_{j=1,\hdots,m}
$$
$$
\stackrel{P}\sim  \sqrt{\frac{N}{\tau}} \frac{1}{s(\alpha)}\frac{1}{c_w(\alpha)}
\left(\begin{array}{c}
a_{j} \sqrt{c_w(\alpha)} \\
\\
b_{j} s(\alpha)  - a_{j}\Big[\sum^{m}_{k_1=1}\sum^{m}_{k_2=1}(\log w_{k_1})s_{k_1 k_2}(\alpha) \Big]
\end{array}\right)_{j=1,\hdots,m}
$$
\begin{equation}\label{e:sqrt(n/h)Psi}
=: \sqrt{\frac{N}{\tau}}\hspace{1mm}\Psi \in \bbR^{2 \times m},
\end{equation}
%
%
%$$
%\sim c_w(\alpha) \left(\begin{array}{c}
%\log \tau \Big\{\sum^{m}_{k_1=1}\sum^{m}_{k_2=1} \log w_{k_1}  s_{k_1,k_2}(\alpha) \Big( \sum^{m}_{k=1}s_{k,j}(\alpha)\Big) - s(\alpha) \sum^{m}_{k=1}\log w_k s_{k,j}(\alpha)\Big\}\\
%s(\alpha)  \sum^{m}_{k=1}\log w_k s_{k,m}(\alpha) - \sum^{m}_{k_1=1}\sum^{m}_{k_2=1} \log w_{k_1}  s_{k_1,k_2}(\alpha) \Big( \sum^{m}_{k=1}s_{k,j}(\alpha)\Big)
%\end{array}\right)_{j=1,\hdots,m}
%$$
as $N \rightarrow \infty$. For ${\mathbf y}$ and $X$ as in \eqref{e:y,X}, rewrite the left-hand side of expression \eqref{e:asympt_dist_estim} as
$$
(X^T \Upsilon^{-1}( A_{\ols},N) X)^{1/2}({\boldsymbol E}-{\boldsymbol \xi})
$$
\begin{equation}\label{e:hatbeta-beta}
=(X^T \Upsilon^{-1}( A_{\ols},N) X)^{1/2}(X^T \Upsilon^{-1}( A_{\ols},N) X)^{-1}X^T \Upsilon^{-1}( A_{\ols},N) ({\mathbf y}-X{\boldsymbol \xi}).
\end{equation}
Recast
$$
{\mathbf y} = \Big( \log[ \hspace{1mm}M_N(\tau_k)e^{\frac{\tau_k}{N}\beta_N( A_{\ols},\tau_k)} \hspace{1mm}] \Big)_{k=1,\hdots,m}.
$$
By entrywise first order Taylor expansions,
$$
{\mathbf y}- X {\boldsymbol \xi}=\bigg( \log\Big(\frac{M_N(\tau_k)e^{\frac{\tau_k}{N}\beta_N( A_{\ols},\tau_k)}}{\sigma^2 \tau_k^{\alpha}}\Big) \bigg)_{k=1,\hdots,m}
$$
\begin{equation}\label{e:Qn-Mnbeta}
= \bigg( \frac{M_N(\tau_k)e^{\frac{\tau}{N}\beta_N( A_{\ols},\tau_k)}}{\sigma^2 \tau_k^{\alpha}} -1 \bigg)_{k=1,\hdots,m}
+\Big( O\bigg( \frac{M_N(\tau_k)e^{\frac{\tau}{N}\beta_N( A_{\ols},\tau_k)}}{\sigma^2 \tau_k^{\alpha}} -1 \bigg)^2\Big)_{k=1,\hdots,m}
\end{equation}
However, for $k = 1,\hdots,m$, the first term on the right-hand side of \eqref{e:Qn-Mnbeta} can be reexpressed as
\begin{equation}\label{e:norm_MSD_with_and_without_exp}
\frac{M_N(\tau_k)}{\sigma^2 \tau_k^{\alpha}} \Big( e^{\frac{\tau_k}{N}\beta_N( A_{\ols},\tau_k)} - 1\Big)
+ \frac{M_N(\tau_k)}{\sigma^2 \tau_k^{\alpha}}-1.
\end{equation}
Again by a first order Taylor expansion,
\begin{equation}\label{e:exp_Taylor}
e^{\frac{\tau_k}{N}\beta_N( A_{\ols},\tau_k)} - 1  = \frac{\tau_k}{N}\beta_N( A_{\ols},\tau_k) + o_P \Big( \frac{\tau_k}{N} \Big).
\end{equation}
Therefore, by \eqref{e:sqrt(n/h)Psi}, \eqref{e:norm_MSD_with_and_without_exp} and \eqref{e:exp_Taylor}, we can rewrite relation \eqref{e:hatbeta-beta} as
\begin{equation}\label{e:sqrt(n/h)(MSD/theta*hk^alpha-1)+o(1)}
(\Psi + o_P(1))\hspace{1mm}\sqrt{\frac{N}{\tau}}\Big( \frac{M_N(\tau_k)}{\sigma^2 \tau_k^{\alpha}}-1 \Big)_{k=1,\hdots,m} + o_P \Big( \sqrt{\frac{\tau}{N}} \Big).
\end{equation}
Expression \eqref{e:asympt_dist_estim} is a consequence of \eqref{e:sqrt(n/h)(MSD/theta*hk^alpha-1)+o(1)} and Theorem \ref{t:asympt_dist_MSD}, where the estimator \eqref{e:standardized_zetahat} is asymptotically standardized.
\end{proof}

The following lemma establishes the convergence of the main bias and variance factors and is used in the proofs of Proposition \ref{p:asympt_dist_estim} and Theorem \ref{t:tildesigmak1k2}.
\begin{lemma}\label{l:bias_var_converge}
For $0 < \alpha < 3/2$, consider the main bias and variance factors \eqref{e:beta(alpha,h,n)} and \eqref{e:varsigma-n(alpha,h1,h2)}, respectively, under the assumptions of Theorem \ref{t:asympt_dist_MSD}. Then, there are functions $\beta(\alpha,\cdot) > 0$ and $\varsigma(\alpha,\cdot,\cdot) > 0$ such that
\begin{equation}\label{e:beta(alpha,h,n)_conv}
\beta_N( A_{\ols},\tau_k) \stackrel{P}\rightarrow \Big(\beta(\alpha,\tau_k) \Big)_{k=1,\hdots,m},
\end{equation}
\begin{equation}\label{e:varsigma(alpha-hat,h,h)_asympt_equiv}
\Big( \varsigma_N( A_{\ols},\tau_{k_1},\tau_{k_2})\Big)_{k_1,k_2=1,\hdots,m} \stackrel{P}\rightarrow \Big(\varsigma(\alpha,k_1,k_2)\Big)_{k_1,k_2=1,\hdots,m},
\end{equation}
as $N \rightarrow \infty$.
\end{lemma}
\begin{proof}
From expression \eqref{e:tildesigmak1k2e1} for the variance term $\upsilon_{k_1,k_2}({\boldsymbol \xi})$, recast
  $$
  \varsigma_N(\alpha,\tau_{k_1},\tau_{k_2}) = \frac{1}{2 } \Big\{ \sum_{|i| \leq \tau} + \sum_{\tau+1 \leq |i| \leq N-1} \Big\} \bigg( 1- \frac{\abs{i}}{N}\bigg)\bigg\{ \abs{ \frac{i}{\tau\sqrt{w_{k_1} w_{k_2}}} + \sqrt{\frac{w_{k_1}}{w_{k_2}}} }^{\alpha}
 $$
\begin{equation} \label{e:varsigma_n(alpha,h,h)}
-\abs{ \frac{i}{\tau \sqrt{w_{k_1} w_{k_2}}} + \sqrt{\frac{w_{k_1}}{w_{k_2}}} - \sqrt{\frac{w_{k_2}}{w_{k_1}}}}^{\alpha}-\abs{ \frac{i}{\tau \sqrt{w_{k_1} w_{k_2}}} }^{\alpha}  + \abs{ \frac{i}{\tau \sqrt{w_{k_1} w_{k_2}}} - \sqrt{\frac{w_{k_2}}{w_{k_1}}}}^{\alpha} \bigg\}^2
\frac{1}{\tau } .
\end{equation}
The second sum term on the right-hand side of \eqref{e:varsigma_n(alpha,h,h)} is bounded by
$$
\frac{1}{2} \Big( 1 - \frac{\tau}{N}\Big) \sum_{\tau+1 \leq |i| \leq N-1} \bigg\{ \abs{ \frac{i}{\tau \sqrt{w_{k_1} w_{k_2}}} + \sqrt{\frac{w_{k_1}}{w_{k_2}}} }^{\alpha}
$$
$$
-\abs{ \frac{i}{\tau \sqrt{w_{k_1} w_{k_2}}} + \sqrt{\frac{w_{k_1}}{w_{k_2}}} - \sqrt{\frac{w_{k_2}}{w_{k_1}}}}^{\alpha}-\abs{ \frac{i}{\tau \sqrt{w_{k_1} w_{k_2}}} }^{\alpha}  + \abs{ \frac{i}{\tau \sqrt{w_{k_1} w_{k_2}}} - \sqrt{\frac{w_{k_2}}{w_{k_1}}}}^{\alpha} \bigg\}^2
\frac{1}{\tau } .
$$
$$
\sim \frac{1}{2} \Big( 1 - \frac{\tau}{N}\Big) \int_{\tau+1 \leq |y|\leq n} \bigg\{ \abs{ \frac{y}{\sqrt{w_{k_1} w_{k_2}}} + \sqrt{\frac{w_{k_1}}{w_{k_2}}} }^{\alpha}
$$
\begin{equation}\label{e:finite_sample_expansion_resid}
-\abs{ \frac{y}{\sqrt{w_{k_1} w_{k_2}}} + \sqrt{\frac{w_{k_1}}{w_{k_2}}} - \sqrt{\frac{w_{k_2}}{w_{k_1}}}}^{\alpha}-\abs{ \frac{y}{\sqrt{w_{k_1} w_{k_2}}} }^{\alpha}  + \abs{ \frac{y}{\sqrt{w_{k_1} w_{k_2}}} - \sqrt{\frac{w_{k_2}}{w_{k_1}}}}^{\alpha} \bigg\}^2 dy \rightarrow 0,
\end{equation}
as $N \rightarrow \infty$. On the other hand, by condition \eqref{e:h(n)} and the dominated convergence theorem, the first sum term on the right-hand side of \eqref{e:varsigma_n(alpha,h,h)} converges to
$$
\varsigma(\alpha,k_1,k_2) := \frac{1}{2} \int^{1}_{-1} \bigg\{ \abs{ \frac{y}{\sqrt{w_{k_1} w_{k_2}}} + \sqrt{\frac{w_{k_1}}{w_{k_2}}} }^{\alpha}
-\abs{ \frac{y}{\sqrt{w_{k_1} w_{k_2}}} + \sqrt{\frac{w_{k_1}}{w_{k_2}}} - \sqrt{\frac{w_{k_2}}{w_{k_1}}}}^{\alpha}
$$
\begin{equation}\label{e:summ_asympt_equiv}
-\abs{ \frac{y}{\sqrt{w_{k_1} w_{k_2}}} }^{\alpha}  + \abs{ \frac{y}{\sqrt{w_{k_1} w_{k_2}}} - \sqrt{\frac{w_{k_2}}{w_{k_1}}}}^{\alpha} \bigg\}^2 dy > 0.
\end{equation}
Moreover, by Corollary \ref{c:asympt_dist_MSD}, $ A_{\ols} \stackrel{P}\rightarrow \alpha$. So, pick a small enough $\epsilon_0$ such that $\alpha \in (\epsilon_0, 3/2-\epsilon_0)$. Let $A = \{\omega:  A_{\ols}(\omega) \in (\alpha - \frac{\epsilon_0}{2}, \alpha + \frac{\epsilon_0}{2}) \}$. In the set $A$, by a simple adaptation of the argument leading to \eqref{e:finite_sample_expansion_resid} and the convergence to \eqref{e:summ_asympt_equiv},
\begin{equation}\label{e:varsigma(alpha-hat,h,h)_asympt_equiv_proof}
\varsigma_N( A_{\ols},\tau_{k_1},\tau_{k_2}) \rightarrow \varsigma(\alpha,k_1,k_2), \quad N \rightarrow \infty,
\end{equation}
where $\bbP(A) \rightarrow 1$. This shows \eqref{e:varsigma(alpha-hat,h,h)_asympt_equiv}. A similar argument can be used to show \eqref{e:beta(alpha,h,n)_conv}. %In turn, consider the bias correction factor \eqref{e:beta(alpha,h,n)}. By the same argument leading to \eqref{e:varsigma(alpha-hat,h,h)_asympt_equiv},
%\begin{equation}\label{e:beta(alpha,h,n)_conv_proof}
%\beta_N( A_{\ols},h) \rightarrow \frac{1}{4}\int_{\bbR} \{ | y+1 |^{\alpha} - 2 |y|^{\alpha} + |y-1|^{\alpha}\}^2 dy, \quad  N \rightarrow \infty,
%\end{equation}
%where $\bbP(A) \rightarrow 1$. Expressions \eqref{e:varsigma(alpha-hat,h,h)_asympt_equiv} and \eqref{e:beta(alpha,h,n)_conv} are a consequence of \eqref{e:varsigma(alpha-hat,h,h)_asympt_equiv_proof} and \eqref{e:beta(alpha,h,n)_conv_proof}.
\end{proof}

\section{Pseudocode for generating the improved $\MSD$-based estimator ${\boldsymbol E}$}\label{s:pseudocode}

{\small
\begin{center}
\begin{tabular}{|l|}
\hline
\multicolumn{1}{|c|}{\textbf{Generating the improved pathwise estimator ${\boldsymbol E}$} (see \eqref{e:hetero_hatxi})}\\
\\
\hline
%\multicolumn{1}{|c|}{\textbf{Pseudo-code for estimating $\alpha$ and $\sigma^2$}}\\
%\hline
\textbf{Input}: \\
$\bullet$ one observed particle path $\{X_1,X_2,\hdots,X_N\}_{N \in \bbN}$ of length $N$;\\
$\bullet$ regression lag values $\tau_k$, $k = 1,\hdots,m$ (typically, $\tau_k= \tau w_k$, $w_1 < \hdots < w_m$, $\tau \ll N$);\\
$\bullet$ the expression for the asymptotic covariance matrix $\Upsilon(\alpha)$ as a function of $\alpha$;\\
\\
\textbf{Step 1}: obtain the standard estimator $ A_{\ols}$ over the chosen lag values;\\
\\
\textbf{Step 2}: estimate the asymptotic covariance matrix $\Upsilon({\boldsymbol \xi})$ by means of $\Upsilon( A_{\ols})$ (see \eqref{e:Sigmatilde(alpha-hat)});\\
\\
\textbf{Step 3}: use $ A_{\ols}$ and the estimator \eqref{e:beta(alphahat,h,n)} of the bias vector to produce the bias-corrected \\
regression system \eqref{e:regression_bias-corrected};\\
\\
\textbf{Step 4}: obtain the estimator ${\boldsymbol E}$ by means of $\Upsilon( A_{\ols})$-based GLS on the bias-corrected\\
regression system \eqref{e:regression_bias-corrected}.\\
\hline
\end{tabular}
\end{center}
}
\bibliographystyle{plain}
\bibliography{MSD_asympt_Gaussian_si_physics}

\begin{thebibliography}{10}

\bibitem{andreanov:grebenkov:2012}
A.~Andreanov and D.~Grebenkov.
\newblock Time-averaged {MSD} of {B}rownian motion.
\newblock {\em Journal of Statistical Mechanics: Theory and Experiment},
  2012(07):P07001, 2012.

\bibitem{barkai:garini:metzler:2012}
E.~Barkai, Y.~Garini, and R.~Metzler.
\newblock Strange kinetics of single molecules in living cells.
\newblock {\em Physics Today}, 65(8):29--35, 2012.

\bibitem{berglund:2010}
A.~J. Berglund.
\newblock Statistics of camera-based single-particle tracking.
\newblock {\em Physical Review E}, 82(1):011917, 2010.

\bibitem{bertseva:grebenkov:schmidhauser:gribkova:jeney:forro:2012}
E.~Bertseva, D.~S. Grebenkov, P.~Schmidhauser, S.~Gribkova, S.~Jeney, and
  L.~Forr{\'o}.
\newblock Optical trapping microrheology in cultured human cells.
\newblock {\em European Physical Journal E}, 35(7):63, 2012.

\bibitem{boucheron:lugosi:massart:2013}
S.~Boucheron, G.~Lugosi, and P.~Massart.
\newblock {\em Concentration inequalities: a nonasymptotic theory of
  independence}.
\newblock Oxford University Press, 2013.

\bibitem{boyer:dean:mejia:oshanin:2012}
D.~Boyer, D.~S. Dean, C.~Mej\'{i}a-Monasterio, and G.~Oshanin.
\newblock Optimal estimates of the diffusion coefficient of a single {B}rownian
  trajectory.
\newblock {\em Physical Review E}, 85(3):031136, 2012.

\bibitem{boyer:dean:mejia:oshanin:2013}
D.~Boyer, D.~S. Dean, C.~Mej\'{i}a-Monasterio, and G.~Oshanin.
\newblock Distribution of the least-squares estimators of a single {B}rownian
  trajectory diffusion coefficient.
\newblock {\em Journal of Statistical Mechanics: Theory and Experiment},
  2013(04):P04017, 2013.

\bibitem{briane:kervrann:vimond:2017}
V.~Briane, C.~Kervrann, and M.~Vimond.
\newblock A statistical analysis of particle trajectories in living cells.
\newblock {\em \texttt{arXiv:1707.01838}}, pages 1--38, 2017.

\bibitem{burnecki:2012}
K.~Burnecki.
\newblock {FARIMA} processes with application to biophysical data.
\newblock {\em Journal of Statistical Mechanics: Theory and Experiment},
  2012(05):P05015, 2012.

\bibitem{burnecki:kepten:garini:sikora:weron:2015}
K.~Burnecki, E.~Kepten, Y.~Garini, G.~Sikora, and A.~Weron.
\newblock Estimating the anomalous diffusion exponent for single particle
  tracking data with measurement errors-an alternative approach.
\newblock {\em Scientific Reports}, 5(11306):1--11, 2015.

\bibitem{burnecki:kepten:janczura:bronshtein:garini:weron:2012}
K.~Burnecki, E.~Kepten, J.~Janczura, I.~Bronshtein, Y.~Garini, and A.~Weron.
\newblock Universal algorithm for identification of fractional {B}rownian
  motion. a case of telomere subdiffusion.
\newblock {\em Biophysical Journal}, 103(9):1839--1847, 2012.

\bibitem{burnecki:muszkieta:sikora:weron:2012}
K.~Burnecki, M.~Muszkieta, G.~Sikora, and A.~Weron.
\newblock Statistical modelling of subdiffusive dynamics in the cytoplasm of
  living cells: a {FARIMA} approach.
\newblock {\em Europhysics Letters}, 98(1):10004, 2012.

\bibitem{burov:jeon:metzler:barkai:2011}
S.~Burov, J.-H. Jeon, R.~Metzler, and E.~Barkai.
\newblock Single particle tracking in systems showing anomalous diffusion: the
  role of weak ergodicity breaking.
\newblock {\em Physical Chemistry Chemical Physics}, 13(5):1800--1812, 2011.

\bibitem{cheridito:kawaguchi:maejima:2003}
P.~Cheridito, H.~Kawaguchi, and M.~Maejima.
\newblock Fractional {O}rnstein-{U}hlenbeck processes.
\newblock {\em Electronic Journal of Probability}, 8(3):1--14, 2003.

\bibitem{christensen:2011}
R.~Christensen.
\newblock {\em Plane answers to complex questions: the theory of linear
  models}.
\newblock Springer Science \& Business Media, 2011.

\bibitem{dawson:wirtz:hanes:2003}
M.~Dawson, D.~Wirtz, and J.~Hanes.
\newblock Enhanced viscoelasticity of human cystic fibrotic sputum correlates
  with increasing microheterogeneity in particle transport.
\newblock {\em Journal of Biological Chemistry}, 278(50):50393--50401, 2003.

\bibitem{deng:barkai:2009}
W.~Deng and E.~Barkai.
\newblock Ergodic properties of fractional {B}rownian-{L}angevin motion.
\newblock {\em Physical Review E}, 79(1):011112, 2009.

\bibitem{didier:fricks:2014}
G.~Didier and J.~Fricks.
\newblock On the wavelet-based simulation of anomalous diffusion.
\newblock {\em Journal of Statistical Computation and Simulation},
  84(4):697--723, 2014.

\bibitem{didier:mckinley:hill:fricks:2012}
G.~Didier, S.~A. McKinley, D.~B. Hill, and J.~Fricks.
\newblock Statistical challenges in microrheology.
\newblock {\em Journal of Time Series Analysis}, 33(55):724--743, September
  2012.

\bibitem{didier:zhang:2017}
G.~Didier and K.~Zhang.
\newblock The asymptotic distribution of the pathwise mean squared displacement
  in single particle tracking experiments.
\newblock {\em Journal of Time Series Analysis}, 38(3):395--416, May 2017.

\bibitem{dobrushin:major:1979}
R.~Dobrushin and P.~Major.
\newblock Non-central limit theorems for non-linear functional of {G}aussian
  fields.
\newblock {\em Probability Theory and Related Fields}, 50(1):27--52, 1979.

\bibitem{gajda:wylomanska:kantz:chechkin:sikora:2018}
J.~Gajda, A.~Wy{\l}oma{\'n}ska, H.~Kantz, A.V. Chechkin, and G.~Sikora.
\newblock Large deviations of time-averaged statistics for {G}aussian
  processes.
\newblock {\em Statistics \& Probability Letters}, 2018.

\bibitem{ghosh:cherstvy:grebenkov:metzler:2016}
S.~K. Ghosh, A.~G. Cherstvy, D.~S. Grebenkov, and R.~Metzler.
\newblock Anomalous, non-{G}aussian tracer diffusion in crowded two-dimensional
  environments.
\newblock {\em New Journal of Physics}, 18(1):013027, 2016.

\bibitem{giraitis:surgailis:1985}
L.~Giraitis and D.~Surgailis.
\newblock Clt and other limit theorems for functionals of {G}aussian processes.
\newblock {\em Zeitschrift f{\"u}r Wahrscheinlichkeitstheorie und verwandte
  Gebiete}, 70(2):191, 1985.

\bibitem{giraitis:surgailis:1990}
L~Giraitis and D~Surgailis.
\newblock A central limit theorem for quadratic forms in strongly dependent
  linear variables and its application to asymptotical normality of {W}hittle's
  estimate.
\newblock {\em Probability Theory and Related Fields}, 86(1):87--104, 1990.

\bibitem{gorman:greene:2008}
J.~Gorman and E.~C. Greene.
\newblock Visualizing one-dimensional diffusion of proteins along {DNA}.
\newblock {\em Nat. Struct. Mol. Biol.}, 15(8):768--774, 2008.

\bibitem{grebenkov:2011prob}
D.~Grebenkov.
\newblock Probability distribution of the time-averaged mean-square
  displacement of a {G}aussian process.
\newblock {\em Physical Review E}, 84(3):031124, 2011.

\bibitem{grebenkov:2011functionals}
D.~Grebenkov.
\newblock Time-averaged quadratic functionals of a {G}aussian process.
\newblock {\em Physical Review E}, 83(6):061117, 2011.

\bibitem{grebenkov:2013}
D.~Grebenkov.
\newblock Optimal and suboptimal quadratic forms for noncentered {G}aussian
  processes.
\newblock {\em Physical Review E}, 88(3):032140, 2013.

\bibitem{grebenkov:vahabi:bertseva:forro:jeney:2013}
D.~S. Grebenkov, M.~Vahabi, E.~Bertseva, L.~Forr{\'o}, and S.~Jeney.
\newblock Hydrodynamic and subdiffusive motion of tracers in a viscoelastic
  medium.
\newblock {\em Physical Review E}, 88(4):040701, 2013.

\bibitem{guyon:leon:1989}
Xavier Guyon and Jos{\'e} Le{\'o}n.
\newblock Convergence en loi des {H}-variations d'un processus {G}aussien
  stationnaire sur $\textbf{R}$.
\newblock {\em Annales de l'IHP: Probabilit{\'e}s et Statistiques},
  25(3):265--282, 1989.

\bibitem{halford:marko:2004}
S.~E. Halford and J.~F. Marko.
\newblock How do site-specific {DNA}-binding proteins find their targets?
\newblock {\em Nucleic Acids Res.}, 32(10):3040--3052, 2004.

\bibitem{helenius:brouhard:kalaidzidis:diez:howard:2006}
J.~Helenius, G.~Brouhard, Y.~Kalaidzidis, S.~Diez, and J.~Howard.
\newblock The depolymerizing kinesin {MCAK} uses lattice diffusion to rapidly
  target microtubule ends.
\newblock {\em Nature}, 441(7089):115--119, 2006.

\bibitem{hess:girirajan:mason:2006}
S.~T. Hess, T.P.K. Girirajan, and M.~D. Mason.
\newblock Ultra-high resolution imaging by fluorescence photoactivation
  localization microscopy.
\newblock {\em Biophysical {J}ournal}, 91(11):4258--4272, 2006.

\bibitem{hill:vasquez:mellnik:mckinley:vose:mu:henderson:donaldson:alexis:boucher:forest:2014}
D.~B. Hill, P.~A. Vasquez, J.~Mellnik, S.~A. McKinley, A.~Vose, F.~Mu, A.~G.
  Henderson, S.~H. Donaldson, N.~E. Alexis, R.~C. Boucher, and M.~G. Forest.
\newblock A biophysical basis for mucus solids concentration as a candidate
  biomarker for airways disease.
\newblock {\em PloS one}, 9(2):e87681, 2014.

\bibitem{hoze:hochman:2017}
N.~Hoz\'{e} and D.~Hochman.
\newblock Statistical methods for large ensembles of super-resolution
  stochastic single particle trajectories in cell biology.
\newblock {\em Annual Review of Statistics and its Application}, 4:189--223,
  2017.

\bibitem{isserlis1916}
L.~Isserlis.
\newblock On certain probable errors and correlation coefficients of multiple
  frequency distributions with skew regression.
\newblock {\em Biometrika}, 11:185–190, 1916.

\bibitem{jeon:barkai:metzler:2013}
J.-H. Jeon, E.~Barkai, and R.~Metzler.
\newblock Noisy continuous time random walks.
\newblock {\em Journal of Chemical Physics}, 139(12):121916, 2013.

\bibitem{jeon:metzler:2010}
J.-H. Jeon and R.~Metzler.
\newblock Analysis of short subdiffusive time series: scatter of the
  time-averaged mean-squared displacement.
\newblock {\em Journal of Physics A: Mathematical and Theoretical},
  43(25):252001, 2010.

\bibitem{kepten:bronshtein:garini:2013}
E.~Kepten, I.~Bronshtein, and Y.~Garini.
\newblock Improved estimation of anomalous diffusion exponents in
  single-particle tracking experiments.
\newblock {\em Physical Review E}, 87(5):052713, 2013.

\bibitem{meerschaert:nane:xiao:2009}
E.~Kepten, A.~Weron, G.~Sikora, K.~Burnecki, and Y.~Garini.
\newblock Correlated continuous time random walks.
\newblock {\em Statistics and Probability Letters}, 79:1194--1202, 2009.

\bibitem{kepten:weron:sikora:burnecki:garini:2015}
E.~Kepten, A.~Weron, G.~Sikora, K.~Burnecki, and Y.~Garini.
\newblock Guidelines for the fitting of anomalous diffusion mean square
  displacement graphs from single particle tracking experiments.
\newblock {\em PLoS One}, 10(2):e0117722, 2015.

\bibitem{kou:2008}
S.~C. Kou.
\newblock Stochastic modeling in nanoscale biophysics: subdiffusion within
  proteins.
\newblock {\em Annals of Applied Statistics}, 2(2):501--535, 2008.

\bibitem{kou:xie:2004}
S.~C. Kou and X.~S. Xie.
\newblock Generalized {L}angevin equation with fractional {G}aussian noise:
  subdiffusion within a single protein molecule.
\newblock {\em Physical Review Letters}, 93(18):180603, 2004.

\bibitem{lai:wang:cone:wirtz:hanes:2009}
S.K. Lai, Y.Y. Wang, R.~Cone, D.~Wirtz, and J.~Hanes.
\newblock Altering mucus rheology to solidify human mucus at the nanoscale.
\newblock {\em PLoS One}, 4(1):e4294, 2009.

\bibitem{lasne:etal:2006}
D.~Lasne, G.~A. Blab, S.~Berciaud, M.~Heine, L.~Groc, D.~Choquet, L.~Cognet,
  and B.~Lounis.
\newblock Single nanoparticle photothermal tracking ({SN}a{PT}) of 5-nm gold
  beads in live cells.
\newblock {\em Biophysical Journal}, 91(12):4598--4604, 2006.

\bibitem{laurent:massart:2000}
B.~Laurent and P.~Massart.
\newblock Adaptive estimation of a quadratic functional by model selection.
\newblock {\em Annals of Statistics}, pages 1302--1338, 2000.

\bibitem{ledoux:2005}
M.~Ledoux.
\newblock {\em The concentration of measure phenomenon}.
\newblock Number~89. American Mathematical Society, 2005.

\bibitem{levine2000one}
A.J. Levine and TC~Lubensky.
\newblock One-and two-particle microrheology.
\newblock {\em Physical Review Letters}, 85(8):1774--1777, 2000.

\bibitem{lieleg:vladescu:ribbeck:2010}
O.~Lieleg, I.~Vladescu, and K.~Ribbeck.
\newblock Characterization of particle translocation through mucin hydrogels.
\newblock {\em Biophysical {J}ournal}, 98(9):1782, 2010.

\bibitem{lysy:pillai:hill:forest:mellnik:vasquez:mckinley:2016}
M.~Lysy, N.~Pillai, D.~B. Hill, M.~G. Forest, J.~Mellnik, P.~Vasquez, and S.~A.
  McKinley.
\newblock Model comparison for single particle tracking in biological fluids.
\newblock {\em To appear in Journal of the American Statistical Association},
  pages 1--44, 2016.

\bibitem{major:1981}
P.~Major.
\newblock Limit theorems for non-linear functionals of {G}aussian sequences.
\newblock {\em Zeitschrift f{\"u}r Wahrscheinlichkeitstheorie und verwandte
  Gebiete}, 57(1):129--158, 1981.

\bibitem{mason:weitz:1995}
T.G. Mason and D.A. Weitz.
\newblock Optical measurements of the linear viscoelastic moduli of complex
  fluids.
\newblock {\em Physical Review Letters}, 74:1250--1253, 1995.

\bibitem{matsui:wagner:hill:etal:2006}
H.~Matsui, V.E. Wagner, D.B. Hill, U.E. Schwab, T.D. Rogers, B.~Button, R.M.
  Taylor, R.~Superfine, M.~Rubinstein, B.H. Iglewski, and R.C. Boucher.
\newblock A physical linkage between cystic fibrosis airway surface dehydration
  and {P}seudomonas aeruginosa biofilms.
\newblock {\em Proceedings of the National Academy of Sciences}, 103(48):18131,
  2006.

\bibitem{meerschaert:scheffler:2004}
M.~Meerschaert and H.-P. Scheffler.
\newblock Limit theorems for continuous-time random walks with infinite mean
  waiting times.
\newblock {\em Journal of Applied Probability}, 41:623--638, 2004.

\bibitem{mellnik:lysy:vasquez:pillai:hill:cribb:mckinley:forest:2016}
J.~W.~R. Mellnik, M.~Lysy, P.~A. Vasquez, N.~S. Pillai, D.~B. Hill, J.~Cribb,
  S.~A. McKinley, and M.~G. Forest.
\newblock Maximum likelihood estimation for single particle, passive
  microrheology data with drift.
\newblock {\em Journal of Rheology}, 60(3):379--392, 2016.

\bibitem{meroz:sokolov:2015}
Y.~Meroz and I.~M. Sokolov.
\newblock A toolbox for determining subdiffusive mechanisms.
\newblock {\em Physics Reports}, 573:1--29, 2015.

\bibitem{metzler:jeon:cherstvy:2016}
R.~Metzler, J.-H. Jeon, and A.G. Cherstvy.
\newblock Non-{B}rownian diffusion in lipid membranes: experiments and
  simulations.
\newblock {\em Biochimica et Biophysica Acta}, 1858(10):2451--2467, 2016.

\bibitem{metzler:tejedor:jeon:he:deng:burov:barkai:2009}
R.~Metzler, V.~Tejedor, J.H. Jeon, Y.~He, W.H. Deng, S.~Burov, and E.~Barkai.
\newblock Analysis of single particle trajectories: from normal to anomalous
  diffusion.
\newblock {\em Acta Physica Polonica B}, 40(5):1315--1331, 2009.

\bibitem{michalet:berglund:2012}
X.~Michalet and A.~J. Berglund.
\newblock Optimal diffusion coefficient estimation in single-particle tracking.
\newblock {\em Physical Review E}, 85(6):061916, 2012.

\bibitem{minoura:katayama:eisaku:sekimoto:muto:2010}
I.~Minoura, E.~Katayama, K.~Sekimoto, and E.~Muto.
\newblock One-dimensional {B}rownian motion of charged nanoparticles along
  microtubules: a model system for weak binding interactions.
\newblock {\em Biophysical Journal}, 98(8):1589--1597, 2010.

\bibitem{moulines:roueff:taqqu:2007fractals}
E.~Moulines, F.~Roueff, and M.~S. Taqqu.
\newblock Central limit theorem for the log-regression wavelet estimation of
  the memory parameter in the {G}aussian semi-parametric context.
\newblock {\em Fractals}, 15(04):301--313, 2007.

\bibitem{moulines:roueff:taqqu:2007:spectral}
E.~Moulines, F.~Roueff, and M.S. Taqqu.
\newblock On the spectral density of the wavelet coefficients of long-memory
  time series with application to the log-regression estimation of the memory
  parameter.
\newblock {\em Journal of Time Series Analysis}, 28(2):155--187, 2007.

\bibitem{moulines:roueff:taqqu:2008}
E.~Moulines, F.~Roueff, and M.S. Taqqu.
\newblock A wavelet {W}hittle estimator of the memory parameter of a
  nonstationary {G}aussian time series.
\newblock {\em Annals of Statistics}, pages 1925--1956, 2008.

\bibitem{nandi:heinrich:lindner:2012}
A.~Nandi, D.~Heinrich, and B.~Lindner.
\newblock Distributions of diffusion measures from a local mean-square
  displacement analysis.
\newblock {\em Physical Review E}, 86(2):021926, 2012.

\bibitem{nguyen:mckinley:2017}
H.~D. Nguyen and S.~A. McKinley.
\newblock Anomalous diffusion and the generalized {L}angevin equation.
\newblock {\em \texttt{https://arxiv.org/abs/1711.00560}}, pages 1--40, 2017.

\bibitem{nishimura:etal:2006}
S.~Y. Nishimura, S.~J. Lord, L.~O. Klein, K.~A. Willets, M.~He, Z.~Lu, R.~J.
  Twieg, and W.~E. Moerner.
\newblock Diffusion of lipid-like single-molecule fluorophores in the cell
  membrane.
\newblock {\em J. Phys. Chem. B}, 110(15):8151--8157, 2006.

\bibitem{ottobre:pavliotis:2011}
M.~Ottobre and G.~Pavliotis.
\newblock Asymptotic analysis for the generalized {L}angevin equation.
\newblock {\em Nonlinearity}, 24(5):1629, 2011.

\bibitem{pipiras:taqqu:2017}
V.~Pipiras and M.~S. Taqqu.
\newblock {\em Long-Range Dependence and Self-Similarity}.
\newblock Cambridge Series on Statistical and Probabilistic Mathematics.
  Cambridge University Press, Cambridge, United Kingdom, 2017.

\bibitem{prakasarao:2010}
B.~L.~S. Prakasa~Rao.
\newblock {\em Statistical Inference for Fractional Diffusion Processes}.
\newblock Wiley Series in Probability and Statistics, 2010.

\bibitem{qian:sheetzL:elson:1991}
H.\ Qian, M.\ Sheetz, and E.~Elson.
\newblock Single particle tracking. {A}nalysis of diffusion and flow in
  two-dimensional systems.
\newblock {\em Biophysical Journal}, 60(4):910--921, 1991.

\bibitem{reighard:ehre:rushton:ahonen:hill:schoenfisch:2017}
K.~P. Reighard, C.~Ehre, Z.~L. Rushton, M.~J.~R. Ahonen, D.~B. Hill, and M.~H.
  Schoenfisch.
\newblock Role of nitric oxide-releasing chitosan oligosaccharides on mucus
  viscoelasticity.
\newblock {\em ACS Biomaterials Science \& Engineering}, 3(6):1017--1026, 2017.

\bibitem{reighard:hill:dixon:worley:schoenfisch:2015}
K.~P. Reighard, D.~B. Hill, G.~A. Dixon, B.~V. Worley, and M.~H. Schoenfisch.
\newblock Disruption and eradication of {P}.\ aeruginosa biofilms using nitric
  oxide-releasing chitosan oligosaccharides.
\newblock {\em Biofouling}, 31(9-10):775--787, 2015.

\bibitem{rosenblatt:1961}
M.~Rosenblatt.
\newblock Independence and dependence.
\newblock In {\em Proceedings of the $4^{\textnormal{th}}$ {B}erkeley symposium
  on mathematical statistics and probability}, volume~2, pages 431--443, 1961.

\bibitem{sandev:metzler:tomovksi:2012}
T.~Sandev, R.~Metzler, and {\v{Z}}.~Tomovski.
\newblock Velocity and displacement correlation functions for fractional
  generalized {L}angevin equations.
\newblock {\em Fractional Calculus and Applied Analysis}, 15(3):426--450, 2012.

\bibitem{saxton1994anomalous}
M.J. Saxton.
\newblock Anomalous diffusion due to obstacles: a monte carlo study.
\newblock {\em Biophysical Journal}, 66(2):394--401, 1994.

\bibitem{saxton1996anomalous}
M.J. Saxton.
\newblock Anomalous diffusion due to binding: a monte carlo study.
\newblock {\em Biophysical Journal}, 70(3):1250--1262, 1996.

\bibitem{shiryaev:2000}
A.N. Shiryaev.
\newblock {\em Probability Theory}.
\newblock Springer-Verlag, New York, 2000.

\bibitem{sikora:teuerle:wylomanska:grebenkov:2017}
G.~Sikora, M.~Teuerle, A.~Wy{\l}oma{\'n}ska, and D.~Grebenkov.
\newblock Statistical properties of the anomalous scaling exponent estimator
  based on time-averaged mean-square displacement.
\newblock {\em Physical Review E}, 96(2):022132, 2017.

\bibitem{smith:karatekin:etal:2011}
M.~B. Smith, E.~Karatekin, A.~Gohlke, H.~Mizuno, N.~Watanabe, and D.~Vavylonis.
\newblock Interactive, computer-assisted tracking of speckle trajectories in
  fluorescence microscopy: application to actin polymerization and membrane
  fusion.
\newblock {\em Biophysical Journal}, 101(7):1794--1804, 2011.

\bibitem{sokolov:2008}
I.~M. Sokolov.
\newblock Statistics and the single molecule.
\newblock {\em Physics}, 1:8, 2008.

\bibitem{sonesson:elofsson:callisen:brismar:2007}
A.~W. Sonesson, U.~M. Elofsson, T.~H. Callisen, and H.~Brismar.
\newblock Tracking single lipase molecules on a trimyristin substrate surface
  using quantum dots.
\newblock {\em Langmuir}, 23(16):8352--8356, 2007.

\bibitem{suh:dawson:hanes:2005}
J.~Suh, M.~Dawson, and J.~Hanes.
\newblock Real-time multiple-particle tracking: applications to drug and gene
  delivery.
\newblock {\em Advanced Drug Delivery Reviews}, 57:63--78, 2005.

\bibitem{tafvizi:mirny:oijen:2011}
A.~Tafvizi, L.~A. Mirny, and A.~M. van Oijen.
\newblock Dancing on {DNA}: kinetic aspects of search processes on {DNA}.
\newblock {\em Chem.\ Phys.\ Chem.}, 12(8):1481--1489, 2011.

\bibitem{taqqu:1975}
M.~S. Taqqu.
\newblock Weak convergence to fractional {B}rownian motion and to the
  {R}osenblatt process.
\newblock {\em Probability Theory and Related Fields}, 31(4):287--302, 1975.

\bibitem{taqqu:1979}
M.~S. Taqqu.
\newblock Convergence of integrated processes of arbitrary {H}ermite rank.
\newblock {\em Probability Theory and Related Fields}, 50(1):53--83, 1979.

\bibitem{taqqu:2003}
M.~S. Taqqu.
\newblock Fractional {B}rownian motion and long range dependence.
\newblock In {\em Theory and Applications of Long-Range Dependence (P. Doukhan,
  G. Oppenheim and M. S. Taqqu, eds.)}, pages 5--38. Birkh\"{a}user, Boston,
  2003.

\bibitem{taqqu:2011}
M.~S. Taqqu.
\newblock The {R}osenblatt process.
\newblock In {\em The selected works of {M}urray {R}osenblatt (Davis, R. A. and
  Lii, K.-S. and Politis, D. N., eds.)}, pages 29--45. Springer, 2011.

\bibitem{vale:soll:gibbons:1989}
R.~D. Vale, D.~R. Soll, and I.~R. Gibbons.
\newblock One-dimensional diffusion of microtubules bound to flagellar dynein.
\newblock {\em Cell}, 59(5):915--925, 1989.

\bibitem{valentine:kaplan:thota:crocker:gisler:prudhomme:beck:weitz:2001}
M.~Valentine, P.~Kaplan, D.~Thota, J.~Crocker, T.~Gisler, R.~Prud抙omme,
  M.~Beck, and D.~A. Weitz.
\newblock Investigating the microenvironments of inhomogeneous soft materials
  with multiple particle tracking.
\newblock {\em Physical Review E}, 64(6):061506, 2001.

\bibitem{veillette:taqqu:2013}
M.~Veillette and M.~S. Taqqu.
\newblock Properties and numerical evaluation of the rosenblatt distribution.
\newblock {\em Bernoulli}, 19(3):982--1005, 2013.

\bibitem{veitch:abry:1999}
D.~Veitch and P.~Abry.
\newblock A wavelet-based joint estimator of the parameters of long-range
  dependence.
\newblock {\em IEEE Transactions on Information Theory}, 45(3):878--897, 1999.

\bibitem{vestergaard:blainey:flyvbjerg:2014}
C.~L. Vestergaard, P.~C. Blainey, and H.~Flyvbjerg.
\newblock Optimal estimation of diffusion coefficients from single-particle
  trajectories.
\newblock {\em Physical Review E}, 89(2):022726, 2014.

\bibitem{wendt:didier:combrexelle:abry:2017}
H.~Wendt, G.~Didier, S.~Combrexelle, and P.~Abry.
\newblock Multivariate {H}adamard self-similarity: testing fractal
  connectivity.
\newblock {\em Physica D: Nonlinear Phenomena}, 356--357:1--36, 2017.

\bibitem{wieser:schutz:2008}
S.~Wieser and G.~J. Sch{\"u}tz.
\newblock Tracking single molecules in the live cell plasma membrane -- do's
  and don't's.
\newblock {\em Methods}, 46(2):131--140, 2008.

\bibitem{zwanzig:2001}
R.~Zwanzig.
\newblock {\em Nonequilibrium Statistical Mechanics}.
\newblock Oxford University Press, 2001.

\end{thebibliography}

\end{document}